\documentclass[%
 reprint,
superscriptaddress,
 amsmath,amssymb,
 aps,
]{revtex4-2}

\usepackage[utf8]{inputenc}
\usepackage{graphicx}
\usepackage{algorithm}
\usepackage{algorithmic}
\usepackage{color}
\usepackage{physics}
\usepackage{braket}
\usepackage{comment}
\usepackage{mathrsfs}
\usepackage{enumerate}

\usepackage[version=4]{mhchem}
\usepackage[breaklinks=true]{hyperref}
\usepackage{lipsum}

\newcommand{\mr}[1]{\mathrm{#1}}

\newcommand{\mcl}[1]{\mathcal{#1}}
\newcommand{\bbC}{\mathbb{C}}
\newcommand{\bbR}{\mathbb{R}}

\newcommand{\bbN}{\mathbb{N}}
\newcommand{\bbE}{\mathbb{E}}

\newcommand{\ad}{\mathrm{ad}}

\newcommand{\supp}{\mathrm{supp}}

\newcommand{\poly}[1]{\mathrm{poly} \left( #1 \right)}
\newcommand{\polylog}[1]{\mathrm{polylog} \left( #1 \right)}
\newcommand{\Otilde}[1]{\tilde{\mathcal{O}} \left( #1 \right)}
\date{\today}
\usepackage{amsthm}
\theoremstyle{definition}
\newtheorem{theorem}{Theorem}[]
\newtheorem{definition}[theorem]{Definition}

\newtheorem{lemma}[theorem]{Lemma}
\newtheorem{corollary}[theorem]{Corollary}

\newtheorem*{theorem*}{Theorem}
\newtheorem*{lemma*}{Lemma}
\newtheorem*{corollary*}{Corollary}
\newtheorem*{proposition*}{Proposition}

\begin{document}
\title{Lattice Lindbladian simulation by patching and merging}

\author{Kaoru Mizuta}
\email{mizuta.kaoru.qiqb@osaka-u.ac.jp}
\affiliation{Center for Quantum Information and Quantum Biology, The University of Osaka, 1-2 Machikaneyama, Toyonaka, Osaka 560-0043, Japan}
\affiliation{Department of Applied Physics, Graduate School of Engineering, The University of Tokyo, Hongo 7-3-1, Bunkyo, Tokyo 113-8656, Japan}
\affiliation{RIKEN Center for Quantum Computing (RQC), Hirosawa 2-1, Wako, Saitama 351-0198, Japan}

\begin{abstract}

Simulating the dynamics of dissipative quantum many-body systems governed by local Lindbladians is a fundamental task in quantum computation. 
While the near-optimal gate count has been achieved for Hamiltonian simulation, comparable results for Lindbladian simulation remain elusive. 
In this work, we develop quantum algorithms for lattice Lindbladians by exploiting their locality.
First, we consider sparsely dissipative systems, in which the dissipation is sparsely located, including boundary-driven systems.
We establish a near-optimal quantum algorithm for their dynamics with gate count $\mathcal{O}(Nt \, \mathrm{polylog}(Nt/\varepsilon))$, where $N$ is the system size, $t$ is the evolution time, and $\varepsilon$ is the allowable error.
We then consider generic lattice Lindbladians with finite-range interactions and dissipation, and develop an algorithm for simulating time-evolved observables with gate count $\mathcal{O}((Nt)^{4/3} \, \mathrm{polylog}(Nt/\varepsilon))$ per sample with the sampling complexity $\Theta(\varepsilon^{-2})$.
The gate count has the smallest known dependence on the system size among algorithms retaining polylogarithmic dependence on $1/\varepsilon$.
Our algorithms are based on two techniques that exploit locality: patching and merging. 
Patching decomposes dissipative dynamics into dynamics on subsystems with exponentially small error, generalizing a key idea underlying the Haah-Hastings-Kothari-Low algorithm for near-optimal Hamiltonian simulation. 
Merging absorbs reversed dissipative dynamics into other parts of the evolution, substantially reducing the overhead associated with quasi-probabilistic sampling.
These results demonstrate that locality can be fully exploited to achieve fast quantum simulation of dissipative many-body dynamics, opening the way to applications such as predicting nonequilibrium phenomena and preparing desirable quantum states.

\end{abstract}
\maketitle

\section{Introduction}

Simulating quantum many-body dynamics is one of the central problems in quantum physics and quantum chemistry, for which quantum computers are expected to offer an exponential speedup over classical computers.
A fundamental goal in this field is to develop optimal quantum algorithms whose cost scales as favorably as possible with the system size $N$, evolution time $t$, and allowable error $\varepsilon$.
Over the past several decades, various quantum algorithms have been established for the simulation of Hamiltonian dynamics, i.e., Hamiltonian simulation.
Prominent examples include product formulas (PFs, also known as Trotterization) \cite{Lloyd1996-ko,childs-prl2019-pf,childs2021-trotter} and post-Trotter methods, such as linear combinations of unitaries (LCU) \cite{Berry-prl2015-LCU} and the quantum singular value transformation (QSVT) \cite{Low2017-qsp,Low2019-qubitization,Gilyen2019-qsvt}.
Some advanced algorithms such as multi-product formulas (MPF) \cite{low2019-mpf,aftab2024-mpf,Mizuta_2026_mpf} and the Haah-Hastings-Kothari-Low (HHKL) algorithm \cite{Haah2021-hhkl} have achieved gate counts nearly matching the theoretical lower bound for local lattice Hamiltonians.

In parallel, open quantum many-body systems described by the Gorini-Kossakowski-Sudarshan-Lindblad (GKSL) master equation \cite{Gorini1976,Lindblad1976},
\begin{equation}\label{Eq_Intro:GKSL}
    \frac{d}{dt} \rho (t) = \mathcal{L} \rho (t),
\end{equation}
\begin{equation}\label{Eq_Intro:Lindbladian}
    \mathcal{L}\rho = - i [H,\rho] + \sum_{m=1}^M \left( 2 L_m \rho L_m^\dagger - \{ L_m^\dagger L_m, \rho \}\right),
\end{equation}
have attracted significant interest, where the superoperator $\mathcal{L}$ is called a Lindbladian.
Simulating such dissipative dynamics (Lindbladian simulation) is crucial for modeling realistic noisy quantum processors, exploring nonequilibrium phenomena under dissipation \cite{Fazio_2025-open-review}, and preparing ground states or thermal states as steady states \cite{Ding2024SingleAncilla,Zhan2026Rapid,Li2025Dissipative,Chen2025Efficient,Chen2023Noncommutative,Ding2026Simple,Lin-2025-open-review}.
Consequently, developing fast and efficient quantum algorithms for Lindbladian simulation is of paramount importance \cite{Kliesch-prl2011-open,cleve-2016-open,ChildsLi2017Sparse,Li-Wang-2022-open,Ding-Li-Lin-2024-open,kato-wada-2024_open,Yu-Yuan-2025-open,BorrasMarvian2025Lindblad,Peng-prxq-2025-open,Chen2025RandomizedLindblad,Pocrnic2025Repeated,David_2026_qdrift,Yu-Cirac-2025-open,mohammadipour2026-open,Wang-2026-open,wang2026-open-optimal,chen2026-open-optimal}.

In contrast, Lindbladian simulation presents distinct challenges absent in Hamiltonian simulation.
For generic local Hamiltonians with finite-range interactions, the gate count can reach the near-optimal scaling $\order{Nt \, \polylog{Nt/\varepsilon}}$ \cite{Haah2021-hhkl}.
In sharp contrast, for generic local Lindbladians with finite-range interactions and dissipation:
The PF-based approach yields the gate count $\order{(Nt)^{1+1/p}\varepsilon^{-1/p}}$ for the orders $p=1,2$ \cite{Kliesch-prl2011-open,Wang-2026-open}.
It fails to attain the higher-order scaling with $p \geq 3$ \cite{Sheng-1989-splitting,Suzuki1991-za}, which is available in Hamiltonian simulation.
The LCU-based approaches yield the gate count $\order{N^2t \, \polylog{Nt/\varepsilon}}$ for both Hamiltonian simulation \cite{Berry-prl2015-LCU} and Lindbladian simulation \cite{cleve-2016-open,Li-Wang-2022-open}.
Although these approaches achieve the near-optimal query complexity within the block-encoding framework, their gate counts have worse $N$-dependency.
To date, the near-optimal gate count $\order{Nt \, \polylog{Nt/\varepsilon}}$ has only been achieved for restricted models where the jump operators $\{L_m\}$ are mutually commuting and Hermitian assuming efficient access to a specified QRAM \cite{Yu-Cirac-2025-open}.
For generic lattice Lindbladians, extrapolation applied to the second-order PF has been the leading strategy for observable estimation, achieving the gate count $\order{(Nt)^{3/2} \, \mathrm{polylog}(Nt/\varepsilon)}$ \cite{Wang-2026-open}.
Whether one can construct quantum algorithms for generic dissipative systems that closely approach the fundamental lower bound $\Omega(Nt)$ remains a major open problem.

In this work, we make progress toward this goal by establishing efficient quantum algorithms for Lindbladian simulation, which fully exploit the locality.
First, we prove a dissipative counterpart of the patching lemma used in the HHKL algorithm \cite{Osborne-2006-patching,michalakis-2012-patching,Haah2021-hhkl}, though this algorithm itself is unavailable for Lindbladian dynamics due to the exponential overhead.
Using this framework, we construct two quantum algorithms by developing two techniques, \textit{patching} and \textit{merging}, based on the locality:

\begin{enumerate}
    \item \textbf{Sparsely-dissipative Lindbladians:}
    We consider lattice Lindbladians where dissipative domains have the size $o(\log N)$ and are separated by a distance of $\omega(\log N)$. 
    Boundary-driven systems are included as a significant class in nonequilibrium physics \cite{Landi-2022-boundary}.
    We develop a patching strategy that can decompose the evolution into blocks respecting complete positivity by adjusting the patch size.
    This algorithm achieves the near-optimal gate count $\order{Nt \, \mathrm{polylog}(Nt/\varepsilon)}$.
    
    \item \textbf{Generic lattice Lindbladians:} For general systems with finite-range interactions and dissipation, we introduce a \emph{merging} technique that recombines partitioned evolution operators, substantially suppressing the sampling overhead in observable estimation.
    Combined with optimal patch sizing, this yields the gate count $\order{(Nt)^{4/3} \, \mathrm{polylog}(Nt/\varepsilon)}$ in one dimension (the extension to higher dimensions is provided in Appendix \ref{SecA:High_dim}).
    As far as we know, it achieves the best $N$-dependence among the existing algorithms while retaining polylogarithmic error scaling.
\end{enumerate}

These results provide a partial resolution to the optimality problem of Lindbladian simulation and significantly narrow the gap toward the theoretical limit for generic dissipative many-body dynamics.
They will shed light on Lindbladian simulation and broader physics with various applications such as nonequilibrium phenomena and state preparation.

The remainder of this paper is organized as follows.
In Section \ref{Sec:Summary}, we formalize the problem setup and summarize our main theoretical results.
In Section \ref{Sec:Patching}, we establish the dissipative patching lemma.
Sections \ref{Sec:Sparse_algorithm} and \ref{Sec:Generic_algorithm} detail our quantum algorithms for sparsely-dissipative and globally-dissipative systems, respectively.
Finally, Section \ref{Sec:Conclusion} concludes with a discussion of future directions.

\section{Summary of results}\label{Sec:Summary}

\begin{table*}[t]
  \centering
  \label{tab:two_tables}
  
  \begin{minipage}{0.48\textwidth}
    \centering
    \begin{tabular}{ccc}
    \hline \hline
       & Gate count & Remark \\ \hline \\[-2ex]
     PF \cite{childs-prl2019-pf,childs2021-trotter}& \(\displaystyle Nt \left( \frac{Nt}{\varepsilon}\right)^{\frac1p}\) & $\quad p = 1,2,3,\cdots \quad$ \\ \\[-2ex]
      LCU \cite{Berry-prl2015-LCU} & \( \displaystyle N^2t \, \polylog{Nt/\varepsilon} \) \\ \\[-2ex]
      QSVT \cite{Low2019-qubitization,Gilyen2019-qsvt} & \( \displaystyle N(Nt+\log(1/\varepsilon))\) \\ \\[-2ex]
      MPF \cite{low2019-mpf,Mizuta_2026_mpf} & \( \displaystyle \quad N^{1+\frac1{p+1}}t \, \polylog{Nt/\varepsilon} \quad \) & $p = 1,2,3,\cdots$ \\ \\[-2ex]
      HHKL \cite{Haah2021-hhkl} & \( \displaystyle Nt \, \polylog{Nt/\varepsilon} \) &  \\ \\[-2ex] \hline
    \end{tabular} \\
    \vspace{10pt}
    (a) Hamiltonian simulation 
  \end{minipage}
  \hfill 
  \begin{minipage}{0.48\textwidth}
    \centering
    \begin{tabular}{ccc}
    \hline \hline
       & Gate count & Remark \\ \hline \\[-2ex]
     PF \cite{Kliesch-prl2011-open}& \(\displaystyle Nt \left( \frac{Nt}{\varepsilon}\right)^{\frac1p}\) & $\quad p = 1,2 \quad$ \\ \\[-2ex]
      LCU \cite{Li-Wang-2022-open} & \( \displaystyle N^2t 
      \, \polylog{Nt/\varepsilon}\) \\ \\[-2ex]
      \begin{tabular}{c} Extrapolation \\[-0.5ex] with PFs \cite{Wang-2026-open} \end{tabular}& \(\quad \displaystyle (Nt)^{1+\frac1p} \, \polylog{Nt/\varepsilon} \quad \) & \begin{tabular}{c}
        $p=1,2$ \\[-0.5ex]
        Observables
      \end{tabular}\\ \\[-2ex]
      \begin{tabular}{c} Algorithm \ref{Algorithm_Sparse} \\[-0.5ex] (Theorem \ref{Thm_Setup:sparse}) \end{tabular} & \( \displaystyle \quad Nt \, \polylog{Nt/\varepsilon} \quad \) & \begin{tabular}{c}
        Sparsely \\[-0.5ex] dissipative
      \end{tabular} \\ \\[-2ex]
      \begin{tabular}{c} Algorithm \ref{Algorithm_Generic} \\[-0.5ex] (Theorem \ref{Thm_Sparse:generic}) \end{tabular} & \( \displaystyle (Nt)^{\frac43} \, \polylog{Nt/\varepsilon} \) & Observables
      \\ \\[-2ex] \hline
    \end{tabular} \\ 
    \vspace{10pt}
    (b) Lindbladian simulation
  \end{minipage}
  \caption{The cost of simulating quantum dynamics under (a) lattice Hamiltonians and (b) lattice Lindbladians. The label ``Observables" means that the algorithm estimates time-evolved observables by sampling, for which the gate counts per sample. Algorithm \ref{Algorithm_Generic} has the sampling complexity $\Theta(\varepsilon^{-2})$. We achieve the near-optimal gate count for sparsely dissipative systems (See Definition \ref{Def_Setup:sparse_dissipation}) by Algorithm \ref{Algorithm_Sparse}, and also improve the gate count for generic lattice Lindbladians by Algorithm \ref{Algorithm_Generic}.}
  \label{Table:Gate_counts}
\end{table*}

In this section, we describe the setup and the summary of our quantum algorithms efficiently simulating Lindbladian dynamics.
We also provide a brief review of some existing algorithms for Hamiltonian or Lindbladian simulation to clarify the context of our results.

\subsection{Setup and problem}\label{Subsec:Setup}

We specify the setup and the problem throughout this paper here.
We begin by introducing some mathematical notation below.

\begin{itemize}
    \item Landau symbols: We use the Landau symbols, $\order{\cdot}$, $o(\cdot)$, $\Omega (\cdot)$, $\omega (\cdot)$, and $\Theta(\cdot)$.
    The symbol $\tilde{\cdot}$ denotes polylogarithmic corrections in $N$, $t$, and $1/\varepsilon$.
    For instance, $\Otilde{f(N,t,1/\varepsilon)}$ for some function $f$ means $\order{f(N,t,1/\varepsilon) \, \polylog{N,t,1/\varepsilon}}$.
    
    \item Lattice and domains: We consider an $N$-qubit lattice $\Lambda=\{1,2,\cdots,N\}$.
    For a domain $X \subset \Lambda$, the symbol $|X|$ means the number of sites in $X$.
    The range of a domain $X$, denoted by $r(X)$, is defined by
    \begin{equation}\label{Eq_Setup:domain_size_def}
        r(X) = \max_{i,j \in X} \mr{dist}(i,j),
    \end{equation}
    with some distance measure $\mr{dist}(i,j)$ ($i,j \in \Lambda$) on the lattice $\Lambda$.
    The distance between domains $X,Y \subset \Lambda$ is defined by
    \begin{equation}\label{Eq_Setup:domain_dist_def}
        \mr{dist}(X,Y) = \min_{i \in X, j \in Y} \mr{dist}(i,j).
    \end{equation}
    Throughout the main text, we will suppose that $\Lambda$ is one-dimensional and hence we have $\mr{dist}(i,j)=|i-j|$.

    \item Pauli matrices: 
    We denote the set of $N$-qubit Pauli matrices on the lattice $\Lambda$ by $\{P_\mu\}$.
    The symbol $\supp (P_\mu)$ ($\subset \Lambda$) means the support of $P_\mu$.

    \item Linear map and its norm: A generic linear map $\mcl{A}$ on an $N$-qubit state $\rho$ can be written as
    \begin{equation}\label{Eq_Setup:Map_Pauli_expansion}
        \mcl{A} \rho = \sum_{\mu,\nu} \gamma_{\mu\nu} P_\mu \rho P_\nu, \quad \gamma_{\mu\nu} \in \bbC. 
    \end{equation}
    We denote the diamond norm of $\mcl{A}$ by $\norm{A}_\Diamond$.
    We also define the Pauli norm of $\mcl{A}$ by
    \begin{equation}
        \norm{\mcl{A}}_\mr{Pauli} = \sum_{\mu,\nu} |\gamma_{\mu\nu}|.
    \end{equation}
    Note that these norms are related by the inequality,
    \begin{eqnarray}
        \norm{\mcl{A}}_\Diamond &\leq& \sum_{\mu,\nu} |\gamma_{\mu\nu}| \norm{P_\mu (\cdot) P_\nu}_\Diamond \nonumber \\
        &\leq& \norm{\mcl{A}}_\mr{Pauli}. \label{Eq_Setup:Norm_relation}
    \end{eqnarray}

    \item Product and commutator: For operators $A_1,A_2,\cdots,A_n$, we denote their products by
    \begin{eqnarray}
        \prod_{n'=1}^n A_{n'} &=& A_n \cdots A_2 A_1, \\
        \prod_{n'=n}^1 A_{n'} &=& A_1 A_2\cdots A_n.
    \end{eqnarray}
    Their commutator is denoted by
    \begin{equation}
        \ad_{A_2} A_1 = [A_2,A_1] = A_2 A_1 -A_1 A_2.
    \end{equation}
    We also define products and commutators of linear maps $\mcl{A}_1,\cdots,\mcl{A}_n$ in the same way.

    \item Locality and extensiveness: Let $\hat{\mcl{A}}_X$ be a map supported on a domain $X \subset \Lambda$, given by

    \begin{equation}
        \qquad \quad \hat{\mcl{A}}_X \rho = \sum_{\substack{\mu,\nu: \\ \supp(P_\mu) \subset X, \\ \supp(P_\nu) \subset X}} \gamma_{X,\mu\nu} P_\mu \rho P_\nu, \quad \gamma_{X,\mu\nu} \in \bbC. 
    \end{equation}
    We consider a map $\mcl{A}$ in the form of
    \begin{equation}
        \mcl{A} = \sum_{X \subset \Lambda} \hat{\mcl{A}}_X.
    \end{equation}
    We define the support of the map $\mcl{A}$ by
    \begin{equation}
        \supp (\mcl{A}) = \bigcup_{\substack{X \subset \Lambda: \\ \hat{A}_X \neq 0}} X.
    \end{equation}
    We define the locality of the map $\mcl{A}$ by a quantity $k(\mcl{A})$ such that
    \begin{equation}\label{Eq_Setup:k_def}
        \qquad \hat{\mcl{A}}_X = 0, \quad \text{if} \quad  |X| > k(\mcl{A}).
    \end{equation}
    We define the extensiveness $g(\mcl{A})$ by a quantity such that
    \begin{equation}\label{Eq_Setup:g_def}
        \qquad \max_{i \in \supp(\mcl{A})} \left( \sum_{X: X \ni i} \norm{\hat{\mcl{A}}_X}_\mr{Pauli} \right) \leq g(\mcl{A})
    \end{equation}
    is satisfied.
    Throughout, we set $\hat{A}_\emptyset =0$. 
    The support, locality, and extensiveness are understood with respect to the specified local decomposition.
    
    The extensiveness gives an upper bound on the norm of $\mcl{A}$ by
    \begin{eqnarray}
        \qquad \norm{\mcl{A}}_\mr{Pauli} &\leq& \sum_{i \in \supp(\mcl{A})} \sum_{X: X \ni i} \norm{\hat{\mcl{A}}_X}_\mr{Pauli}\nonumber \\
        &\leq& |\supp(\mcl{A})| \, g(\mcl{A}). \label{Eq_Setup:Pauli_norm_bound}
    \end{eqnarray}
    Thus, it means the energy scale per site under $\mcl{A}$.

    \item Hermiticity-preserving (HP) and complete positivity (CP):
    We often consider a Hermiticity-preserving (HP) map such that $(\mcl{A}(\rho))^\dagger = \mcl{A}(\rho)$ for any Hermitian matrix $\rho$.
    The coefficient $\gamma_{\mu\nu}$ in Eq. (\ref{Eq_Setup:Map_Pauli_expansion}) satisfies $\gamma_{\mu\nu} = \gamma_{\nu\mu}^\ast$ when the map is HP.
    In addition, $\mcl{A}$ is completely positive (CP) if and only if the matrix $(\gamma_{\mu\nu})$ is positive semidefinite.

    A non-CP map cannot be implemented directly by quantum channels, but observables of its output can be estimated by the quasi-probabilistic sampling.
    For an HP map $\mcl{A}$, the expectation value $\mr{Tr}[O(1+\mcl{A})\rho]$ ($\norm{O} \leq 1$) under the non-CP map $1+\mcl{A}$ can be estimated with the sampling of quantum circuits and the classical postprocessing.
    The estimation within an additive error $\varepsilon$ with constant success probability can be executed with the sampling complexity, 
    \begin{equation}
        \order{\frac{(1+\norm{\mcl{A}}_\text{Pauli})^2}{\varepsilon^2}},
    \end{equation}
    where each sampled circuit can be reproduced by at most $\order{k(\mcl{A})}$ quantum gates.
    See Lemma \ref{LemmaA:quasiprobabilistic} in Appendix \ref{SecA:Basic} for details.

\end{itemize}

We next discuss the setup for simulation.
We suppose that the $N$-qubit lattice $\Lambda$ is one-dimensional, but many parts of our results can be extended to higher-dimensional cases as discussed later.
We consider a local Lindbladian with finite-range interactions and dissipation.
To be precise, we suppose that the Lindbladian $\mcl{L}$ is given by
\begin{equation}\label{Eq_Basic:L_x}
    \mcl{L} = \sum_{X \subset \Lambda} \hat{\mcl{L}}_X,
\end{equation}
where each $\hat{\mcl{L}}_X$ is a Lindbladian having the support $X=\supp(\hat{\mcl{L}}_X)$.
We assume finite-range interactions and dissipation with range $\xi \in \bbN$, where $\xi \in \order{1}$, in the sense that
\begin{equation}\label{Eq_Basic:range_xi}
    \hat{\mcl{L}}_X = 0, \quad \text{if} \quad r(X) \geq \xi
\end{equation}
is satisfied.
This means that each term involves sites within the distance $\xi$.
We schematically illustrate such generic dissipative systems subject to local interactions and dissipation in Fig. \ref{Fig_systems} (a).

We denote the locality $k(\mcl{L})$ and the extensiveness $g(\mcl{L})$ of the Lindbladian $\mcl{L}$, which are defined by Eqs. (\ref{Eq_Setup:k_def}) and (\ref{Eq_Setup:g_def}), simply by $k$ and $g$.
When we expand the Hamiltonian part $H_X$ and the Lindblad operators $\{L_{X,m}\}$ of each local Lindbladian $\hat{\mcl{L}}_X$ by Pauli operators as
\begin{eqnarray}
    H_X &=& \sum_\mu h_\mu^X P_\mu, \quad h_\mu^X \in \bbR, \label{Eq_Setup:H_Pauli_expansion}\\
    L_{X,m} &=& \sum_\mu l_{m\mu}^X P_\mu, \quad l_{m\mu}^X \in \bbC, \label{Eq_Setup:Lm_Pauli_expansion}
\end{eqnarray}
the locality $k$ implies that $H_X$ and $L_{X,m}$ are supported on at most $k$ sites.
The range $\xi$ immediately implies the relation,
\begin{equation}\label{Eq_Setup:locality_range_relation}
    k \leq \xi.
\end{equation}
Substituting Eqs. (\ref{Eq_Setup:H_Pauli_expansion}) and (\ref{Eq_Setup:Lm_Pauli_expansion}), each local Lindbladian $\hat{\mcl{L}}_X$ is expressed as
\begin{eqnarray}
    \hat{\mcl{L}}_X \rho &=& -i\sum_\mu h_\mu^X [P_\mu,\rho] + 2\sum_{\mu,\nu} \left( \sum_m l_{m\mu}^X l_{m\nu}^{X\ast} \right) P_\mu \rho P_\nu \nonumber \\
    && \quad - \sum_{\mu,\nu} \left( \sum_m l_{m\mu}^X l_{m\nu}^{X\ast} \right) \left\{ P_\nu P_\mu, \rho \right\}. \label{Eq_Setup:Lindbladian_expansion}
\end{eqnarray}
Its Pauli norm can be bounded by
\begin{equation}
    \norm{\hat{\mcl{L}}_X}_\mr{Pauli} \leq 2 \sum_\mu |h_\mu^X| + 4 \sum_m \left( \sum_\mu |l_{m\mu}^X| \right)^2.
\end{equation}
The Lindbladian $\mcl{L}$ has the extensiveness $g$ that can be bounded by
\begin{equation}
    g \leq \max_{i \in \Lambda} \left( \sum_{X: X \ni i} \left[ 2 \sum_\mu |h_\mu^X| + 4 \sum_m \left( \sum_\mu |l_{m\mu}^X| \right)^2 \right]\right).
\end{equation}
The right-hand side can be calculated efficiently by classical computers.
For lattice Lindbladians with finite-range interactions and dissipation, where each local term $\hat{\mcl{L}}_X$ has an upper bound independent of the system size $N$, we have $k \in \order{1}$ and $g \in \order{1}$.

Based on the relations, Eqs. (\ref{Eq_Setup:Norm_relation}) and (\ref{Eq_Setup:Pauli_norm_bound}), we have $\norm{\mcl{L}}_\Diamond, \norm{\mcl{L}}_\mr{Pauli} \leq Ng \in \order{N}$. 

We next describe the problem.
The simulation of Lindbladian dynamics has two goals.
The first one is the simulation of the time-evolved state $e^{\mcl{L}t}\rho$ for the time $t$ and the allowable error $\varepsilon$, in which we realize a quantum state $\rho'$ such that
\begin{equation}
    \norm{\rho'-e^{\mcl{L}t}\rho}_1 \leq \varepsilon
\end{equation}
from any initial state $\rho$.
The symbol $\norm{\cdot}_1$ represents the trace norm.
It is sufficient to construct a quantum channel $\mcl{C}$ such that
\begin{equation}
    \norm{\mcl{C}-e^{\mcl{L}t}}_\Diamond \leq \varepsilon.
\end{equation}
When the algorithm works deterministically, the computational cost for this problem is measured by the number of $\order{1}$-qubit gates in the channel $\mcl{C}$.
The other task is the simulation of the time-evolved observable $\mr{Tr}[O e^{\mcl{L}t}\rho]$ for an observable $O$.
In this case, we aim to obtain an estimate $O_\rho(t)$ satisfying
\begin{equation}
    \left| O_\rho (t) - \mr{Tr} \left[ O e^{\mcl{L}t}(\rho) \right]\right| \leq \varepsilon,
\end{equation}
for any observable $O$ such that $\norm{O}=1$.
In the standard estimation of expectation values, we repeat measurement on the output $\mcl{C}\rho$ generated by some quantum channel $\mcl{C}$.
The computational cost is measured by the cost per experiment, i.e., the number of $\order{1}$-qubit gates in $\mcl{C}$, and the sampling complexity.
Note that the simulation of the time-evolved observable is reproduced by that of the time-evolved state.
Some of our algorithms are available for the simulation of both the time-evolved states and observables, while the others are limited to the time-evolved observables.
We will specify them when each algorithm is established.
Throughout this paper, we exclude the cost for the state preparation, i.e., the gate counts for preparing an initial state $\rho$, and the one for measuring in the basis of an observable $O$.

\subsection{Brief review of existing algorithms}

In this section, we briefly review existing quantum algorithms for Lindbladian simulation.
We will use some of them  as a subroutine of our algorithms, and also compare their computational costs.

The most standard algorithm may be the product formula (PF), which is often called Trotterization \cite{Kliesch-prl2011-open}.
Supposing that the Lindbladian $\mcl{L}$ is decomposed into several terms by $\mcl{L}=\sum_{\gamma=1}^\Gamma \mcl{L}_\gamma$, it relies on the product formulas
\begin{eqnarray}
    \mcl{T}_1(\tau) &=& e^{\mcl{L}_\Gamma\tau} \cdots e^{\mcl{L}_2\tau} e^{\mcl{L}_1\tau} = \prod_{\gamma=1}^\Gamma e^{\mcl{L}_\gamma \tau}, \\
    \mcl{T}_2(\tau) &=& \prod_{\gamma=\Gamma}^1 e^{\mcl{L}_\gamma \tau/2} \prod_{\gamma=1}^\Gamma e^{\mcl{L}_\gamma \tau/2}. \label{Eq_Gen:2nd_PF}
\end{eqnarray}
They approximate Lindbladian dynamics under small time $\tau \to 0$ by $\mcl{T}_p(\tau)=e^{\mcl{L}\tau}+\order{\tau^{p+1}}$ ($p=1,2$).
Choosing each local term in $\mcl{L}$ as $\mcl{L}_\gamma$, each completely-positive and trace-preserving (CPTP) map $e^{\mcl{L}_\gamma \tau}$ can be implemented by Stinespring dilation \cite{Kliesch-prl2011-open}.
The simulation for large evolution time $t$ is executed by implementing $[\mcl{T}_p(\tau)]^{r_t}$, where the repetition number $r_t=t/\tau$ is large enough to achieve the error $\varepsilon$.
Owing to the commutator scaling \cite{childs2021-trotter,Wang-2026-open}, the gate count for simulating lattice Lindbladians with finite-range interactions amounts to
\begin{equation}\label{Eq_Setup:PF_cost}
    \order{\frac{(Nt)^{1+\frac1p}}{\varepsilon^{\frac1p}}},
\end{equation}
where the order $p$ can be either $1$ or $2$.
We note that the higher-order PFs with $p \geq 3$ are unavailable in contrast to Hamiltonian simulation.
This comes from the no-go theorem \cite{Sheng-1989-splitting,Suzuki1991-za}, which states that higher-order PFs cannot be composed solely by forward time evolution operators \cite{Note_Zassenhaus}.
Namely, higher-order PFs for Lindbladian dynamics inevitably involve non-CPTP maps like $e^{-\mcl{L}_\gamma \tau}$, which cannot be implemented.

Another promising quantum algorithm is the extension of the LCU-based approach to Lindbladians \cite{cleve-2016-open,Li-Wang-2022-open}.
It employs the series expansion of $e^{\mcl{L}t}$ and realizes it with queries to block-encodings.
For instance, Li and Wang (2022) \cite{Li-Wang-2022-open} develop an algorithm based on the series expansion by Duhamel's principle,
\begin{equation}
    e^{\mcl{L}t} = e^{\mcl{L}_\mr{D}t}\sum_{q=0}^\infty \int_0^t \dd t_q \cdots \int_0^{t_2} \dd t_1 \prod_{q'=1}^q \left[ e^{-t_{q'} \ad_{\mcl{L}_\mr{D}}} \mcl{L}_\mr{J} \right],
\end{equation}
where $\mcl{L}_\mr{J}\rho = 2 \sum_{m=1}^M L_m \rho L_m^\dagger$ is the jump term and $\mcl{L}_\mr{D} = \mcl{L} -\mcl{L}_\mr{J}$ is the dynamical term.
The algorithm runs with $\order{\norm{\mcl{L}}_\mr{BE}t \log (\norm{\mcl{L}}_\mr{BE}t/\varepsilon)}$ queries to the block-encodings of $H$ and $L_m$, and $\order{M\norm{\mcl{L}}_\mr{BE}t [\log (\norm{\mcl{L}}_\mr{BE}t/\varepsilon)]^2}$ additional 1- or 2-qubit gates.
The symbol $\norm{\cdot}_\mr{BE}$ is a kind of norm determined by the block-encoding, which shares the scaling with $\norm{\cdot}_\mr{Pauli}$ for generic lattice Lindbladians.
Although it achieves the near-optimal query complexity in $t$ and $1/\varepsilon$, it does not mean the optimality in gate count.
Indeed, when considering generic lattice Lindbladians with finite-range interactions and dissipation, the block-encodings require $\order{N}$ local gates, and we have the number of Lindblad operators $M \in \order{N}$ due to the number of terms proportional to the system size.
The gate count for the LCU-based approach amounts to
\begin{equation}\label{Eq_Setup:Gate_LCU}
    \order{N^2 t \, \polylog{Nt/\varepsilon}}.
\end{equation}
It has worse dependency on $N$ than the second-order PF [See Eq. (\ref{Eq_Setup:PF_cost}) for $p=2$].
The number of ancilla qubits for this algorithm is $ \order{\polylog{Nt/\varepsilon}}$.

Various quantum algorithms have appeared for Lindbladian dynamics in the past decade \cite{Ding-Li-Lin-2024-open,kato-wada-2024_open,Yu-Yuan-2025-open,BorrasMarvian2025Lindblad,Peng-prxq-2025-open,Chen2025RandomizedLindblad,Pocrnic2025Repeated,David_2026_qdrift,Yu-Cirac-2025-open,mohammadipour2026-open,Wang-2026-open,wang2026-open-optimal,chen2026-open-optimal}.
Some of them \cite{kato-wada-2024_open,Yu-Yuan-2025-open} employ quasi-probabilistic sampling for simulating time-evolved observables and have inherent advantages like the smaller number of ancilla qubits, but their gate counts are at least as large as $\order{(Nt)^2 \, \polylog{Nt/\varepsilon}}$.
To the best of our knowledge, a quantum algorithm using extrapolation of the second-order PF \cite{Wang-2026-open} has achieved the best size dependency, whose gate count is as large as $\order{(Nt)^{3/2} \, \polylog{Nt/\varepsilon}}$ for observable estimation.

We compare the costs of Lindbladian simulation with those of Hamiltonian simulation in Table \ref{Table:Gate_counts}.
The lower bound on the gate count for Hamiltonian simulation is known to be $\tilde{\Omega}(Nt)$ \cite{Haah2021-hhkl}.
The higher order PFs with the order $p \in \order{1}$ can achieve the scaling close to this lower bound both in $N$ and $t$ \cite{childs-prl2019-pf,childs2021-trotter}.
In addition, the scaling simultaneously good in $N$, $t$, and $1/\varepsilon$ has been recently achieved for Hamiltonian simulation.
The multi-product formula (MPF) combining PF and LCU \cite{low2019-mpf} achieves the gate count $\order{N^{1+1/(p+1)}t \, \polylog{Nt/\varepsilon}}$, which exploits the commutator scaling of PF \cite{Mizuta_2026_mpf}.
The HHKL algorithm utilizing the Lieb-Robinson bound achieves the near-optimal gate count $\order{Nt \, \polylog{Nt/\varepsilon}}$ \cite{Haah2021-hhkl}.
The scaling $\tilde{\Omega}(Nt)$ works as the lower bound on the gate count also for Lindbladian simulation.
The current status of Lindbladian simulation is totally different from that of Hamiltonian simulation.
The current best size-dependency is $\order{N^{3/2}}$, achieved by the second-order PF and its extrapolation \cite{Kliesch-prl2011-open,Wang-2026-open}, which is far from the optimal scaling $\tilde{\Omega}(N)$.
While the optimal dependency in $t$ is saturated by the LCU-based approaches \cite{cleve-2016-open,Li-Wang-2022-open}, their dependence in size $N$ is rather worse.
All the above algorithms for Hamiltonian simulation, i.e., the higher-order PFs, the higher-order MPFs, and the HHKL algorithms, contain backward time evolutions, and cannot be extended to Lindbladian dynamics due to the breakdown of the CP property.
As far as we know, the near-optimal gate count $\order{Nt \, \polylog{Nt/\varepsilon}}$ is achieved for the very limited case, where all the Lindblad operators $\{ L_m \}$ in Eq. (\ref{Eq_Intro:Lindbladian}) are Hermitian and commute with one another under an additional assumption on QRAM queries \cite{Yu-Cirac-2025-open}.
It has been a long-standing open problem whether or how we can achieve the gate count close to the lower bound $\tilde{\Omega}(Nt)$ for simulating broad classes of Lindbladians.

\subsection{Brief summary of our results}

\begin{figure}
    \centering
    \includegraphics[width=0.95\linewidth]{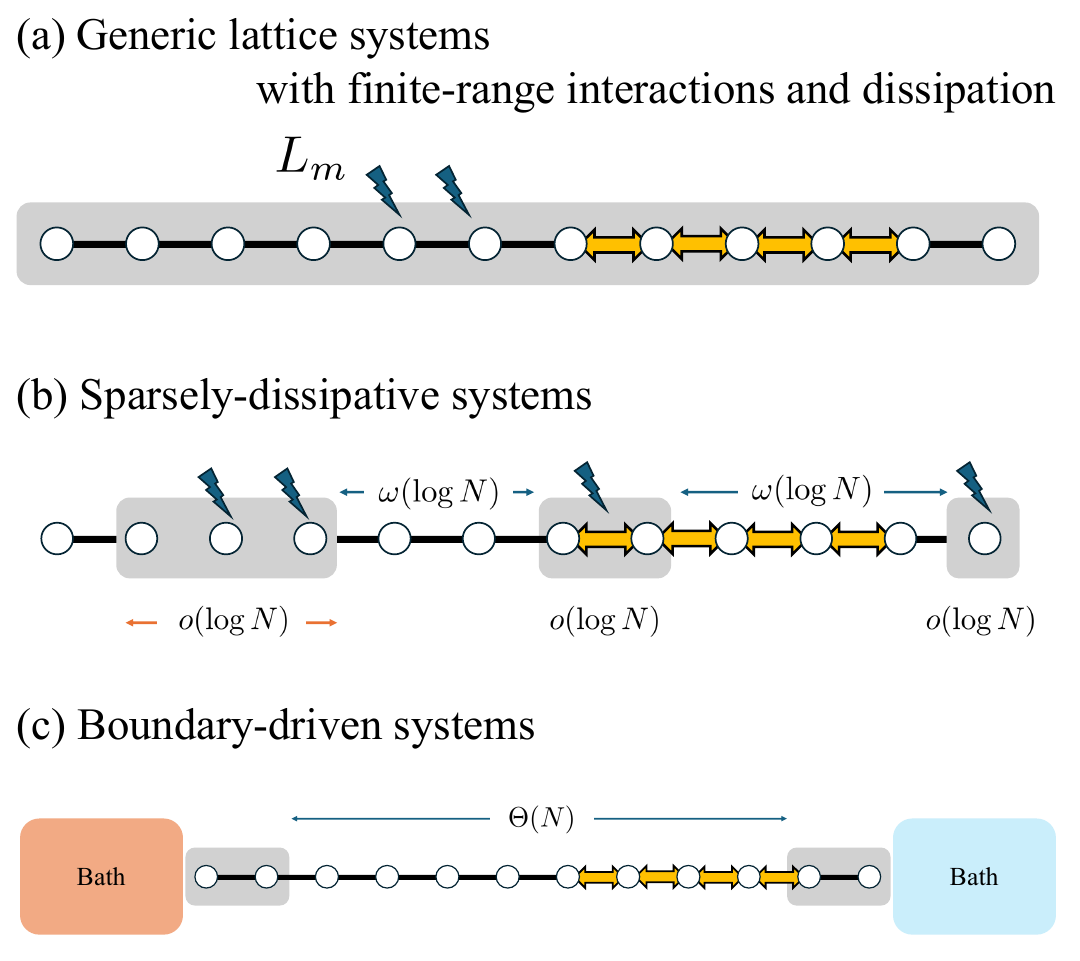}
    \caption{Schematic picture of the system. (a) Generic systems addressed by Algorithm \ref{Algorithm_Generic}. (b) Sparsely dissipative systems satisfying Definition \ref{Def_Setup:sparse_dissipation}, and addressed by Algorithm \ref{Algorithm_Sparse}. (c) Boundary-driven systems as a typical example of sparsely dissipative systems. We consider high-dimensional systems in Appendix \ref{SecA:High_dim}.}
    \label{Fig_systems}
\end{figure}

We briefly summarize our results.
We establish two kinds of quantum algorithms for Lindbladian dynamics.
The first algorithm achieves the near-optimal gate counts for simulating a certain class of Lindbladians, in which the dissipation is sparsely located.
The second algorithm simulates generic Lindbladians with finite-ranged interaction and dissipation, whose gate count achieves better scaling than existing algorithms.

We describe the first algorithm.
In addition to the assumptions in Section \ref{Subsec:Setup}, we assume that the dissipation is sparsely located.
Let us define the support of the dissipation $\Lambda^\mr{diss} \subset \Lambda$ by
\begin{equation}
    \Lambda^\mr{diss} = \bigcup_{X,m} \supp (L_{X,m}),
\end{equation}
where $L_{X,m}$ is the Lindblad operator included in the local term $\hat{\mcl{L}}_X$.
We define sparsely dissipative systems as follows.

\begin{definition}\label{Def_Setup:sparse_dissipation}
\textbf{(Sparsely dissipative systems)}

We call the Lindbladian $\mcl{L}$ sparsely dissipative when the support of its dissipation $\Lambda^\mr{diss}$ is composed of disjoint domains $\{\Lambda^\mr{diss}_\alpha \}$ as
\begin{equation}
    \Lambda^\mr{diss} = \bigcup_\alpha \Lambda^\mr{diss}_\alpha, \quad \Lambda^\mr{diss}_\alpha \cap \Lambda^\mr{diss}_{\alpha'} = \emptyset, \quad (\alpha \neq \alpha'),
\end{equation}
and satisfies the following conditions:
\begin{itemize}
    \item Domain size: Every domain $\Lambda^\mr{diss}_\alpha$ has the size bounded by
    \begin{equation}\label{Eq_Setup:dissipative_domain_size}
        r(\Lambda^\mr{diss}_\alpha) \in o(\log N),
    \end{equation}
    where $r(X)$ is defined by Eq. (\ref{Eq_Setup:domain_size_def}).

    \item Domain distance: Every pair of distinct domains $\Lambda^\mr{diss}_\alpha$, $\Lambda^\mr{diss}_{\alpha'}$ satisfies
    \begin{equation}\label{Eq_Setup:dissipative_domain_dist}
        \mr{dist}(\Lambda^\mr{diss}_\alpha,\Lambda^\mr{diss}_{\alpha'}) \in \omega (\log N), \quad (\alpha \neq \alpha'),
    \end{equation}
    where $\mr{dist}(X,Y)$ is defined by Eq. (\ref{Eq_Setup:domain_dist_def}).
\end{itemize}
\end{definition}

Figure \ref{Fig_systems} (b) shows the schematic picture of sparsely dissipative systems satisfying the above definition.
Importantly, they include boundary-driven systems, where the dissipation is located around the left and right boundaries like Fig. \ref{Fig_systems} (c).
Boundary-driven systems are typical targets in open quantum many-body systems \cite{Landi-2022-boundary}, and hence the first algorithm has broad utility in nonequilibrium physics.
We develop a near-optimal quantum algorithm for sparsely dissipative systems, whose cost is given by the following theorem.

\begin{theorem}\label{Thm_Setup:sparse}
\textbf{(Near-optimal simulation of sparsely dissipative Lindbladians)}

Let $\mcl{L}$ be a Lindbladian for a sparsely dissipative system given by Definition \ref{Def_Setup:sparse_dissipation}.
We also suppose that the time $t$ and the inverse error $1/\varepsilon$ are at most $\poly{N}$.
There exists a quantum algorithm outputting the time-evolved state $e^{\mcl{L}t}\rho$ within an additive error $\varepsilon$, which can be executed by the following cost:
\begin{itemize}
    \item Number of $\order{1}$-qubit gates:
    \begin{equation}
        \order{Nt \, \polylog{Nt/\varepsilon}}.
    \end{equation}

    \item Ancilla qubits and circuit depth:
    The algorithm runs with $\Theta (\polylog{Nt/\varepsilon})$ ancilla qubits, and then it yields the circuit depth $\order{Nt \, \polylog{Nt/\varepsilon}}$.
    When $\tilde{\Theta}(N)$ ancilla qubits are available, the circuit depth can be $\order{t \, \polylog{Nt/\varepsilon}}$.
\end{itemize}
\end{theorem}
The above gate count matches the lower bound $\tilde{\Omega}(Nt)$.
In addition, when $\tilde{\Theta}(N)$ ancilla qubits are available, the circuit depth is also near-optimal.

As the second algorithm, we establish the way to simulate generic local Lindbladians with finite-ranged interactions and dissipation.
This algorithm allows us to reproduce the time-evolved observable $\mr{Tr}[Oe^{\mcl{L}t}\rho]$ by repeating the execution of quantum gates and measurements and classically post-processing the measurement outcomes.
Its cost is summarized in the following theorem.

\begin{theorem}\label{Thm_Sparse:generic}
\textbf{(Observable simulation of generic lattice Lindbladians)}

Let $\mcl{L}$ be a one-dimensional lattice Lindbladian with finite-ranged interactions and dissipation.
There exists a quantum algorithm outputting the time-evolved observable $\mr{Tr}[O e^{\mcl{L}t}\rho]$ within an additive error $\varepsilon$, which can be executed by the following cost:
\begin{itemize}
    \item Number of $\order{1}$-qubit gates per experiment:
    \begin{equation}
        \order{(Nt)^{\frac43} \, \polylog{Nt/\varepsilon}}.
    \end{equation}

    \item Ancilla qubits and circuit depth: The algorithm requires $\Theta (\polylog {Nt/\varepsilon})$ ancilla qubits, and then, the circuit depth amounts to $(Nt)^{4/3} \, \polylog{Nt/\varepsilon}$.
    When $\tilde{\Theta} (N^{2/3})$ ancilla qubits are available, the circuit depth amounts to $t (Nt)^{2/3} \, \polylog{Nt/\varepsilon}$.

    \item Sampling complexity: $\Theta (\varepsilon^{-2})$.
\end{itemize}
\end{theorem}

As far as we know, this algorithm achieves the best size-$N$ dependency in the gate counts among the existing algorithms, which is close to the lower bound $\tilde{\Omega}(N)$.
We also note that our gate count becomes the best among the known results in the time regime $t \in \order{N^2}$.

The strategies for constructing these quantum algorithms are patching and merging, which are based on the locality and the range of interactions and dissipation.
First, we prove the so-called patching lemma for Lindbladian dynamics, which allows us to decompose the time-evolution operators $e^{\mcl{L}t}$ into those for small dissipative systems, as we will show in Section \ref{Sec:Patching}.
Such decomposition was originally developed for Hamiltonian dynamics, which led to the near-optimal quantum algorithm for Hamiltonian simulation \cite{Haah2021-hhkl}.
However, if the Lindbladian dynamics is decomposed in the same way as the Hamiltonian dynamics, such an algorithm fails to be efficient due to the existence of inverse time-evolution operators, which violates complete positivity (CP).
We develop techniques for avoiding this problem, i.e., patching and merging.
To be precise, we organize a way to decompose the Lindbladian dynamics into those for the optimized-size blocks with merging some of them like Fig. \ref{Fig_algorithm_generic} in Section \ref{Sec:Generic_algorithm}.
This deletes or suppresses the violation of the CP property respectively for the first or second algorithms, which makes them the most efficient among the existing algorithms.
We will discuss the construction of these algorithms with the patching and merging strategy in Sections \ref{Sec:Sparse_algorithm} and \ref{Sec:Generic_algorithm}.

\section{Patching lemma for Lindbladians}\label{Sec:Patching}

In this section, we prove the patching lemma for Lindbladian dynamics, which decomposes the time-evolution operator $e^{\mcl{L}t}$ into those for small systems for the algorithms.

\begin{figure}
    \centering
    \includegraphics[width=\linewidth]{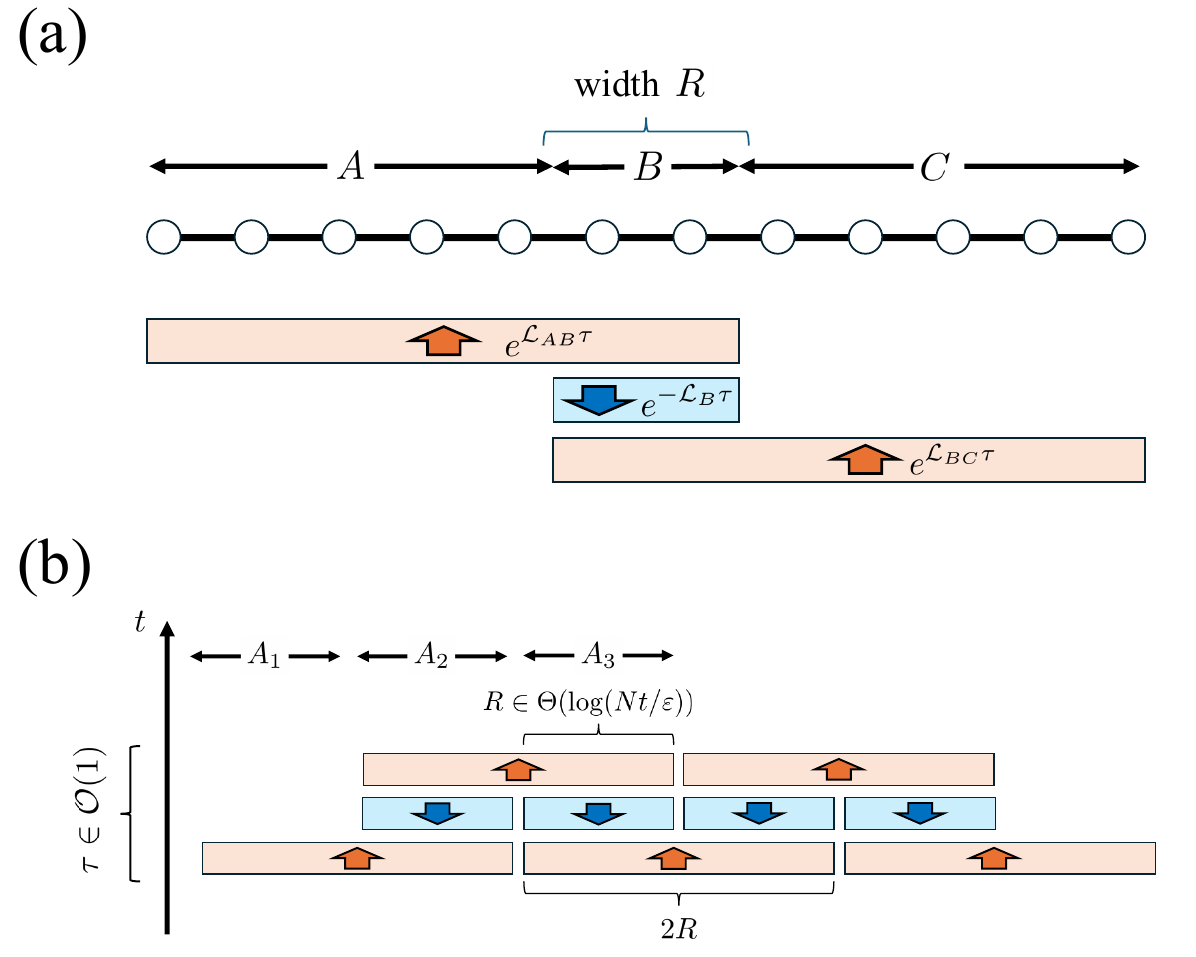}
    \caption{(a) Patching lemma for the Lindbladian system. (b) The decomposition used in the HHKL algorithm.}
    \label{Fig_patching}
\end{figure}

We first split the system $\Lambda$ into $\Lambda = A \cup B \cup C$ as shown in Fig. \ref{Fig_patching} (a).
We define the subsystem Lindbladian $\mcl{L}_{A_1 A_2 \cdots}$ for some domains $\{A_\alpha\}$ by
\begin{equation}\label{Eq_Patch:Subsys_Lindbladian}
    \mcl{L}_{A_1A_2\cdots} = \sum_{X \subset (A_1 \cup A_2 \cup \cdots)} \hat{\mcl{L}}_X,
\end{equation}
where each $\hat{\mcl{L}}_X$ denotes a local Lindbladian having the support $X=\supp(\hat{\mcl{L}}_X)$, as shown in Eq. (\ref{Eq_Basic:L_x}).
For instance, the subsystem Lindbladians $\mcl{L}_A$ and $\mcl{L}_{AB}$ respectively given by
\begin{equation}
    \mcl{L}_A = \sum_{X \subset A} \hat{\mcl{L}}_X, \quad 
    \mcl{L}_{AB} = \sum_{X \subset (A \cup B)} \hat{\mcl{L}}_X
\end{equation}
mean the collection of the terms whose supports are included in $A$ or $A \cup B$.
We also use $\mcl{L}_B$, $\mcl{L}_C$, and $\mcl{L}_{BC}$ defined in the same way.

\begin{theorem}\label{Thm_Patch:Patching_lemma}
\textbf{(Patching lemma for Lindbladians)}

Let $R = |B|$ be the size of the domain $B$ and satisfy $R > \xi$.
When the time $\tau$ is small enough to satisfy
\begin{equation}\label{Eq_Patch:time_assumption}
    0 \leq \tau \leq \frac{1}{6e \xi g} \in \order{1},
\end{equation}
the time-evolution operator $e^{\mcl{L}\tau}$ is approximated by
\begin{equation}\label{Eq_Patch:error}
    \norm{e^{\mcl{L}\tau} - e^{\mcl{L}_{AB}\tau} e^{-\mcl{L}_{B}\tau} e^{\mcl{L}_{BC}\tau}}_\Diamond \leq e^{-\frac{R}{\xi}}.
\end{equation}
\end{theorem}

This theorem is an extension of the so-called patching lemma to Lindbladian dynamics.
In Hamiltonian dynamics governed by a lattice Hamiltonian $H$ with finite-range interactions, the patching lemma yields
\begin{equation}\label{Eq_Patch:Patching_Lemma_Hamiltonian}
    \norm{e^{-iH\tau}-e^{-iH_{AB}\tau} e^{iH_B\tau} e^{-iH_{BC}\tau}} \leq c e^{-\frac{R}\xi},
\end{equation}
under $\tau \in \order{(\xi g)^{-1}}$, where $c > 0$ denotes a constant \cite{Osborne-2006-patching,michalakis-2012-patching,Haah2021-hhkl,Babbush-2018-encoding}.
The operators $H_{AB}$, $H_B$, and $H_{BC}$ are subsystem Hamiltonians defined in a similar manner to Eq. (\ref{Eq_Patch:Subsys_Lindbladian}).
It relies on the Lieb-Robinson bound for Hamiltonian dynamics \cite{Lieb1972-uo}.
The error bound in the existing extension to generic non-unitary time evolutions contains a factor exponentially large in the norm of the generator (i.e., it is exponentially large in the system size $N$) \cite{Kuwahara-2021-patching}.
We derive the patching lemma for Lindbladian dynamics, which is free from exponentially-large factors, by explicitly using the norm restriction $\norm{e^{\mcl{L}\tau}}_\Diamond \leq 1$ ($\tau \geq 0$).
Our derivation relies on the following lemma, which comes from the locality of Lindbladians.

\begin{lemma}\label{Lem_Patch:Commutator_bound}
\textbf{(Bound on nested commutators)}

Let $\mcl{A}_0,\mcl{A}_1,\cdots,\mcl{A}_q$ be HP maps whose locality and extensiveness are respectively $(k_0,g_0),(k_1,g_1), \cdots, (k_q,g_q)$.
The nested commutator $\prod_{q'=1}^q (\ad_{\mcl{A}_{q'}}) \mcl{A}_0$ is HP and at most $(\sum_{q'=0}^q k_{q'})$-local.
Its Pauli norm is bounded by
\begin{eqnarray}
    && \norm{\left( \prod_{q'=1}^q \ad_{\mcl{A}_{q'}}\right) \mcl{A}_0}_\mr{Pauli} \nonumber \\
    && \qquad \leq \norm{\mcl{A}_0}_\mr{Pauli} \prod_{q'=1}^q \left[ \left( \sum_{q'' =0}^{q'-1} k_{q''} \right) 2g_{q'} \right].
\end{eqnarray}
In particular, if $\mcl{A}_0,\mcl{A}_1,\cdots,\mcl{A}_q$ share the same locality $k$ and extensiveness $g$, the nested commutator $\prod_{q'=1}^q (\ad_{\mcl{A}_{q'}}) \mcl{A}_0$ is at most $(q+1)k$-local and has the Pauli norm bounded by $q! (2kg)^q g \, |\supp(\mcl{A}_0)|$.
\end{lemma}

The proof of Lemma \ref{Lem_Patch:Commutator_bound} follows the same argument as the proof for local Hamiltonians \cite{Kuwahara2016-yn}, where the norm is replaced by the Pauli norm.
We proceed to the proof of Theorem \ref{Thm_Patch:Patching_lemma} as follows.

\textbf{Proof of Theorem \ref{Thm_Patch:Patching_lemma}.---}
We define a map $\mcl{N}(\tau)$ by
\begin{equation}\label{Eq_Patch:N_tau}
    \mcl{N}(\tau) = e^{-\mcl{L}\tau} e^{\mcl{L}_{AB}\tau} e^{-\mcl{L}_{B}\tau} e^{\mcl{L}_{BC}\tau}.
\end{equation}
It gives the error bound by
\begin{eqnarray}
    \norm{e^{\mcl{L}\tau}-e^{\mcl{L}_{AB}\tau} e^{-\mcl{L}_{B}\tau} e^{\mcl{L}_{BC}\tau}}_\Diamond &=& \norm{e^{\mcl{L}\tau}(1-\mcl{N}(\tau))}_\Diamond \nonumber \\
    &\leq& \norm{\mcl{N}(\tau)-1}_\Diamond. \label{Eq_Patch:N_tau_1_bound}
\end{eqnarray}
We calculate the map $\mcl{N}(\tau)$ as follows,
\begin{eqnarray}
    && \mcl{N}(\tau) \nonumber \\
    && \quad =\mcl{N}(0)+\int_0^\tau \dd \tau' \dv{\tau'} \mcl{N}(\tau') \nonumber \\
    && \quad = 1 - \int_0^\tau \dd \tau' e^{-\mcl{L}\tau'} (\mcl{L}-\mcl{L}_{AB}) e^{\mcl{L}_{AB}\tau'} e^{-\mcl{L}_{B}\tau'} e^{\mcl{L}_{BC}\tau'} \nonumber \\
    && \qquad + \int_0^\tau \dd \tau' e^{-\mcl{L}\tau'}  e^{\mcl{L}_{AB}\tau'} e^{-\mcl{L}_{B}\tau'} (\mcl{L}_{BC}-\mcl{L}_B) e^{\mcl{L}_{BC}\tau'}. \nonumber \\
    && \label{Eq_Patch:N_tau_integral}
\end{eqnarray}
Let $\mcl{L}_{B:C}$ denote the inter-block interactions, defined by
\begin{equation}\label{Eq_Patch:Lindbladian_boundary}
    \mcl{L}_{B:C} = \sum_{\substack{X \subset \Lambda: \\ X \cap B \neq \emptyset, X \cap C \neq \emptyset}} \hat{\mcl{L}}_X,
\end{equation}
when $\mcl{L}$ is expressed by Eq. (\ref{Eq_Basic:L_x}).
Since the size $|B| = R$ is larger than the range $\xi$, it satisfies
\begin{eqnarray}
    \mcl{L}_{B:C}  &=& \mcl{L}-\mcl{L}_{AB} - \mcl{L}_C \label{Eq_Patch:AB_boundary_1} \\
    &=& \mcl{L}_{BC} - \mcl{L}_B - \mcl{L}_C. \label{Eq_Patch:AB_boundary_2}
\end{eqnarray}
We substitute the above relations into Eq. (\ref{Eq_Patch:N_tau_integral}).
Considering that $\mcl{L}_C$ commutes with $e^{\mcl{L}_{AB}\tau'}$ and $e^{-\mcl{L}_B\tau'}$, we obtain
\begin{eqnarray}
    \mcl{N}(\tau) &=& 1 + \int_0^\tau \dd \tau' e^{-\mcl{L}\tau'} \left[ e^{\mcl{L}_{AB}\tau'} e^{-\mcl{L}_B\tau'}, \mcl{L}_{B:C}\right] e^{\mcl{L}_{BC}\tau'} \nonumber \\
    &=& 1 +  \int_0^\tau \dd \tau' \mcl{K}(\tau') \mcl{N}(\tau'), \label{Eq_Patch:N_tau_integral_eq}
\end{eqnarray}
where we define the map $\mcl{K}(\tau')$ by
\begin{equation}\label{Eq_Patch:K_s}
    \mcl{K}(\tau') = e^{-\tau'\ad_{\mcl{L}}} \left( e^{\tau' \ad_{\mcl{L}_{AB}}} e^{-\tau'\ad_{\mcl{L}_B}}-1 \right) \mcl{L}_{B:C}.
\end{equation}
Since the map $\mcl{K}(\tau')$ is bounded, we can solve the integral equation Eq. (\ref{Eq_Patch:N_tau_integral_eq}) and arrive at the expression,
\begin{equation}
    \mcl{N}(\tau) = \sum_{n=0}^\infty \int_0^\tau \dd \tau_n \cdots \int_0^{\tau_2} \dd \tau_1 \mcl{K}(\tau_n) \cdots \mcl{K}(\tau_1).
\end{equation}
We obtain
\begin{eqnarray}
    \norm{\mcl{N}(\tau)-1}_\Diamond &\leq& \sum_{n=1}^\infty \int_0^\tau \dd \tau_n \cdots \int_0^{\tau_2} \dd \tau_1 \prod_{n'=1}^n \norm{\mcl{K}(\tau_{n'})} \nonumber \\
    &\leq& \sum_{n=1}^\infty \frac{1}{n!} \left( \tau \, \sup_{\tau' \in [0,\tau]} (\norm{\mcl{K}(\tau')}_\Diamond)\right)^n, \label{Eq_Patch:Ns_1_bound}
\end{eqnarray}
which gives the upper bound on Eq. (\ref{Eq_Patch:N_tau_1_bound}).

We next evaluate the upper bound on $\norm{\mcl{K}(\tau')}_\Diamond$ based on the locality and the extensiveness.
Considering the series expansion of Eq. (\ref{Eq_Patch:K_s}), the map $\mcl{K}(\tau')$ is expressed by
\begin{eqnarray}
    && \mcl{K}(\tau') \nonumber \\
    && \, = \sum_{l=0}^\infty \sum_{\substack{m,n \geq 0: \\ 1 \leq m+n}} \frac{(-\tau' \ad_{\mcl{L}})^l(\tau' \ad_{\mcl{L}_{AB}})^{m} (-\tau'\ad_{\mcl{L}_B})^{n}}{l!m!n!}  \mcl{L}_{B:C}. \nonumber \\
    &&
\end{eqnarray}
Then, we use the fact that the interactions and the dissipation are finite-ranged.
Equations (\ref{Eq_Patch:Subsys_Lindbladian}) and (\ref{Eq_Patch:Lindbladian_boundary}) give the nested commutator in the above equation as follows,
\begin{eqnarray}
    &&  (\ad_{\mcl{L}_{AB}})^{m} (\ad_{\mcl{L}_B})^{n} \mcl{L}_{B:C} \nonumber \\
    && = \sum_{\substack{Z \subset \Lambda: \\ Z \cap B \neq \emptyset, \\ Z \cap C \neq \emptyset}}\sum_{\substack{X_1,\cdots,X_m\\ \subset (A \cup B)}} \sum_{\substack{Y_1,\cdots,Y_n \\ \subset B}} \prod_{m'=1}^m \ad_{\hat{\mcl{L}}_{X_{m'}}} \prod_{n'=1}^n \ad_{\hat{\mcl{L}}_{Y_{n'}}} \hat{\mcl{L}}_Z. \nonumber \\
    && \label{Eq_Patch:Commutator_expansion}
\end{eqnarray}
When the commutator with $\hat{\mcl{L}}_{X_{m'}}$ such that $X_{m'} \cap A \neq \emptyset$ can give nontrivial contributions, the domain $X_{m'-1} \cup \cdots \cup X_1 \cup Y_n \cup \cdots \cup Y_1 \cup Z$ should be connected and include a site within the distance $\xi$ from the boundary of $A$ and $B$.
In other words, each of $\hat{\mcl{L}}_{X_{m'}}$ such that $X_{m'} \cap A \neq \emptyset$ gives no contribution if we have
\begin{equation}
    r(X_{m} \cup \cdots \cup X_1 \cup Y_n \cup \cdots \cup Y_1 \cup Z) < \mr{dist}(A,C).
\end{equation}
Since we have $r(X_{m} \cup \cdots \cup X_1 \cup Y_n \cup \cdots \cup Y_1 \cup Z) \leq (m+n+1)\xi$ and $\mr{dist}(A,C) \geq R$, this implies the relation,
\begin{equation}\label{Eq_Patch:commutator_boundary}
    (\ad_{\mcl{L}_{AB}})^{m} (\ad_{\mcl{L}_B})^{n} \mcl{L}_{B:C} = (\ad_{\mcl{L}_B})^{m+n} \mcl{L}_{B:C}
\end{equation}
under $m+n < R/\xi - 1$.
As a result, the terms in Eq. (\ref{Eq_Patch:Commutator_expansion}) with $m+n < R/\xi - 1$ vanish as
\begin{equation}\label{Eq_Patch:commutator_delete}
    \sum_{\substack{m,n \geq 0: \\ 1 \leq m+n < R/\xi-1}} \frac{(\ad_{\mcl{L}_{AB}})^{m} (-\ad_{\mcl{L}_B})^{n}}{m!n!}  \mcl{L}_{B:C} = 0.
\end{equation}
The norm of the map $\mcl{K}(\tau')$ is bounded by
\begin{eqnarray}
    \norm{\mcl{K}(\tau')}_\Diamond &\leq& \sum_{l=0}^\infty \sum_{\substack{m,n \geq 0: \\ \lceil \frac{R}\xi-1 \rceil \leq m+n}} \frac{(\tau')^{l+m+n}}{l!m!n!} \nonumber \\
    && \qquad \times \norm{(\ad_{\mcl{L}})^l(\ad_{\mcl{L}_{AB}})^m (\ad_{\mcl{L}_B})^n \mcl{L}_{B:C}}_\Diamond \nonumber \\
    &\leq& \sum_{l=0}^\infty \sum_{\substack{m,n \geq 0: \\ \lceil \frac{R}\xi-1 \rceil \leq m+n}} \frac{(\tau')^{l+m+n}}{l!m!n!} \nonumber \\
    && \qquad \times (l+m+n)! (2kg)^{l+m+n}\norm{\mcl{L}_{B:C}}_\mr{Pauli} \nonumber \\
    &\leq& \sum_{q=\lceil \frac{R}\xi -1 \rceil}^\infty \sum_{\substack{l,m,n \geq 0: \\ l+m+n = q}} \frac{q!}{l!m!n!} \nonumber \\
    && \qquad \qquad \times (2kg\tau')^q g \, |\supp (\mcl{L}_{B:C}) |  \nonumber \\
    &\leq& 2 \xi g \sum_{q=\lceil \frac{R}\xi -1 \rceil}^\infty (6\xi g\tau')^q \nonumber \\
    &\leq& \frac{2e^2}{e-1} e^{-\frac{R}\xi} \xi g . \label{Eq_Patch:K_s_bound}
\end{eqnarray}
In the second inequality, we use the fact that the diamond norm is smaller than the Pauli norm as Eq. (\ref{Eq_Setup:Norm_relation}), and apply Lemma \ref{Lem_Patch:Commutator_bound}, which comes from the locality and the extensiveness.
We use Eqs. (\ref{Eq_Setup:Pauli_norm_bound}) and (\ref{Eq_Setup:locality_range_relation}) respectively for the third and fourth inequalities.
We also use $\sum_{l,m,n \geq 0: l+m+n=q} q!/(l!m!n!) = 3^q$.
The last inequality relies on the assumption Eq. (\ref{Eq_Patch:time_assumption}), which is applicable for $\tau' \in [0,\tau]$.

Using the above upper bound and the assumption Eq. (\ref{Eq_Patch:time_assumption}), the quantity appearing in Eq. (\ref{Eq_Patch:Ns_1_bound}) is bounded by
\begin{equation}
    \tau \, \sup_{\tau' \in [0,\tau]} (\norm{\mcl{K}(\tau')}_\Diamond) \leq \frac{e}{3(e-1)} e^{-\frac{R}\xi} \leq 1.
\end{equation}
This immediately results in the relation,
\begin{equation}\label{Eq_Patch:N_tau_1_bound_result}
     \norm{\mcl{N}(\tau)-1}_\Diamond \leq \sum_{n=1}^\infty \frac1{n!} \frac{e}{3(e-1)} e^{-\frac{R}\xi} \leq e^{-\frac{R}\xi}.
\end{equation}
Since it gives the error bound as we discussed in Eq. (\ref{Eq_Patch:N_tau_1_bound}), we complete the proof of Eq. (\ref{Eq_Patch:error}). $\quad \square$

Theorem \ref{Thm_Patch:Patching_lemma} implies that the time-evolution operator $e^{\mcl{L}\tau}$ can be decomposed based on the locality and the finite range of the interactions like that for Hamiltonian dynamics \cite{Osborne-2006-patching}, as shown in Fig. \ref{Fig_patching} (a).
In the HHKL algorithm \cite{Haah2021-hhkl}, the patching lemma is repeated and the time evolution $e^{-iH\tau}$ is decomposed like Fig. \ref{Fig_patching} (b).
We show that the time evolution $e^{\mcl{L}\tau}$ can be decomposed into those for small blocks in the same way as follows.

\begin{corollary}\label{Cor_Patch:HHKL_Lindbladian}
\textbf{}

Let $N/R$ be an even integer for simplicity.
We assume $R \geq \xi \max (1,\log N)$.
We set a domain $A_\alpha$ by $A_\alpha = \{ (\alpha-1)R+1, (\alpha-1)R+2, \cdots, \alpha R\}$, and the HHKL decomposition of $e^{\mcl{L}\tau}$ by
\begin{eqnarray}
    && \mcl{U}_\mr{HHKL}(\tau) \nonumber \\
    && \quad = \prod_{\alpha=1}^{\frac{N}{2R}-1} e^{\mcl{L}_{A_{2\alpha}A_{2\alpha+1}}\tau} \prod_{\alpha=2}^{\frac{N}{R}-1} e^{-\mcl{L}_{A_\alpha}\tau} \prod_{\alpha=1}^{\frac{N}{2R}} e^{\mcl{L}_{A_{2\alpha-1}A_{2\alpha}}\tau}, \nonumber \\
    && \label{Eq_Patch:HHKL_Lindblad}
\end{eqnarray}
like Fig. \ref{Fig_patching} (b).
When the time $\tau$ is small enough to satisfy $\tau \in \order{(\xi g)^{-1}} = \order{1}$, it approximates the Lindbladian dynamics as 
\begin{equation}
    \norm{e^{\mcl{L}\tau}-\mcl{U_\mr{HHKL}(\tau)}}_\Diamond \leq N e^{-\frac{R}\xi}.
\end{equation}
\end{corollary}

This corollary is derived by the same calculation as the proof of Theorem \ref{Thm_Patch:Patching_lemma}, rather than following directly from Theorem \ref{Thm_Patch:Patching_lemma}, and hence we provide its proof in Appendix \ref{SubsecA:Patching_corollary}.
In spite of the validity of the HHKL decomposition $\mcl{U}_\mr{HHKL}(\tau)$, it does not give any efficient algorithm for Lindbladian dynamics.
This comes from the backward evolutions $e^{-\mcl{L}_{A_\alpha}\tau}$, which break the CP property.
Non-CP maps cannot be simulated directly by quantum channels, and we cannot implement the map, Eq. (\ref{Eq_Patch:HHKL_Lindblad}), by any quantum circuit.
On the other hand, observables of their outputs can be reproduced by quasi-probabilistic sampling, composed of sampling of quantum circuits and classical post-processing of measurement outcomes.
However, even if we use the quasi-probabilistic sampling, the HHKL algorithm suffers from the exponential sampling overhead.
The sampling overhead for reproducing the non-CP map $e^{-\mcl{L}_{A_\alpha}\tau}$ generally amounts to $e^{\order{\norm{\mcl{L}_{A_\alpha}}_\text{Pauli} \tau}}$ (See also Appendix \ref{SecA:Basic}).
When we wish to reproduce $e^{\mcl{L}t}$ by $(t/\tau)$-times application of the HHKL decomposition $\mcl{U}_\mr{HHKL}(\tau)$, the sampling overhead in total is as large as
\begin{equation}\label{Eq_Patch:Overhead_HHKL}
    \left( \prod_{\alpha=2}^{\frac{N}R-1} e^{\order{\norm{\mcl{L}_{A_\alpha}}_\text{Pauli} \tau}} \right)^{t/\tau}  \subset e^{\order{Ngt}}. 
\end{equation}
We use the relation $\norm{\mcl{L}_{A_\alpha}}_\text{Pauli} \leq g \, |A_\alpha| \in \order{Rg}$ for the second line.
Thus, the overhead is exponentially large in spacetime when we execute quasi-probabilistic sampling independently for the backward time evolutions, and the HHKL algorithm is not available to Lindbladian dynamics.
Our algorithms in Sections \ref{Sec:Sparse_algorithm} and \ref{Sec:Generic_algorithm} completely avoid or suppress this problem by elaborating the patching and merging strategies.
\section{Algorithm for sparsely dissipative cases}\label{Sec:Sparse_algorithm}

In this section, we develop a near-optimal quantum algorithm for simulating a sparsely dissipative system defined by Definition \ref{Def_Setup:sparse_dissipation}.
Throughout this section, we assume $t \in \poly{N}$ and $1/\varepsilon \in \poly{N}$, which are natural for efficient computation.

\subsection{Patching strategy}

We introduce the patching strategy for the fast simulation.
In the HHKL algorithm, achieving the near-optimality for Hamiltonian dynamics, the time evolution $e^{-iH\tau}$ is decomposed into those for blocks having the sizes $R, 2R \in \Theta (\log N)$.
As we discussed at the end of Section \ref{Sec:Patching}, this fails for Lindbladian dynamics.
To overcome this difficulty, we introduce the decomposition with different block sizes like Fig. \ref{Fig_patch_sparse} (a) as follows.
We split the lattice $\Lambda$ into $N_\mr{p}$ domains $\{ A_\alpha\}$, where each domain $A_\alpha$ has the flexible size $R_\alpha$.
Each domain can be explicitly given by
\begin{equation}
    A_\alpha = \Set{ \sum_{\alpha' < \alpha} R_{\alpha'} + 1, \sum_{\alpha' < \alpha} R_{\alpha'} + 2, \cdots, \sum_{\alpha' \leq \alpha} R_{\alpha'} }
\end{equation}
for $\alpha=1,2,\cdots,N_\mr{p}$.
We assume that the block number $N_\mr{p}$ is an even integer without loss of generality below.
We organize the decomposed time-evolution operator $\mcl{U}_\mr{Patch}(\tau)$ by
\begin{eqnarray}
    \mcl{U}_\mr{Patch}(\tau) &=& \prod_{\alpha} e^{\mcl{L}_{A_{4\alpha-2}A_{4\alpha-1}A_{4\alpha}} \tau }\nonumber \\
    && \quad \times \prod_{\alpha} e^{-\mcl{L}_{A_{2\alpha}}\tau} \prod_{\alpha} e^{\mcl{L}_{A_{4\alpha}A_{4\alpha+1}A_{4\alpha+2}}\tau}, \nonumber \\
    && \label{Eq_Sparse:Patch_op}
\end{eqnarray}
where a domain $A_\alpha$ for $\alpha \notin \{1,2,\cdots,N_\mr{p} \}$ denotes the empty set $\emptyset$.
This decomposition is obtained by repeating the one in the patching lemma, Theorem \ref{Thm_Patch:Patching_lemma}:
First, we regard $A_1$, $A_2$, and $A_3 \cup \cdots A_{N_\mr{p}}$ respectively as $A$, $B$, and $C$ in Eq. (\ref{Eq_Patch:error}).
We next split the system $A_3 \cup \cdots A_{N_\mr{p}}$ into $A_5 \cup \cdots \cup A_{N_\mr{p}}$, $A_4$, and $A_3$.
Repeating this procedure results in Eq. (\ref{Eq_Sparse:Patch_op}) like the HHKL decomposition in Corollary \ref{Cor_Patch:HHKL_Lindbladian}.
Indeed, the decomposed time-evolution $\mcl{U}_\mr{Patch}(\tau)$ approximates $e^{\mcl{L}\tau}$ by the following corollary.

\begin{figure}
    \centering
    \includegraphics[width=\linewidth]{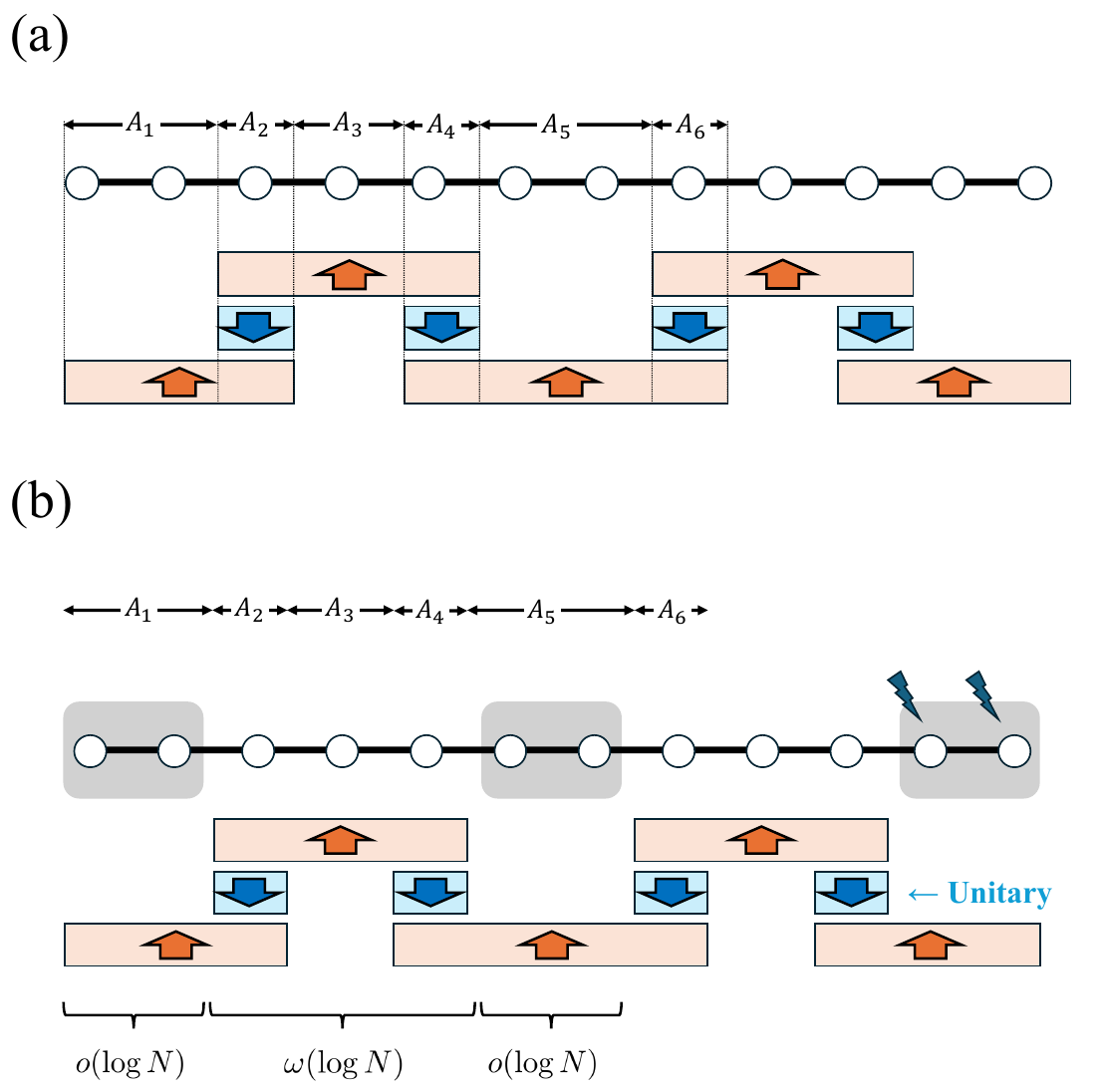}
    \caption{(a) The patching lemma with the flexible block sizes, given by Corollary \ref{Cor_Sparse:Patching_strategy}. (b) The patching strategy for the sparsely dissipative systems. We set the domains $A_2,A_4,\cdots$ so that they can avoid the sites subject to the dissipation.}
    \label{Fig_patch_sparse}
\end{figure}

\begin{corollary}\label{Cor_Sparse:Patching_strategy}
\textbf{}

We assume that the size of the even-indexed blocks $R_{2\alpha}$ is at least $R$, satisfying
\begin{equation}
    R \geq \xi \max (1, \log N).
\end{equation}
When the time $\tau \in \order{1}$ is small enough to satisfy Eq. (\ref{Eq_Patch:time_assumption}), the map $\mcl{U}_\mr{Patch}(\tau)$ defined by Eq. (\ref{Eq_Sparse:Patch_op}) approximates the exact time evolution with an error bounded by
\begin{equation}\label{Eq_Sparse:error_strategy}
    \norm{e^{\mcl{L}\tau} - \mcl{U}_\mr{Patch}(\tau)}_\Diamond \leq N e^{-\frac{R}\xi}.
\end{equation}

\end{corollary}
The proof of this corollary is completely parallel to the one for the patching lemma, Theorem \ref{Thm_Patch:Patching_lemma}, in Section \ref{Sec:Patching}.
We provide it in Appendix \ref{SubsecA:Patching_corollary}.

Next, we choose the domain sizes $\{R_\alpha\}$ for sparsely dissipative systems.
Based on Definition \ref{Def_Setup:sparse_dissipation}, let $N_\mr{d}$ be the number of the disjoint domains $\{ \Lambda_\alpha^\mr{diss} \}$, which are the supports of the dissipation.
Without loss of generality, we assume that the supports of the dissipation begin with the left edge, i.e., we assume $\Lambda_1^\mr{diss} \ni 1$.
The strategy for sparsely dissipative systems is to organize the patching so that the folding domains $\{A_{2\alpha}\}$, which have potentials to bring the non-CP maps in Eq. (\ref{Eq_Sparse:Patch_op}), can avoid the supports of the dissipation $\{ \Lambda_\alpha^\mr{diss} \}$.
To be concrete, we set the domains $\{A_\alpha\}$ as follows:
We first set the domains $A_{4\alpha-3}$ by
\begin{equation}\label{Eq_Sparse:even_block_0}
    A_{4\alpha-3} = \Lambda_\alpha^\mr{diss}
\end{equation}
for $\alpha=1,2,\cdots,N_\mr{d}$.
We next determine the adjacent domains with the size $R$ as
\begin{eqnarray}
    A_{4\alpha-4} &=& \Set{ \min_{j \in \Lambda_\alpha^\mr{diss}} (j) - R, \cdots, \min_{j \in \Lambda_\alpha^\mr{diss}} (j) -1} \cap \Lambda, \nonumber \\
    && \label{Eq_Sparse:even_block_1}\\
    A_{4\alpha-2} &=& \Set{ \max_{j \in \Lambda_\alpha^\mr{diss}} (j) + 1, \cdots, \max_{j \in \Lambda_\alpha^\mr{diss}} (j) + R } \cap \Lambda. \nonumber \\
    &&\label{Eq_Sparse:even_block_2}
\end{eqnarray}
The domains $\{A_{4\alpha-1} \}$ are located between them, explicitly given by
\begin{equation}\label{Eq_Sparse:even_block_3}
    A_{4\alpha-1} = \Set{ j \in \Lambda| \max_{j' \in A_{4\alpha-2}} (j') < j < \min_{j' \in A_{4\alpha}} (j')}.
\end{equation}
See Fig. \ref{Fig_patch_sparse} (b) for the schematic picture of this decomposition.
Importantly, each even-indexed domain $A_{2\alpha}$ does not contain any site in the supports of the dissipation $\{ \Lambda_\alpha^\mr{diss} \}$.
As a result, the inverse time-evolution $e^{-\mcl{L}_{A_{2\alpha}} \tau}$ in Eq. (\ref{Eq_Sparse:Patch_op}) becomes a unitary map.
This is why we can construct a near-optimal quantum algorithm, as we will explicitly construct and evaluate it in the next section.

\subsection{Algorithm and Cost}

Here, we construct the algorithm and show its cost.
For the evolution time $t$, we set the time $\tau$ by
\begin{equation}
    \tau = \frac{t}{\lceil 6e \xi g t\rceil} \in \order{1},
\end{equation}
which satisfies the assumption Eq. (\ref{Eq_Patch:time_assumption}) required for Corollary \ref{Cor_Sparse:Patching_strategy}.
We organize a quantum channel $\mcl{C}(\tau)$, which is composed of the application of $\order{1}$-local qubit gates and tracing out of the ancilla systems, so that its repetition for
\begin{equation}
    r_t = \lceil 6 e \xi g t \rceil \in \bbN
\end{equation}
times can approximate the time evolution $e^{\mcl{L}t}$ as
\begin{equation}
    \norm{e^{\mcl{L}t} - [\mcl{C}(\tau)]^{r_t}}_\Diamond \leq \varepsilon.
\end{equation}
According to the patching lemma by Corollary \ref{Cor_Sparse:Patching_strategy}, it is sufficient to construct $\mcl{C}(\tau)$ such that
\begin{equation}\label{Eq_Sparse:requirement_patching}
    \norm{e^{\mcl{L}\tau}-\mcl{U}_\mr{Patch}(\tau)}_\Diamond \leq \frac{\varepsilon}{2r_t}
\end{equation}
and
\begin{equation}\label{Eq_Sparse:requirement_channel}
    \norm{\mcl{C}(\tau)-\mcl{U}_\mr{Patch}(\tau)}_\Diamond \leq \frac{\varepsilon}{2r_t}
\end{equation}
can be satisfied.
The sufficiency of the above condition is easily confirmed by $\norm{e^{\mcl{L}t} - [\mcl{C}(\tau)]^{r_t}}_\Diamond \leq r_t \norm{e^{\mcl{L}\tau}-\mcl{C(\tau)}}_\Diamond$ for the quantum channels $e^{\mcl{L}\tau}, \mcl{C}(\tau)$ and the triangle inequality.

We consider the requirement Eq. (\ref{Eq_Sparse:requirement_patching}).
Based on the patching lemma by Corollary \ref{Cor_Sparse:Patching_strategy}, we set the number $R$, which is the size of the even-indexed blocks $A_{2\alpha}$ in Eqs. (\ref{Eq_Sparse:even_block_1}) and (\ref{Eq_Sparse:even_block_2}), by
\begin{equation}\label{Eq_Sparse:block_size_R}
    R = \lceil \xi \log (4 Nr_t/\varepsilon) \rceil \in \order{\log (Nt/\varepsilon)}.
\end{equation}
This choice is possible due to the assumption of sparsely dissipative systems in Definition \ref{Def_Setup:sparse_dissipation}, which ensures that $\mr{dist}(\Lambda_\alpha^\mr{diss},\Lambda_{\alpha+1}^\mr{diss}) \in \omega (\log N)$ is larger than $2R$ under $t, 1/\varepsilon \in \poly{N}$.
Equation (\ref{Eq_Sparse:error_strategy}) in Corollary \ref{Cor_Sparse:Patching_strategy} immediately ensures the satisfaction of Eq. (\ref{Eq_Sparse:requirement_patching}).
The remaining task is to construct the quantum channel $\mcl{C}(\tau)$ approximating $\mcl{U}_\mr{Patch}(\tau)$ by Eq. (\ref{Eq_Sparse:requirement_channel}).
We implement a set of quantum channels that approximate each component in $\mcl{U}_\mr{Patch}(\tau)$ described by Eq. (\ref{Eq_Sparse:Patch_op}).
We use different existing quantum algorithms depending on the blocks in the following way.
\begin{itemize}
    \item Implementation of $e^{\mcl{L}_{A_{4\alpha}A_{4\alpha+1}A_{4\alpha+2}}\tau}$: Since $A_{4\alpha+1}=\Lambda_{\alpha+1}^\mr{diss}$ is a support of dissipation, this part is dissipative dynamics.
    We run the LCU-based quantum algorithm for Lindbladian dynamics \cite{Li-Wang-2022-open} so that each $e^{\mcl{L}_{A_{4\alpha}A_{4\alpha+1}A_{4\alpha+2}}\tau}$ can be approximated within an error $\varepsilon/(6Nr_t)$.
    
    \item Implementation of $e^{-\mcl{L}_{A_{2\alpha}} \tau}$: This part is unitary dynamics.
    We employ the quantum algorithms for Hamiltonian dynamics whose gate count can be poly-logarithmic in $1/\varepsilon$ (i.e., LCU \cite{Berry-prl2015-LCU}, QSVT \cite{Low2019-qubitization,Gilyen2019-qsvt}, MPF \cite{low2019-mpf,Mizuta_2026_mpf}, or HHKL \cite{Haah2021-hhkl}).
    Each time evolution $e^{-\mcl{L}_{A_{2\alpha}} \tau}$ is reproduced by quantum circuits within an error $\varepsilon/(6Nr_t)$.
    
    \item Implementation of $e^{\mcl{L}_{A_{4\alpha-2}A_{4\alpha-1}A_{4\alpha}}\tau}$: This part is unitary dynamics.
    We run the HHKL algorithm \cite{Haah2021-hhkl}, achieving the near-optimal gate count for Hamiltonian dynamics.
    Each time evolution $e^{\mcl{L}_{A_{4\alpha-2}A_{4\alpha-1}A_{4\alpha}}\tau}$ is reproduced within an error $\varepsilon/(6Nr_t)$.
    
\end{itemize}
We set the quantum channel $\mcl{C}(\tau)$ by the set of the quantum operations above.
Since $\mcl{U}_\mr{Patch}(\tau)$ contains at most $3N$ time-evolution operators for the blocks, the quantum channel $\mcl{C}(\tau)$ can approximate $\mcl{U}_\mr{Patch}(\tau)$ within an error $\varepsilon/(2r_t)$, indicating the satisfaction of the requirement Eq. (\ref{Eq_Sparse:requirement_channel}).
We summarize the protocol in Algorithm \ref{Algorithm_Sparse}.
The cost of this quantum algorithm is given by the following theorem.

\begin{figure}[t]
\begin{algorithm}[H] 
	\caption{Near-Optimal Simulation of Sparsely Dissipative Systems}
	\label{Algorithm_Sparse}
	\begin{algorithmic}[1]
	\REQUIRE Initial state $\rho$, evolution time $t$, target error $\varepsilon$, the set of local operators that make up the Lindbladian $\mcl{L}$, the support of the sparse dissipation $\Lambda^\mathrm{diss} = \bigcup_\alpha \Lambda_\alpha^\mathrm{diss}$.
\ENSURE Quantum state $\rho'$ satisfying $\|\rho' - e^{\mathcal{L}t}\rho\|_1 \le \varepsilon$.
\STATE Set step count $r_t \leftarrow \lceil 6\mathrm{e}\xi g t \rceil$, step size $\tau \leftarrow t / r_t$, and block size $R \leftarrow \lceil \xi \log (4Nr_t/\varepsilon) \rceil$.
\STATE Construct partition $\{A_\alpha\}$ according to Eqs.~(\ref{Eq_Sparse:even_block_0})--(\ref{Eq_Sparse:even_block_3}).
\FOR{step $r' = 1$ \TO $r_t$}
    \STATE Apply dissipative channel $\bigotimes_\alpha e^{\mathcal{L}_{A_{4\alpha}A_{4\alpha+1}A_{4\alpha+2}}\tau}$ via the LCU-based Lindbladian simulation.
    \STATE Apply unitary backward time evolution $\bigotimes_\alpha e^{-\mcl{L}_{A_{2\alpha}} \tau}$ via QSVT.
    \STATE Apply unitary time evolution $\bigotimes_\alpha e^{\mcl{L}_{A_{4\alpha-2}A_{4\alpha-1}A_{4\alpha}}\tau}$ via the HHKL algorithm.
\ENDFOR
\STATE \textbf{return} Final output state $\rho'$.
	\end{algorithmic}
\end{algorithm}
\end{figure}

\begin{theorem*}
\textbf{(Restatement of Theorem \ref{Thm_Setup:sparse})}

Let $\mcl{L}$ be a Lindbladian for a sparsely dissipative system given by Definition \ref{Def_Setup:sparse_dissipation}.
Algorithm \ref{Algorithm_Sparse} outputs a quantum state $\rho' = [\mcl{C}(\tau)]^{r_t} \rho$ such that the time-evolved state $e^{\mcl{L}t}\rho$ can be approximated as $\norm{\rho'-e^{\mcl{L}t}\rho}_1 \leq \varepsilon$.
The cost of Algorithm \ref{Algorithm_Sparse}, i.e., that for implementing the quantum channel $[\mcl{C}(\tau)]^{r_t}$, is composed of the following resources:
\begin{itemize}
    \item Number of $\order{1}$-qubit gates: It amounts to
    \begin{equation}\label{Eq_Sparse:optimal_gate_count}
        \order{ Nt \, \polylog{Nt/\varepsilon}},
    \end{equation}
    which is near-optimal.

    \item Ancilla qubit number and circuit depth:
    The algorithm requires $\Theta (\polylog{Nt/\varepsilon})$ ancilla qubits, and then, the circuit depth amounts to
    \begin{equation}
        \order{Nt \, \polylog{Nt/\varepsilon}}.
    \end{equation}
    When $\tilde{\Theta}(N)$ ancilla qubits are available, the circuit depth amounts to
    \begin{equation}\label{Eq_Sparse:optimal_gate_depth}
        \order{t \, \polylog{Nt/\varepsilon}},
    \end{equation}
    which is near-optimal.
    
\end{itemize}
We note that the above cost holds even when we are only allowed to use geometrically local $\order{1}$-qubit gates.
\end{theorem*}

\textbf{Proof.---}
We evaluate the gate count for implementing each component in $\mcl{U}_\mr{Patch}(\tau)$ as follows.
\begin{itemize}
    \item Implementation of $e^{\mcl{L}_{A_{4\alpha}A_{4\alpha+1}A_{4\alpha+2}}\tau}$: The block size is as large as
    \begin{eqnarray}
        |A_{4\alpha} \cup A_{4\alpha+1} \cup A_{4\alpha+2}| &\leq& |\Lambda_{\alpha+1}^\mr{diss}|+2R \nonumber \\
        &\in& \order{\log (Nt/\varepsilon)},
    \end{eqnarray}
    which comes from the assumption Eq. (\ref{Eq_Setup:dissipative_domain_size}) and the choice of $R$ by Eq. (\ref{Eq_Sparse:block_size_R}).
    The LCU-based implementation of Lindbladian dynamics $e^{\mcl{L}_{A_{4\alpha}A_{4\alpha+1}A_{4\alpha+2}}\tau}$ within an error $\varepsilon/(6Nr_t)$ requires 
    \begin{equation}
        \Otilde{[\log(Nt/\varepsilon)]^2 \tau} \in \order{\polylog{Nt/\varepsilon}}
    \end{equation}
    $\order{1}$-qubit gates according to Eq. (\ref{Eq_Setup:Gate_LCU}) \cite{Li-Wang-2022-open}.
    We use $\order{\polylog {[(|\Lambda_{\alpha+1}^\mr{diss}|+2R)\tau]/[\varepsilon/(Nr_t)]}} \subset \order{\polylog {Nt/\varepsilon}}$ ancilla qubits.
    
    \item Implementation of $e^{-\mcl{L}_{A_{2\alpha}} \tau}$: The block size $|A_{2\alpha}|=R$ scales as $\order{\log (Nt/\varepsilon)}$.
    When we employ QSVT for Hamiltonian dynamics \cite{Low2019-qubitization,Gilyen2019-qsvt}, the gate count amounts to
    \begin{equation}
        \qquad \order{R(R\tau + \log (R\tau / [\varepsilon/(Nr_t)]))} \subset \order{[\log (Nt/\varepsilon)]^2}
    \end{equation}
    for each $e^{-\mcl{L}_{A_{2\alpha}} \tau}$.
    We use $\order{\polylog{R}} \subset \order{\poly{\log \log (Nt/\varepsilon)}}$ ancilla qubits.
    
    \item Implementation of $e^{\mcl{L}_{A_{4\alpha-2}A_{4\alpha-1}A_{4\alpha}}\tau}$: 
    We run the HHKL algorithm for Hamiltonian dynamics \cite{Haah2021-hhkl}.
    Each time evolution $e^{\mcl{L}_{A_{4\alpha-2}A_{4\alpha-1}A_{4\alpha}}\tau}$ on the block size $|A_{4\alpha-2} \cup A_{4\alpha-1} \cup A_{4\alpha}| \leq |A_{4\alpha-1}|+2R$ ($\leq N$) can be reproduced with the gate count,
    \begin{eqnarray}
        && \order{(|A_{4\alpha-1}|+2R) \tau \, \polylog{N\tau /[\varepsilon/(Nr_t)]}} \nonumber \\
        && \quad \subset \order{(|A_{4\alpha-1}|+2R) \, \polylog{Nt/\varepsilon}},
    \end{eqnarray}
    within an error $\varepsilon/(6Nr_t)$.
    The number of ancilla qubits is $\order{\log \log ([|A_{4\alpha-1}|+2R]/[\varepsilon/(Nr_t)])} \subset \order{\log \log (Nt/\varepsilon)}$ when we do not parallelize.
    When we parallelize with $\Theta (|A_{4\alpha-1}|+2R)=\Theta(|A_{4\alpha-1}|)$ ancilla qubits, the circuit depth becomes $\order{\tau \, \polylog{Nt/\varepsilon}} = \order{\polylog{Nt/\varepsilon}}$.
\end{itemize}
When we repeat the above implementation for every block $r_t$ times, the total gate count is as large as
\begin{eqnarray}
    && \order{ r_t \left( N_\mr{d}+ \sum_{\alpha=1}^{N_\mr{d}} |A_{4\alpha-1}| \right) \, \polylog{Nt/\varepsilon}}\nonumber \\
    && \quad \subset \order{Nt \, \polylog{Nt/\varepsilon}}.
\end{eqnarray}
This completes the proof of the near-optimal gate count by Eq. (\ref{Eq_Sparse:optimal_gate_count}).

We next evaluate the number of ancilla qubits and the circuit depth.
In the LCU-based algorithm or the HHKL algorithm for each block, each ancilla qubit begins with the state $\ket{0}$, and returns to $\ket{0}$ or is discarded at the end.
The ancilla qubits can be reused over the blocks and the $r_t$ steps.
As a result, their number amounts to $\order{\polylog{Nt/\varepsilon}}$, which is the maximal number among the subroutines.
The circuit depth is as large as the gate count, given by Eq. (\ref{Eq_Sparse:optimal_gate_count}).
On the other hand, when we prepare ancilla qubits respectively for different blocks, we can parallelize the algorithm.
Then, the number of ancilla qubits is as large as
\begin{equation}
    \order{N_\mr{d} \, \polylog{Nt/\varepsilon} + \sum_{\alpha=1}^{N_\mr{d}} (|A_{4\alpha-1}|+R)}.
\end{equation}
The number $N_\mr{d}$ is at most $\order{N/\log (Nt/\varepsilon)}$ since the supports $\{\Lambda_\alpha^\mr{diss}\}$ are located with the distance $\omega (\log (Nt/\varepsilon))$.
Thus, $\tilde{\Theta} (N)$ ancilla qubits suffice for parallelization, and then, the circuit depth becomes $\order{t \, \polylog{Nt/\varepsilon}}$.
Finally, concerning the geometrical locality of $\order{1}$-qubit gates, we note that every quantum algorithm used as a subroutine is closed within each block.
Geometrically nonlocal gates used for $e^{\mcl{L}_{A_{4\alpha}A_{4\alpha+1}A_{4\alpha+2}}\tau}$ and $e^{-\mcl{L}_{A_{2\alpha}} \tau}$ act on distant qubits whose distance is at most $\order{\log (Nt/\varepsilon)}$.
Each of them can be replaced by $\order{\log (Nt/\varepsilon)}$ geometrically local gates using SWAP operations.
In the implementation of $e^{\mcl{L}_{A_{4\alpha-2}A_{4\alpha-1}A_{4\alpha}}\tau}$, the HHKL algorithm for Hamiltonian dynamics runs also with geometrically local gates while keeping the cost.
Thus, the cost when we use geometrically local $\order{1}$-qubit gate is essentially the same as Eqs. (\ref{Eq_Sparse:optimal_gate_count})-(\ref{Eq_Sparse:optimal_gate_depth}).
This completes the proof of Theorem \ref{Thm_Setup:sparse}. $\quad \square$

The gate count $\Otilde{Nt}$ is near-optimal.
This is immediately confirmed by the fact that Lindbladian dynamics of sparsely dissipative systems contain generic Hamiltonian dynamics, whose simulation including time-dependent cases requires $\tilde{\Omega}(Nt)$ gates \cite{Haah2021-hhkl}.
The optimality of the circuit depth $\Otilde{t}$, achieved with the parallelization, follows from the same reason.
Whether and how we can achieve the optimal cost in the gate number has been the central problem, while that for Hamiltonian dynamics under finite-ranged interaction was resolved by HHKL algorithm \cite{Haah2021-hhkl}.

Recently, Yu et al. (2025) \cite{Yu-Cirac-2025-open} have found the possibility of achieving the optimal cost in a very limited case.
They consider dissipative systems in which all the Lindblad operators $\{L_m\}$ in Eq. (\ref{Eq_Intro:Lindbladian}) are Hermitian and commute with one another as
\begin{equation}\label{Eq_Sparse:commuting_Lindblad}
    L_m^\dagger = L_m, \quad [L_m,L_{m'}] = 0, \, ^\forall m,m'.
\end{equation}
They express such Lindblad dynamics as a stochastic Hamiltonian dynamics, and simulate it in the interaction picture.
Their algorithm employs QRAM access to the coefficients of the trajectory-dependent interaction-picture Hamiltonian, having the dimension $\order{(Nt/\varepsilon)^4}$.
Under the assumption that the QRAM for each trajectory can be efficiently implemented by $\order{\polylog{Nt/\varepsilon}}$ gates, it can achieve the near-optimal gate count $\Otilde{Nt}$.
Our results have advantages in broad application and feasible implementation.
Sparsely dissipative systems in our setup do not contain the above case nor vice versa.
Our class can deal with damping noise or particle loss, which are described by non-Hermitian jump operators, as long as they are sparsely located.
Importantly, our class contains boundary-driven systems like Fig. \ref{Fig_systems} (c), which have been vigorously explored in nonequilibrium condensed matter physics \cite{Landi-2022-boundary}.
In addition, our algorithm achieves the near-optimal gate count $\Otilde{Nt}$ without any assumption on the QRAM access.
It will be a significant step toward the construction of an optimal algorithm for generic Lindbladian dynamics.

\section{Algorithm for generic dissipative cases}\label{Sec:Generic_algorithm}

\begin{figure*}[t]
    \centering
    \includegraphics[width=0.95\linewidth]{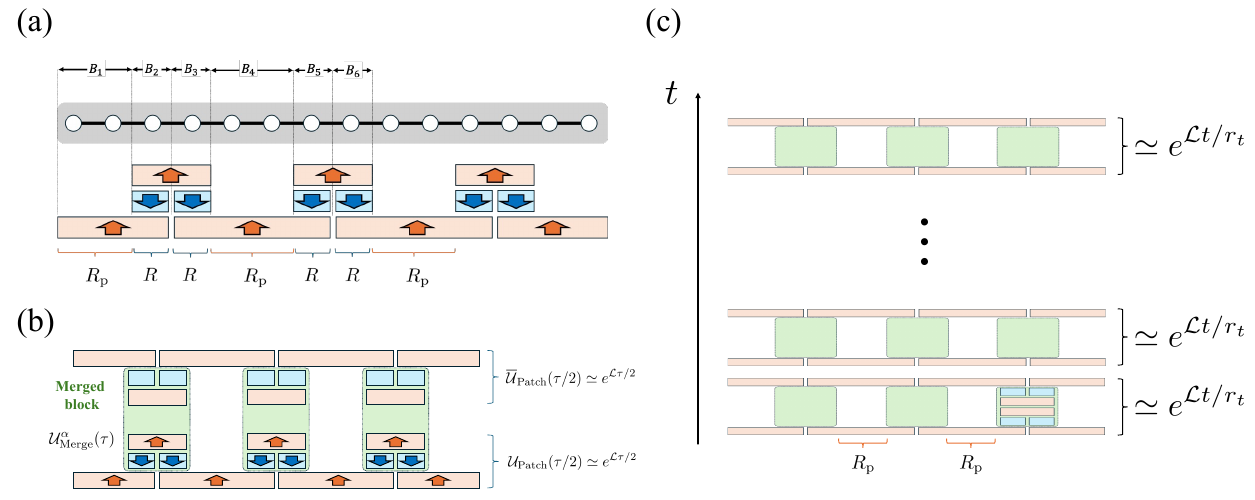}
    \caption{(a) The patching strategy for generic local Lindbladians, giving $\mcl{U}_\mr{Patch}(\tau)$ by Eq. (\ref{Eq_Gen:Patch_op}). We optimize the block size $R_\mr{p}$ to minimize the gate count while we set $R \in \Theta (\log (Nt/\varepsilon))$. (b) The merging strategy based on Eqs. (\ref{Eq_Gen:Merge_op}) and (\ref{Eq_Gen:Patch_w_Merge}). The internal layers are merged and regarded as a set of disjoint green blocks with the size $\order{R}$. (c) The outline of Algorithm \ref{Algorithm_Generic}. We repeat the operation with the time $\tau=t/r_t$. The orange blocks with the size $\order{R_\mr{p}}$ are reproduced by the LCU based approach for Lindbladian dynamics, while we use quasi-probabilistic sampling for the merged green blocks. The repetition number $r_t$ is determined so that the sampling overhead can be suppressed.}
    \label{Fig_algorithm_generic}
\end{figure*}

Throughout this section, we consider quantum many-body systems with finite-ranged interactions and dissipation, satisfying all the assumptions in Section \ref{Sec:Summary}.
We establish a way to efficiently compute arbitrary time-evolved observables under such generic lattice Lindbladians as the second algorithm, using the quasi-probabilistic sampling.
In addition to the patching strategy in Section \ref{Sec:Sparse_algorithm}, we develop the merging strategy, in which we merge some non-CP time-evolution operators on blocks for suppressing the sampling overhead and the resulting gate count.
As a result, our algorithm achieves the best dependency in the size $N$ among the existing algorithms, while retaining the polylogarithmic dependency in $1/\varepsilon$.
We develop the patching and merging strategies in Section \ref{Subsec:Outline_algo}, and discuss the details of the algorithm and its cost in Section \ref{Subsec:Cost_algo}.

\subsection{Algorithm outline: patching and merging}\label{Subsec:Outline_algo}

Here, we discuss the ideas of the patching and merging strategies, and show the outline of the algorithm.
Let us consider generic dissipative systems, whose dissipation can appear everywhere.
In contrast to sparsely-dissipative systems in Section \ref{Sec:Sparse_algorithm}, the decomposition of $e^{\mcl{L}\tau}$ by the patching lemma, Theorem \ref{Thm_Patch:Patching_lemma}, inevitably contains non-CP maps like $e^{-\mcl{L}_{A_{2\alpha}}\tau}$.
We cannot directly implement such non-CP maps by quantum channels, but instead we can reproduce their outputs by the quasi-probabilistic sampling.
Thus, we focus on the problem of reproducing the time-evolved observable $\mr{Tr}[Oe^{\mcl{L}t}\rho]$.
As mentioned in the impossibility of extending the HHKL algorithm in Section \ref{Sec:Patching}, the sampling overhead is the central issue.
We develop an efficient algorithm for the sampling and the classical post-processing based on the patching and merging strategies below, which substantially suppress the sampling overhead.

\subsubsection{Patching strategy}

In the patching strategy, we decompose the time evolution $e^{\mcl{L}\tau}$ by Corollary \ref{Cor_Sparse:Patching_strategy} with the adjustable block sizes.
While we adjust the block sizes for avoiding backward dissipative evolutions for sparsely dissipative systems, we hereby do so to minimize the cost for the sampling complexity and the gate count per sample in total.

We introduce two characteristic scales of length for the blocks $R$ and $R_\mr{p}$, and set the block size as follows,
\begin{eqnarray}
    |A_1| &=& |A_5| = |A_9| = \cdots = R_\mr{p}, \\
    |A_2| &=& |A_4| = |A_6| = \cdots = R, \\
    |A_3| &=& |A_7| = |A_{11}| = \cdots = 0.
\end{eqnarray}
We relabel the indices of the domains $\{A_\alpha\}$ by using
\begin{eqnarray}
    B_{3\alpha-2} &=& \{ (\alpha-1)(R_\mr{p}+2R)+1, \nonumber \\
    && \qquad \cdots, (\alpha-1)(R_\mr{p}+2R) + R_\mr{p} \}, \label{Eq_Gen:B_partition_1} \\
    B_{3\alpha-1} &=& \{ (\alpha-1)(R_\mr{p}+2R) + R_\mr{p} +1, \nonumber \\
    && \quad \cdots, (\alpha-1)(R_\mr{p}+2R) + R_\mr{p} + R \}, \label{Eq_Gen:B_partition_2}\\
    B_{3\alpha} &=& \{ (\alpha-1)(R_\mr{p}+2R) + R_\mr{p} + R +1, \nonumber \\
    && \qquad \qquad\qquad\qquad \cdots, \alpha (R_\mr{p}+2R) \}. \label{Eq_Gen:B_partition_3}
\end{eqnarray}
The domains $B_{3\alpha-2}$, $B_{3\alpha-1}$, $B_{3\alpha}$ are respectively interpreted as $A_{4\alpha-3}$, $A_{4\alpha-2}$, and $A_{4\alpha}$, whose sizes are $R_\mr{p}$, $R$, and $R$.
We denote the number of the blocks by $N_\mr{p}$.
The map $\mcl{U}_\mr{Patch}(\tau)$ in Eq. (\ref{Eq_Sparse:Patch_op}) is rewritten as
\begin{eqnarray}
    \mcl{U}_\mr{Patch}(\tau) &=& \prod_{\alpha} e^{\mcl{L}_{B_{3\alpha-1}B_{3\alpha}}\tau} \prod_{\alpha} e^{-(\mcl{L}_{B_{3\alpha-1}}+\mcl{L}_{B_{3\alpha}})\tau} \nonumber \\
    && \quad \qquad \qquad \times \prod_{\alpha} e^{\mcl{L}_{B_{3\alpha}B_{3\alpha+1}B_{3\alpha+2}}\tau}, \label{Eq_Gen:Patch_op}
\end{eqnarray}
and has the error bound by Eq. (\ref{Eq_Sparse:error_strategy}).
We show the schematic picture of the decomposition in Fig. \ref{Fig_algorithm_generic} (a).

In the algorithm, we keep the block size $R \in \Theta (\log (Nt/\varepsilon))$ like Algorithm \ref{Algorithm_Sparse}, while we adjust the size $R_\mr{p} \in \omega (\log (Nt/\varepsilon))$.
Suppose that the CP maps $e^{\mcl{L}_{B_{3\alpha-1}B_{3\alpha}}\tau}$ and $e^{\mcl{L}_{B_{3\alpha}B_{3\alpha+1}B_{3\alpha+2}}\tau}$ are implemented by the LCU-based Lindbladian simulation, and the non-CP map $e^{-(\mcl{L}_{B_{3\alpha-1}}+\mcl{L}_{B_{3\alpha}})\tau}$ is reproduced by the quasi-probabilistic sampling, though we note that this implementation is actually imprecise due to the following merging strategy.
The change in the size $R_\mr{p}$ brings the tradeoff between the gate count per sample and the sampling overhead.
Let us consider the case where $R_\mr{p}$ is as small as $R \in \Theta (\log (Nt/\varepsilon))$.
This reduces to the naive extension of the HHKL algorithm \cite{Haah2021-hhkl} for Lindbladian dynamics.
Although the gate count per sample amounts to $\Otilde{Nt}$, the sampling complexity becomes exponentially large as discussed in Section \ref{Sec:Patching}.
On the other hand, when we make the size $R_\mr{p}$ larger, the number of the non-CP maps proportional to $N_\mr{p} \in \Theta (N/R_\mr{p})$ becomes smaller.
As a result, the sampling overhead for reproducing the non-CP maps becomes small instead of consuming more gates for implementing the other components.
In summary, there seems to be an intermediate scale suitable for $R_\mr{p}$, with which we can achieve preferable scalings in both of the gate count per sample and the sampling complexity.
The central idea in the patching strategy is to find such an optimal choice of the size $R_\mr{p}$.
We will determine the size $R_\mr{p}$ after identifying the dependence of the gate count and the sampling overhead on it in Section \ref{Subsec:Cost_algo}.

\subsubsection{Merging strategy}

We next develop the merging strategy as a technique suppressing the sampling overhead, in which we absorb some of the non-CP maps into other parts.
Let us first discuss the implementation of the non-CP maps in $\mcl{U}_\mr{Patch}(\tau)$.
We repeat the evolution over $\tau = t/r_t$ with setting a large splitting number $r_t$, and hence, we assume that $\tau$ is small enough below.
When we reproduce all the non-CP maps in Eq. (\ref{Eq_Gen:Patch_op}) by quasi-probabilistic sampling, the overhead becomes as large as
\begin{eqnarray}
    \prod_\alpha e^{\order{\left( \norm{\mcl{L}_{B_{3\alpha-1}}}_\text{Pauli}+\norm{\mcl{L}_{B_{3\alpha}}}_\text{Pauli} \right) \tau}} &\subset& e^{\order{2Rg\tau \times N_\mr{p}}} \nonumber \\
    &=& e^{\order{RNg \tau/ R_\mr{p}}}, \nonumber \\
    && \label{Eq_Gen:overhead_wo_merge}
\end{eqnarray}
in a similar manner to the discussion around Eq. (\ref{Eq_Patch:Overhead_HHKL}).
The total overhead in $r_t$ steps amounts to $(e^{RNg \tau/ R_\mr{p}})^{r_t} = e^{RNgt/R_\mr{p}}$.
We can suppress the sampling overhead up to $\order{1}$ by setting $R_\mr{p} \in \Omega (RNt)$ in contrast to the HHKL algorithm.
However, when the block size $R_\mr{p}$ is proportional to the system size $N$, the gate count does not decrease \footnote{Suppose that we use a quantum algorithm for Lindbladian dynamics, whose gate count for size $N$ amounts to $\order{N^\alpha}$, for each block. The gate count for $\mcl{U}_\mr{Patch}(\tau)$, composed of $N/R_\mr{p}$ blocks with the size $R_\mr{p}$ is as large as $(N/R_\mr{p}) \times (R_\mr{p})^\alpha$. When $R_\mr{p}$ is proportional to the whole size $N$, it becomes $\order{N^\alpha}$. This implies that the decomposition by $\mcl{U}_\mr{Patch}(\tau)$ has no benefit.}.
The sampling overhead for $\mcl{U}_\mr{Patch}(\tau)$ by Eq. (\ref{Eq_Gen:Patch_op}) is too large to improve the computational cost.

We develop a way to reduce the sampling overhead from $e^{\order{RNg\tau/R_\mr{p}}}$ to $e^{\order{N(g\tau)^3/R_\mr{p}}}$.
First, in a similar manner to $\mcl{U}_\mr{Patch}(\tau)$ in Eq. (\ref{Eq_Gen:Patch_op}), we introduce another approximation of the time-evolution $e^{\mcl{L}\tau}$, defined by
\begin{eqnarray}
    \overline{\mcl{U}}_\mr{Patch}(\tau) &=&\prod_\alpha e^{\mcl{L}_{B_{3\alpha}B_{3\alpha+1}B_{3\alpha+2}}\tau}  \nonumber \\
    && \quad \times \prod_\alpha e^{-(\mcl{L}_{B_{3\alpha-1}}+\mcl{L}_{B_{3\alpha}})\tau} \prod_\alpha e^{\mcl{L}_{B_{3\alpha-1}B_{3\alpha}}\tau}. \nonumber \\
    &&
\end{eqnarray}
It is obtained by reversing the order in $\mcl{U}_\mr{Patch}(\tau)$.
The same calculation in the proof of Corollary \ref{Cor_Sparse:Patching_strategy} guarantees its error bound,
\begin{equation}\label{Eq_Gen:error_Patch_bar}
    \norm{e^{\mcl{L}\tau}-\overline{\mcl{U}}_\mr{Patch}(\tau)}_\Diamond \leq N e^{-\frac{R}\xi},
\end{equation}
under the time $\tau$ such that $0 \leq \tau \leq 1/(6e\xi g)$.
We approximate the time evolution $e^{\mcl{L}\tau}$ by the product of $\mcl{U}_\mr{Patch}(\tau/2)$ and $\overline{\mcl{U}}_\mr{Patch}(\tau/2)$ as shown in Fig. \ref{Fig_algorithm_generic} (b).
It is clear that its error bound is given by the following corollary.

\begin{corollary}
\textbf{}

Suppose that the time $\tau$ is small enough to satisfy Eq. (\ref{Eq_Patch:time_assumption}).
When the block size $R$ is larger than $\xi \max\{1,\log N\}$, the time evolution $e^{\mcl{L}\tau}$ is approximated by the map $\overline{\mcl{U}}_\mr{Patch}(\tau/2) \, \mcl{U}_\mr{Patch}(\tau/2)$ with an error bounded by
\begin{equation}\label{Eq_Gen:error_merge_strategy}
    \norm{e^{\mcl{L}\tau}-\overline{\mcl{U}}_\mr{Patch}(\tau/2) \, \mcl{U}_\mr{Patch}(\tau/2)}_\Diamond \leq 3N e^{-\frac{R}\xi}.
\end{equation}
\end{corollary}

\textbf{Proof.---}
Equation (\ref{Eq_Gen:error_Patch_bar}) implies $\norm{\overline{\mcl{U}}_\mr{Patch}(\tau)}_\Diamond \leq \norm{e^{\mcl{L}\tau}}_\Diamond + N e^{-R/\xi} \leq 2$.

The inequality Eq. (\ref{Eq_Gen:error_merge_strategy}) immediately follows from
\begin{eqnarray}
    && [\text{l.h.s of Eq. (\ref{Eq_Gen:error_merge_strategy})}] \nonumber \\
    && \quad \leq \norm{e^{\mcl{L}\tau/2}-\overline{\mcl{U}}_\mr{Patch}(\tau/2)}_\Diamond \norm{e^{\mcl{L}\tau/2}}_\Diamond \nonumber \\
    && \qquad \qquad + \norm{\overline{\mcl{U}}_\mr{Patch}(\tau/2)}_\Diamond \norm{e^{\mcl{L}\tau/2}-\mcl{U}_\mr{Patch}(\tau/2)}_\Diamond  \nonumber \\
    && \quad \leq 3N e^{-\frac{R}\xi}. \quad \square
\end{eqnarray}

The central strategy of merging is to regard the maps in the inner layers of $\overline{\mcl{U}}_\mr{Patch}(\tau/2) \, \mcl{U}_\mr{Patch}(\tau/2)$ as a single HP map like Fig. \ref{Fig_algorithm_generic} (b) and reproduce its expansion by quasi-probabilistic sampling.
Let us define the merged operator by
\begin{eqnarray}
    \mcl{U}_\mr{Merge}^\alpha(\tau) &=& e^{-(\mcl{L}_{B_{3\alpha-1}}+\mcl{L}_{B_{3\alpha}})\tau/2} e^{\mcl{L}_{B_{3\alpha-1}B_{3\alpha}}\tau} \nonumber \\
    && \qquad \qquad \times e^{-(\mcl{L}_{B_{3\alpha-1}}+\mcl{L}_{B_{3\alpha}})\tau/2}, \label{Eq_Gen:Merge_op}
\end{eqnarray}
which enables us to express the approximate time evolution as
\begin{eqnarray}
    && \overline{\mcl{U}}_\mr{Patch}(\tau/2) \, \mcl{U}_\mr{Patch}(\tau/2) \nonumber \\
    && \quad = \prod_\alpha e^{\mcl{L}_{B_{3\alpha}B_{3\alpha+1}B_{3\alpha+2}}\tau/2} \nonumber \\
    && \qquad \times \prod_\alpha \mcl{U}_\mr{Merge}^\alpha (\tau) \prod_\alpha e^{\mcl{L}_{B_{3\alpha}B_{3\alpha+1}B_{3\alpha+2}}\tau/2}. \label{Eq_Gen:Patch_w_Merge}
\end{eqnarray}
We execute quasi-probabilistic sampling for each of $\mcl{U}_\mr{Merge}^\alpha(\tau)$, which contains the non-CP components of $\overline{\mcl{U}}_\mr{Patch}(\tau/2) \, \mcl{U}_\mr{Patch}(\tau/2)$.
Suppression of the sampling overhead compared to Eq. (\ref{Eq_Gen:overhead_wo_merge}) is attributed to the suppression of its norm, $\norm{\mcl{U}_\mr{Merge}^\alpha(\tau)}_\text{Pauli}$ as follows.
The form of the merged operator $\mcl{U}_\mr{Merge}^\alpha(\tau)$ given by Eq. (\ref{Eq_Gen:Merge_op}) coincides with the second-order PF, Eq. (\ref{Eq_Gen:2nd_PF}), in which we set $\mcl{L}_1 = - (\mcl{L}_{B_{3\alpha-1}}+\mcl{L}_{B_{3\alpha}})$ and $\mcl{L}_2 = \mcl{L}_{B_{3\alpha-1}B_{3\alpha}}$.
Thus, the merged operator $\mcl{U}_\mr{Merge}^\alpha(\tau)$ approximates the time evolution under $\mcl{L}_{B_{3\alpha-1}:B_{3\alpha}}= \mcl{L}_{B_{3\alpha-1}B_{3\alpha}} - \mcl{L}_{B_{3\alpha-1}}-\mcl{L}_{B_{3\alpha}}$ by
\begin{equation}\label{Eq_Gen:Merge_op_PF}
    \mcl{U}_\mr{Merge}^\alpha(\tau) = e^{\mcl{L}_{B_{3\alpha-1}:B_{3\alpha}}\tau} \left[ 1 + \order{\tau^3} \right].
\end{equation}
To be precise, we prove the following theorem:

\begin{theorem}\label{Thm:Merge_expansion}
\textbf{}

Suppose that the time $\tau$ satisfies Eq. (\ref{Eq_Patch:time_assumption}).
The merged operator $\mcl{U}_\mr{Merge}^\alpha(\tau)$, defined by Eq. (\ref{Eq_Gen:Merge_op}), is expressed by
\begin{equation}\label{Eq_Gen:Merge_form}
    \mcl{U}_\mr{Merge}^\alpha(\tau) =  e^{\mcl{L}_{B_{3\alpha-1}:B_{3\alpha}}\tau} \left[ 1 + \mcl{A}_\mr{Merge}^\alpha (\tau)\right],
\end{equation}
where the HP map $\mcl{A}_\mr{Merge}^\alpha (\tau)$ is bounded by
\begin{equation}
    \norm{\mcl{A}_\mr{Merge}^\alpha (\tau)}_\text{Pauli} \leq 43 (\xi g\tau)^3 \in \order{\tau^3}.
\end{equation}
\end{theorem}
We provide the proof of the above theorem in Appendix \ref{SecA:Merged_block}.
We implement the Lindbladian dynamics $e^{\mcl{L}_{B_{3\alpha-1}:B_{3\alpha}}\tau}$ by the LCU-based approach \cite{cleve-2016-open,Li-Wang-2022-open}, and reproduce the remaining non-CP part $1+\order{\tau^3}$ by quasi-probabilistic sampling following Lemma \ref{LemmaA:quasiprobabilistic}. 
The sampling overhead for reproducing the dynamics over the evolution time $\tau$ via $\overline{\mcl{U}}_\mr{Patch}(\tau/2) \, \mcl{U}_\mr{Patch}(\tau/2)$ amounts to
\begin{equation}
    \order{\left[ 1+ 43(\xi g\tau)^3 \right]^{2N_\mr{p}}} \subset e^{\order{N(g\tau)^3/R_\mr{p}}}.
\end{equation}
It is clearly improved in the time $\tau$ compared to Eq. (\ref{Eq_Gen:overhead_wo_merge}).
The total sampling overhead for simulating the time $t = r_t \tau$ is as large as 
\begin{equation}\label{Eq_Gen:total_overhead}
\left[ e^{\order{N(g\tau)^3/R_\mr{p}}}\right]^{r_t} = \exp \left( \order{\frac{N(gt)^3}{R_\mr{p}(r_t)^2}}\right).
\end{equation}
In contrast to the case without merging, the sampling overhead can be reduced by increasing the number of steps $r_t$ while keeping the block size $R_\mr{p}$ small.
This leads to the substantial reduction in the gate count as we will discuss in Section \ref{Subsec:Cost_algo}.

We next discuss what kind of quantum gates should be sampled and implemented for reproducing the merged operator.
We first note that details of the sampled quantum circuits are provided in Appendix \ref{SecA:Merged_block} and hereby provide its brief description because the resulting gate count in this part is at most $\order{\polylog{Nt/\varepsilon}}$ and negligibly small compared to the other parts. 
The quantum gates that need to be sampled come from $\mcl{A}_\mr{Merge}^\alpha(\tau) \in \order{\tau^3}$ in Eq. (\ref{Eq_Gen:Merge_form}).
Its explicit form can be derived by the theory of Trotter errors \cite{childs2021-trotter}, since it is the second-order PF.
Applying the Dyson series expansion, we obtain
\begin{eqnarray}
    && \mcl{A}_\mr{Merge}^\alpha (\tau) \nonumber \\
    && \quad = e^{-\mcl{L}_{B_{3\alpha-1}:B_{3\alpha}}\tau} \mcl{U}_\mr{Merge}^\alpha(\tau) -1 \nonumber \\
    && \quad = \sum_{m=1}^\infty \int_0^\tau \dd \tau_m \cdots \int_0^{\tau_2} \dd \tau_1 \prod_{m'=m}^1 \Delta_\mr{Merge}^\alpha (\tau_{m'}), \nonumber \\
    && \label{Eq_Gen:Merge_Dyson}
\end{eqnarray}
where the operator $\Delta_\mr{Merge}^\alpha(\tau)$ is defined by
\begin{eqnarray}
    \Delta_\mr{Merge}^\alpha(\tau) &=& \mcl{U}_\mr{Merge}^\alpha(\tau)^{-1} \dv{\tau} \mcl{U}_\mr{Merge}^\alpha(\tau) \nonumber \\
    && \, - \mcl{U}_\mr{Merge}^\alpha(\tau)^{-1} \mcl{L}_{B_{3\alpha-1}:B_{3\alpha}} \mcl{U}_\mr{Merge}^\alpha(\tau). \nonumber \\
    && \label{Eq_Gen:Delta_alpha}
\end{eqnarray}
For the sampled quantum gates to be finite, we have to set a truncation order on the Dyson series expansion, Eq. (\ref{Eq_Gen:Merge_Dyson}).
We also introduce the truncation to the matrix exponentials in $\mcl{U}_\mr{Merge}^\alpha(\tau)$, Eq. (\ref{Eq_Gen:Merge_op}), in Eq. (\ref{Eq_Gen:Delta_alpha}).
The truncation orders for them are determined so that the error caused by them can be bounded by a preferable value $\epsilon \in (0,1)$.
We find the truncated version of the merged operator, which is suitable for sampling, as follows.

\begin{theorem}\label{Thm:Merge_truncation}
\textbf{}

Suppose that the time $\tau \in \order{1}$ satisfies Eq. (\ref{Eq_Patch:time_assumption}).
For any fixed $\epsilon \in (0,1)$, there exists a map $\tilde{\mcl{U}}_\mr{Merge}^\alpha(\tau)$ satisfying the following conditions:
\begin{enumerate}
    \item The merged operator is approximated by $\tilde{\mcl{U}}_\mr{Merge}^\alpha(\tau)$ with an error bounded by
    \begin{equation}\label{Eq_Gen:Merge_error_epsilon}
        \norm{\tilde{\mcl{U}}_\mr{Merge}^\alpha(\tau)-\mcl{U}_\mr{Merge}^\alpha(\tau)}_\Diamond \leq \epsilon.
    \end{equation}
    \item The map $\tilde{\mcl{U}}_\mr{Merge}^\alpha(\tau)$ is written in the form of
    \begin{equation}\label{Eq_Gen:Merge_truncate}
        \qquad \tilde{\mcl{U}}_\mr{Merge}^\alpha(\tau) = e^{\mcl{L}_{B_{3\alpha-1}:B_{3\alpha}}\tau} \left( 1 + \tilde{\mcl{A}}_\mr{Merge}^\alpha(\tau) \right),
    \end{equation}
    The non-CP part $\tilde{\mcl{A}}_\mr{Merge}^\alpha(\tau)$ is an HP map whose Pauli norm is bounded by
    \begin{equation}\label{Eq_Gen:Merge_truncate_norm}
        \qquad \norm{\tilde{\mcl{A}}_\mr{Merge}^\alpha(\tau)}_\text{Pauli} \leq 43 (\xi g\tau)^3 \in \order{\tau^3},
    \end{equation}
    and has the locality $\order{[\log (1/\epsilon)]^2}$.
\end{enumerate}
\end{theorem}

We briefly explain its proof, while we give the detailed proof in Appendix \ref{SecA:Merged_block}.
In the proof, we evaluate the error by the truncation based on Lemma \ref{Lem_Patch:Commutator_bound} and Theorem \ref{Thm:Merge_expansion}, and show that choosing the truncation orders $\order{\log (1/\epsilon)}$ respectively for Eqs. (\ref{Eq_Gen:Merge_Dyson}) and (\ref{Eq_Gen:Delta_alpha}) suffices to achieve the error $\epsilon$ as Eq. (\ref{Eq_Gen:Merge_error_epsilon}).
The Pauli norm of the non-CP part $\tilde{\mcl{A}}_\mr{Merge}^\alpha(\tau)$ can be bounded in the same way as Theorem \ref{Thm:Merge_expansion}.
The truncation orders $\order{\log (1/\epsilon)}$ give the locality $\order{[\log (1/\epsilon)]^2}$.

The above constructive proof tells us how we can efficiently execute the quasi-probabilistic sampling for reproducing the merged operator.
We set $\epsilon=\varepsilon/\poly{Nt}$, as we will do so in the algorithm. 
Due to the small truncation order up to $\order{\log (1/\epsilon)} = \order{\log(Nt/\varepsilon)}$, we can efficiently expand the truncated version of $\Delta_\mr{Merge}^\alpha(\tau)$ in Pauli operators by classical computation.
Since the form Eq. (\ref{Eq_Gen:Merge_Dyson}) resembles the Dyson series, the quasi-probabilistic sampling can be efficiently done in a similar manner to the randomized simulation of time-dependent Hamiltonians \cite{Zhang-2022-time_dep}.
To be concrete, we randomly sample the order $m \in \order{\log(Nt/\varepsilon)}$ and the time $\tau_1,\cdots,\tau_m \in [0,\tau]$. 
Then, we pick up one $\order{\log(Nt/\varepsilon)}$-local operation respectively from each of $\Delta_\mr{Merge}^\alpha(\tau_1)$, $\Delta_\mr{Merge}^\alpha(\tau_2)$, and $\Delta_\mr{Merge}^\alpha(\tau_m)$, and sequentially apply the operations (See Appendix \ref{SecA:Merged_block} for details).
As a result, we need at most $\order{[\log(Nt/\varepsilon)]^2}$ quantum gates at each sample for reproducing the non-CP part of $\mcl{U}_\mr{Merge}^\alpha(\tau)$.

\subsection{Algorithm and cost}\label{Subsec:Cost_algo}

\begin{figure}[t]
\begin{algorithm}[H] 
	\caption{Observable estimation for generic lattice Lindbladian dynamics}
	\label{Algorithm_Generic}
	\begin{algorithmic}[1]
	\REQUIRE Initial state $\rho$, evolution time $t$, target error $\varepsilon$, observable $O=\sum_x o_x \ket{o_x}\bra{o_x}$, the set of local operators that make up the Lindbladian $\mcl{L}$.
\ENSURE Quantity $O_\rho(t)$ satisfying $|O_\rho(t)-\mr{Tr}[Oe^{\mcl{L}t}\rho]| \le \varepsilon$ with constant probability.
\STATE Set block size $R_\mr{p} \leftarrow \lceil \mr{Const.} \times (Ngt)^{1/3} \rceil$, step count $r_t \leftarrow \lceil  \sqrt{N (gt)^3/R_\mr{p}} \rceil$, step size $\tau \leftarrow t / r_t$, block size $R \leftarrow \lceil \xi \log (4Nr_t/\varepsilon) \rceil$, sample complexity $S \leftarrow \Theta (\varepsilon^{-2})$.
\STATE Construct partition $\{B_\alpha\}$ according to Eqs.~(\ref{Eq_Gen:B_partition_1})--(\ref{Eq_Gen:B_partition_3}).
\STATE Set $O_\rho(t)$ to $0$.
\FOR{step $s = 1$ \TO $S$}
\STATE Prepare state $\rho$.
\FOR{step $r' = 1$ \TO $r_t$}
    \STATE Apply dissipative channel $\bigotimes_\alpha e^{\mathcal{L}_{B_{3\alpha}B_{3\alpha+1}B_{3\alpha+2}}\tau/2}$ via the LCU-based Lindbladian simulation.
    \STATE Sample quantum circuit $Q_\alpha$ having a gate count of $\order{[\log (Nt/\varepsilon)]^2}$ for each merged block $\mcl{U}_\mr{Merge}^\alpha(\tau)$, and apply the quantum circuit $\bigotimes_\alpha Q_\alpha$.
    
    \STATE Apply dissipative channel $\bigotimes_\alpha e^{\mcl{L}_{B_{3\alpha-1}:B_{3\alpha}} \tau}$ via the LCU-based Lindbladian simulation.
    \STATE Apply dissipative channel $\bigotimes_\alpha e^{\mathcal{L}_{B_{3\alpha}B_{3\alpha+1}B_{3\alpha+2}}\tau/2}$ via the LCU-based Lindbladian simulation.
\ENDFOR
\STATE Measure the state with the basis $\{ \ket{o_x} \}_x$ and add $S^{-1}o_x$ multiplied by the prefactor dependent on the sampled circuits $\{Q_\alpha\}$ to $O_\rho(t)$.
\ENDFOR
\STATE \textbf{return} Final quantity $O_\rho(t)$.
	\end{algorithmic}
\end{algorithm}
\end{figure}

In this section, we describe the algorithm and analyze its cost.
The algorithm runs with quasi-probabilistic sampling of quantum circuits and returns the time-evolved observable $\mr{Tr}[O e^{\mcl{L}t}\rho]$ for $O=\sum_i o_i \ket{o_i}\bra{o_i}$.
We approximate the target time evolution $e^{\mcl{L}t}$ by $[\overline{\mcl{U}}_\mr{Patch}(\tau/2) \mcl{U}_\mr{Patch}(\tau/2)]^{r_t}$ with $\tau = t/r_t$.
The map $\overline{\mcl{U}}_\mr{Patch}(\tau/2) \mcl{U}_\mr{Patch}(\tau/2)$ is composed of $e^{\mcl{L}_{B_{3\alpha}B_{3\alpha+1}B_{3\alpha+2}}\tau/2}$ and $\mcl{U}_\mr{Merge}^\alpha(\tau)$ as shown in Eq. (\ref{Eq_Gen:Patch_w_Merge}).
As shown in Fig. \ref{Fig_algorithm_generic} (c), we construct a quantum circuit, a part of which is randomly sampled, in the following way:
\begin{itemize} 
    \item Implementation of $e^{\mcl{L}_{B_{3\alpha}B_{3\alpha+1}B_{3\alpha+2}}\tau/2}$:
    We run the LCU-based algorithm for Lindbladian dynamics \cite{Li-Wang-2022-open} within an error $\order{\varepsilon/(Nr_t)}$.   
    
    \item Implementation of $\mcl{U}_\mr{Merge}^\alpha(\tau)$:
    We approximate it by $\tilde{\mcl{U}}_\mr{Merge}^\alpha(\tau)$ within an error $\order{\varepsilon/(Nr_t)}$ based on Theorem \ref{Thm:Merge_truncation}.
    We run the LCU-based algorithm for $e^{\mcl{L}_{B_{3\alpha-1}:B_{3\alpha}}\tau}$ within an error $\order{\varepsilon/(Nr_t)}$ .
    In order to reproduce the remaining part in Eq. (\ref{Eq_Gen:Merge_truncate}), we sample and apply a quantum circuit while recording it for the prefactor used in classical post-processing.
\end{itemize}
We apply the above sampled quantum circuit to the initial state $\rho$, and make projective measurement with the basis $\{ \ket{o_i}\bra{o_i} \}$.
We compute the average of the observable after classically post-processing the observed value $o_i$ with the recorded indices for the sampled maps based on the way of quasi-probabilistic sampling.
We obtain $\mr{Tr}[O e^{\mcl{L}t}\rho]$ under sufficient sampling complexity.
The algorithm is formally provided as Algorithm \ref{Algorithm_Generic}.
When we properly choose the block sizes $R$, $R_\mr{p}$, and the number of time steps $r_t$, Algorithm \ref{Algorithm_Generic} enables us to efficiently obtain the time-evolved observable with the reasonable sampling complexity as follows.

\begin{theorem*}
\textbf{(Restatement of Theorem \ref{Thm_Sparse:generic})}

Consider a one-dimensional lattice Lindbladian with finite-range interactions and dissipation.
We set the block sizes $R$ and $R_\mr{p}$ respectively by
\begin{equation}\label{Eq_Gen:block_size}
    R \in \Theta (\log (Nt/\varepsilon)), \quad R_\mr{p} \in \Theta \left( (Nt)^{\frac13} \right).
\end{equation}
We set the number of time steps $r_t$ by
\begin{equation}\label{Eq_Gen:Trotter_number}
    r_t \in \Theta \left( t \left( Nt\right)^{\frac13} \right).
\end{equation}
Then, Algorithm \ref{Algorithm_Generic} gives an estimate of the time-evolved observable  $\mr{Tr}[O e^{\mcl{L}t}\rho]$ within an error $\varepsilon$, running with the following computational resources:
\begin{itemize}
    \item Number of $\order{1}$-qubit gates per sample: It amounts to
    \begin{equation}\label{Eq_Gen:Gate_counts}
        \order{(Nt)^{\frac43} \, \polylog{Nt/\varepsilon}}.
    \end{equation}

    \item Ancilla qubit number and circuit depth: 
    The algorithm requires $\Theta (\polylog{Nt/\varepsilon})$ ancilla qubits, and then, the circuit depth amounts to
    \begin{equation}\label{Eq_Gen:complexity_wo_parallel}
       \order{(Nt)^{\frac43} \, \polylog{Nt/\varepsilon}}.
    \end{equation}
    When $\tilde{\Theta} (N^{2/3})$ ancilla qubits are available, the circuit depth amounts to
    \begin{equation}\label{Eq_Gen:complexity_parallel}
        \order{t (Nt)^{\frac23} \, \polylog{Nt/\varepsilon}}.
    \end{equation}
    \item Sampling complexity: It amounts to $\order{\varepsilon^{-2}}$.
\end{itemize}

\end{theorem*}

\textbf{Remark.---}
The gate count $\Otilde{(Nt)^{4/3}}$ becomes as large as or smaller than that of the LCU-based algorithm, $\Otilde{N^2t}$, under the time scale $t \in \order{N^2}$.
Thus, we focus on the time scale $t \in \order{N^2}$.
The block size $R_\mr{p}$ should satisfy $R_\mr{p} \in \order{N}$ so that Eq. (\ref{Eq_Gen:Patch_op}) can work as decomposition into smaller blocks.
The choice of $R_\mr{p}$ by Eq. (\ref{Eq_Gen:block_size}) is always available under $t \in \order{N^2}$.

\textbf{Proof.---}
The block size $R$ is determined so that
\begin{equation}
    \norm{e^{\mcl{L}t}- \left[ \overline{\mcl{U}}_\mr{Patch}(\tau/2) \, \mcl{U}_\mr{Patch}(\tau/2) \right]^{r_t}}_\Diamond \in \order{\varepsilon}
\end{equation}
can be satisfied.
Corollary \ref{Cor_Sparse:Patching_strategy} immediately implies that the choice by $R = \mr{Const.} \times \xi \log (Nr_t/\varepsilon)$ is sufficient.
Since the complexity $r_t$ is at most polynomial in $N$, $t$, and $1/\varepsilon$ [as we will confirm it as Eq. (\ref{Eq_Gen:Trotter_number})], we have $R \in \Theta (\log (Nt/\varepsilon))$.

We next determine the number of time steps $r_t=t/\tau$ based on the sampling overhead.
As discussed in Eq. (\ref{Eq_Gen:total_overhead}), the sampling overhead due to the quasi-probabilistic sampling for $\mcl{U}_\mr{Merge}^\alpha(\tau)$ amounts to $e^{\order{N(gt)^3/[R_\mr{p}(r_t)^2]}}$.
We choose the number of time steps $r_t$ by
\begin{equation}\label{Eq_Gen:Trotter_number_Np}
    r_t = \left\lceil \sqrt{\frac{N (gt)^3}{R_\mr{p}}} \right\rceil \in \Theta \left( gt \sqrt{\frac{Ngt}{R_\mr{p}}}\right),
\end{equation}
so that the overhead can be $\order{1}$.
The sampling complexity with this choice reduces to $\order{\varepsilon^{-2}}$.

We consider the gate count per sample.
Each component in $\overline{\mcl{U}}_\mr{Patch}(\tau/2) \, \mcl{U}_\mr{Patch}(\tau/2)$, expressed by Eq. (\ref{Eq_Gen:Patch_w_Merge}), is accurately implemented by the following cost.
\begin{itemize}
    \item Implementation of $e^{\mcl{L}_{B_{3\alpha}B_{3\alpha+1}B_{3\alpha+2}}\tau/2}$:
    We run the LCU-based algorithm for Lindbladian dynamics \cite{Li-Wang-2022-open} within an error $\varepsilon/(Nr_t)$.
    Applying Eq. (\ref{Eq_Setup:Gate_LCU}) for the block size $R_\mr{p}+2R \in \Theta (R_\mr{p})$ simply results in the scaling $\order{(R_\mr{p})^2 g\tau \, \polylog{R_\mr{p}g\tau Nr_t/\varepsilon}} = \Otilde{(R_\mr{p})^2 g\tau}$.
    However, we note that the query complexity $\tilde{\Theta}(\norm{\mcl{L}_{B_{3\alpha}B_{3\alpha+1}B_{3\alpha+2}}}_\mr{Pauli}\tau/2) = \tilde{\Theta}(R_\mr{p}\tau)$ cannot be smaller than $1$ even for the small time $\tau=t/r_t$.
    The gate count for the block should be at least $\tilde{\Theta}(R_\mr{p})$, which corresponds to the cost of the block encoding and the additional local gates per query.
    Therefore, we use
    \begin{equation}
        \qquad \quad \Otilde{\max \left( R_\mr{p}, (R_\mr{p})^2\tau \right)} \subset  \Otilde{R_\mr{p}+(R_\mr{p})^2\frac{gt}{r_t}} \label{Eq_Gen:dominant_blocks}
    \end{equation}
    quantum gates in this step.

    \item Implementation of $e^{\mcl{L}_{B_{3\alpha-1}:B_{3\alpha}}\tau}$ in $\mcl{U}_\mr{Merge}^\alpha(\tau)$:
    We run the LCU-based algorithm for Lindbladian dynamics.
    The support size of the inter-block terms $\mcl{L}_{B_{3\alpha-1}:B_{3\alpha}}$ is at most $2\xi \in \order{1}$.
    In a similar manner to Eq. (\ref{Eq_Gen:dominant_blocks}), the gate count is as large as
    \begin{eqnarray}
        && \Otilde{\max (\xi, (\xi)^2 g\tau)} \subset \order{\polylog {Ngt/\varepsilon}},
    \end{eqnarray}
    where we use the assumption on $\tau$, Eq. (\ref{Eq_Patch:time_assumption}).

    \item Implementation of  $\mcl{A}_\mr{Merge}^\alpha(\tau)$ in $\mcl{U}_\mr{Merge}^\alpha(\tau)$ by quasi-probabilistic sampling:
    We randomly sample quantum circuits on each size-$R$ block.
    The gate count of each sampled circuit is at most $\order{[\log (Nt/\varepsilon)]^2}$, reflecting its locality by Theorem \ref{Thm:Merge_truncation} [See also Appendix \ref{SubsecA:Merge_implement}].
    
\end{itemize}

The map $\overline{\mcl{U}}_\mr{Patch}(\tau/2) \mcl{U}_\mr{Patch}(\tau/2)$ expressed by Eq. (\ref{Eq_Gen:Patch_w_Merge}) contains $\Theta (N_\mr{p}) = \Theta (N/R_\mr{p})$ copies of $e^{\mcl{L}_{B_{3\alpha}B_{3\alpha+1}B_{3\alpha+2}}\tau/2}$ and $\mcl{U}_\mr{Merge}^\alpha(\tau)$.
Among them, the gate count for the former map given by Eq. (\ref{Eq_Gen:dominant_blocks}) is dominant.
The total gate count for $r_t$ steps scales as
\begin{eqnarray}
    && \Otilde{r_t \times \frac{N}{R_\mr{p}}\times \left[ R_\mr{p}+(R_\mr{p})^2\frac{gt}{r_t}\right]} \nonumber \\
    && \quad \qquad \qquad = \Otilde{Ngt \left[ \sqrt{\frac{Ngt}{R_\mr{p}}} + R_\mr{p} \right]}, \label{Eq_Gen:Gate_count_derivation}
\end{eqnarray}
where we substitute Eq. (\ref{Eq_Gen:Trotter_number_Np}) for $r_t$.
We minimize the above scaling by adjusting the block size $R_\mr{p}$.
This is achieved by setting
\begin{equation}\label{Eq_Gen:block_size_optimized}
    R_\mr{p} = \left\lceil (Ngt)^{\frac13} \right\rceil \in \Theta \left( (Ngt)^{\frac13}\right).
\end{equation}
The number of the blocks $N_\mr{p}$ scales as $N_\mr{p} \in \order{N/R_\mr{p}} \subset \order{N^{2/3}(gt)^{-1/3}}$.
Equation (\ref{Eq_Gen:Gate_count_derivation}) is equal to $\Otilde{(Ngt)^{4/3}}$, which gives the gate count in Eq. (\ref{Eq_Gen:Gate_counts}).

The relation between the number of ancilla qubits $n_a$ and the circuit depth $G$ can be evaluated in a similar manner to Theorem \ref{Thm_Setup:sparse}.
The number of ancilla qubits for each block time evolution $e^{\mcl{L}_{B_{3\alpha}B_{3\alpha+1}B_{3\alpha+2}}\tau/2}$ is at most $\order{\polylog{R r_tN/\varepsilon}} = \order{\polylog{Nt/\varepsilon}}$.
The one for each $e^{\mcl{L}_{B_{3\alpha-1}:B_{3\alpha}}\tau}$ is also $\order{\polylog{Nt/\varepsilon}}$.
The quasi-probabilistic sampling uses one ancilla qubit for each block.
When we run the algorithm using minimal ancilla qubits without parallelization, the circuit depth is as large as the gate count, which results in Eq. (\ref{Eq_Gen:complexity_wo_parallel}).
On the other hand, when we have $\Otilde{N_\mr{p}} \subset \Otilde{N^{2/3}}$ ancilla qubits, all the block time evolutions can be implemented in parallel.
The circuit depth under parallelization is as large as
\begin{equation}
    \Otilde{\frac{(Ngt)^{\frac43}}{N_\mr{p}}} = \order{gt (Ngt)^{\frac23} \, \polylog{Nt/\varepsilon}},
\end{equation}
which corresponds to Eq. (\ref{Eq_Gen:complexity_parallel}).

Before completing the proof, we have to check whether the time $\tau$ satisfies Eq. (\ref{Eq_Patch:time_assumption}). It is necessary for confirming that Corollary \ref{Cor_Sparse:Patching_strategy} is available for the algorithm construction.
Under the choice of $r_t$ and $R_\mr{p}$ respectively by Eqs. (\ref{Eq_Gen:Trotter_number_Np}) and (\ref{Eq_Gen:block_size_optimized}), the renormalized time $g\tau=gt/r_t$ scales as
\begin{equation}
    g\tau \in \Theta \left( (Ngt)^{-\frac13}\right).
\end{equation}
Thus, the assumption Eq. (\ref{Eq_Patch:time_assumption}) is satisfied for the sufficiently large size $N$ or time $t$. $\quad \square$

\subsection{Comparison with the existing algorithms}

We compare the computational cost in Theorem \ref{Thm_Sparse:generic} with those of the existing algorithms for generic lattice Lindbladians with finite-range interactions and dissipation.
We evaluate the costs of other algorithms by setting the number of local terms in the Hamiltonian $H$, the number of Lindblad operators $M$ respectively to $\order{N}$.
The number of Pauli operators in each Lindblad operator $L_{X,m}$ in Eq. (\ref{Eq_Setup:Lm_Pauli_expansion}) is an $\order{1}$ constant.
We also note that the definitions of the norm $\norm{\mcl{L}}$ are different among the references.
We replace them by $\norm{\mcl{L}}_\mr{Pauli} \in \order{N}$ since their scalings are common for generic finite-range interactions and dissipation.
The algorithms discussed below have the sampling complexity $\Theta (\varepsilon^{-2})$ or at most $\tilde{\Theta}(\varepsilon^{-2})$ to estimate time-evolved observables within an error $\varepsilon$.

We summarize the comparison with the standard algorithms in Table \ref{Table:Gate_counts}.
First, the second-order PF \cite{Kliesch-prl2011-open,ChildsLi2017Sparse,Wang-2026-open} achieves the scaling $\order{(Nt)^{3/2}/\varepsilon^{1/2}}$, which is preferable in the size $N$.
The randomized compiling of PF, known as qDRIFT, yields a gate count of $\order{(\norm{\mcl{L}}_\mr{Pauli}t)^2/\varepsilon}\subset \order{(Nt)^2/\varepsilon}$ \cite{Campbell-prl2019-qdrift,Chen2025RandomizedLindblad,David_2026_qdrift}.
Our algorithm outperforms in any of $N$, $t$, and $\varepsilon$, particularly achieving exponential improvement in $1/\varepsilon$.
Concerning the LCU-based approach, Cleve and Wang (2016) \cite{cleve-2016-open} established an algorithm based on the Taylor expansion of $e^{\mcl{L}t}$, whose gate count is $\Otilde{M^2 N^2 \norm{\mcl{L}}_\mr{Pauli} t} \subset \Otilde{N^5t}$.
Later, Li and Wang (2022) \cite{Li-Wang-2022-open} developed an alternative LCU-based algorithm using higher-order expansion by Duhamel's principle, which is used as a subroutine in our algorithms.
Its gate count amounts to $\Otilde{(N+M) \norm{\mcl{L}}_\mr{Pauli} t} = \Otilde{N^2 t}$ as shown in Table \ref{Table:Gate_counts}.
Our algorithm has the smaller exponent of the size $N$ by $2/3$, while having the larger exponent of the time $t$ by $1/3$.
Our algorithm shares the same scaling up to a polylogarithmic factor for $t \in \Theta(N^2)$, and becomes advantageous in $t \in o(N^2)$.
Recently, Ding et al. (2024) \cite{Ding-Li-Lin-2024-open} have established an algorithm that reproduces Lindbladian dynamics from a certain parent Hamiltonian dynamics.
This algorithm yields a gate count of 
\begin{eqnarray}
    && \order{ (N+M^p) \left[ \left( \norm{ \mcl{L}}_\mr{Pauli}t\right)^{1+\frac1p}\varepsilon^{-\frac1p}\right]} \nonumber \\
    && \qquad \qquad \qquad \quad = \order{N^{1+p+\frac1p}t^{1+\frac1p}\varepsilon^{-\frac1p}},
\end{eqnarray}
which consists of the number of the terms in the parent Hamiltonian, and the queries to the parent Hamiltonian dynamics.
While our algorithm can be outperformed in its dependence on time by this algorithm with $p > 3$, our algorithm significantly improves the dependency both in the size $N$ and the error $\varepsilon$.
We note that the above algorithms allow us to simulate both time-evolved states and observables, and that our algorithm can outperform them in the latter task.

We also discuss some quantum algorithms for estimating time-evolved observables.
Kato et al. (2026) \cite{kato-wada-2024_open} have constructed an algorithm using a linear combination of super-operators (LCS) and randomized sampling of dissipation, yielding a gate count of $\Otilde{(\norm{\mcl{L}}_\mr{Pauli}t)^2} \subset \Otilde{(Nt)^2}$.
Yu et al. (2025) \cite{Yu-Yuan-2025-open} have independently developed an approach based on LCS, which works with a gate count of $\Otilde{M (\norm{\mcl{L}}_\mr{Pauli}t)^2} \subset \Otilde{N^3t^2}$ and fewer ancilla qubits.
Our algorithm outperforms these algorithms both in $N$ and $t$ while retaining the polylogarithmic dependence in $1/\varepsilon$.
This comes from the fact that our algorithm fully exploits the locality, while they are not limited to lattice Lindbladians with finite-range interactions and dissipation.
Recently, Wang et al. (2026) \cite{Wang-2026-open} have proven that the extrapolation of the second-order PF achieves a gate count of $\Otilde{(Nt)^{3/2}}$, which has been the best among the known algorithms so far to the best of our knowledge.
Our algorithm based on patching and merging outperforms the above algorithms, achieving a gate count of $\Otilde{(Nt)^{4/3}}$.

\section{Conclusion and discussion}\label{Sec:Conclusion}

In this paper, we consider quantum algorithms for simulating the dynamics under lattice Lindbladians with finite-ranged interactions and dissipation.
With the locality-based techniques called ``patching" and ``merging", we decompose the Lindbladian dynamics into those of small blocks and substantially suppress the sampling overhead.
Our algorithms achieve the near-optimal gate count for sparsely-dissipative systems, and also achieve the one with the best size-dependency among the existing algorithms for simulating time-evolved observables under generic lattice Lindbladians.
Our approach will shed light on simulation of novel nonequilibrium phenomena and various algorithms using dissipative state preparation.

Finally, we provide possible extensions of our algorithms at present, and leave some open problems to be addressed as follows.

\subsubsection{Possible extensions}

\textit{Time-dependent systems.---}
Time-dependent systems are also of interest as a series of algorithms have been developed for Hamiltonian simulation \cite{Suzuki1993-general,Wiebe2010-mu,mizuta2024arxiv-time-dep-PF,Low2018-dyson,Kieferova2019-dyson,Mizuta_Quantum_2023,Mizuta_2023_multi,Watkins-2024-time-dep,chen2026-time_dep}.
While we focus on time-independent Lindbladians $\mcl{L}$ here, the extensions to time-dependent Lindbladians $\mcl{L}(t)$ are straightforward.
When we define the local properties of time-dependent lattice Lindbladians like Section \ref{Subsec:Setup}, the time-dependent analogue of the patching Lemma, Theorem \ref{Thm_Patch:Patching_lemma}, is proven completely in a similar manner.
The decompositions of the time-evolution operator in Algorithm \ref{Algorithm_Sparse} and Algorithm \ref{Algorithm_Generic} are valid.
The HHKL algorithm for Hamiltonian simulation \cite{Haah2021-hhkl} and the LCU-based approach for Lindbladian simulation \cite{Li-Wang-2022-open} are both available for time-dependent systems when we properly assume the smoothness.
As a result, Theorems \ref{Thm_Setup:sparse} and \ref{Thm_Sparse:generic} hold also for time-dependent Lindbladians.
Namely, when the time-dependent dissipation is sparsely located, we can achieve the optimal gate count $\Otilde{Nt}$.
We can also achieve the gate count $\Otilde{(Nt)^{4/3}}$ for time-dependent Lindbladians with generic finite-ranged interactions and dissipation.
Time-dependent Lindbladians are of central interest as dissipative quantum many-body systems under quantum control.
In addition, they may also be a clue to establishing an efficient algorithm for simulating time-independent Lindbladians via the interaction picture.

\textit{Quasi-local Lindbladians.---}
Similarly, our algorithms can be extended to generic lattice Lindbladians with quasi-local interactions and dissipation, where their strength decays exponentially in the distance.
The patching lemma can be applied also to quasi-local Lindbladians, reminiscent of the Lieb-Robinson bound.
This also implies that we can introduce the cutoff $\order{\log (Nt/\varepsilon)}$ on the range of the interactions and dissipation.
As a result, both Algorithm \ref{Algorithm_Sparse} and Algorithm \ref{Algorithm_Generic} work well also for quasi-local Lindbladians when the access to each coefficient can be efficiently done.
Quasi-local Lindbladians are of importance in the context of preparing quantum Gibbs states \cite{Chen2025Efficient,Lin-2025-open-review}.
Our algorithms may be useful for this purpose.

\textit{High-dimensional systems.---}
Another important direction is the extension to high-dimensional systems.
We mainly focus on one-dimensional lattice systems, and the decompositions in Algorithm \ref{Algorithm_Sparse} and Algorithm \ref{Algorithm_Generic} seem to strongly rely on the one-dimensionality.
We discuss the extensions of our algorithms to high-dimensional systems in Appendix \ref{SecA:High_dim}.
For high-dimensional systems in which the dissipation is sparsely located as Definition \ref{Def_Setup:sparse_dissipation}, a decomposition similar to Fig. \ref{Fig_patch_sparse} (b) gives a near-optimal algorithm running with $\Otilde{Nt}$ gates as well.
For Lindbladians with generic finite-ranged interactions and dissipation, the same strategy as Algorithm \ref{Algorithm_Generic} is valid.
After repeating the decomposition like Fig. \ref{Fig_algorithm_generic} (b), we merge the time-evolution operators of the boundary blocks.
We can achieve the gate counts, $\Otilde{(Nt)^{3/2-1/(4d+2)}}$, for $d$-dimensional systems, as shown in Theorem \ref{ThmA_High:generic} in Appendix \ref{SecA:High_dim}.
Although this cost is worse than that of one-dimensional systems due to the growing size of the boundaries, it achieves the best size-$N$ dependency among the known algorithms (i.e., better than $\Otilde{(Nt)^{3/2}}$, achieved by the extrapolation of the second-order PF \cite{Wang-2026-open}).

\subsubsection{Some open problems}

\textit{Algorithm with the optimal gate count}.---
The near-optimal gate count $\Otilde{Nt}$ is achieved only for some classes of Lindbladians.
One is a system with mutually-commuting Hermitian dissipation under an assumption on efficient access to QRAM \cite{Yu-Cirac-2025-open}.
The other is the sparsely-dissipative system, which includes boundary driven systems, in our results.
For generic lattice Lindbladians with finite-ranged interactions and dissipation, our algorithm achieves the gate count $\Otilde{(Nt)^{4/3}}$, which is the best size-$N$ dependency among the known algorithms while retaining the polylogarithmic dependency in the inverse error $1/\varepsilon$.
In Hamiltonian simulation, the HHKL algorithm achieves the gate count $\Otilde{Nt}$ for generic finite-ranged interactions \cite{Haah2021-hhkl}.
It is important to clarify whether the near-optimal gate count $\Otilde{Nt}$ can be achieved only for some limited Lindbladians, or how we can construct the near-optimal algorithm for generic Lindbladians if one exists.
     
\textit{Lindbladians with long-ranged interactions and dissipation}.---
Our algorithms rely on the assumption that each of the interactions and dissipation is finite-ranged (or geometrically-local), which is natural in many models in condensed matter physics.
On the other hand, the long-range interactions whose strength decays polynomially in distance are also of interest.
For instance, the HHKL algorithm using patching can be partially extended to Hamiltonians with long-range interactions \cite{Tran-PRX2019-hhkl}, although the dependence on $1/\varepsilon$ becomes polynomial.
It is natural to ask how we can efficiently simulate long-ranged Lindbladians, as well as Hamiltonian simulation.

\textit{Relation to fast-forwarding}.---
In general, Hamiltonian simulation for the evolution time $t$ requires the computational cost at least proportional to $t$, which is known as no fast forwarding.
However, some specific cases in non-unitary dynamics \cite{Jennings_2024-forward,An_2026-forward} or Lindbladian dynamics \cite{Shang2025-forward,Gao2025-forward,Shang2026-forward} allow the fast forwarding, i.e., the simulation with $o(t)$ cost.
It will be important to investigate whether the decomposition by locality plays a central role in faster algorithms for simulating such fast-forwardable dynamics.

\section*{Note added}

Very recently, a query-optimal algorithm for Lindbladian simulation has been developed \cite{wang2026-open-optimal, chen2026-open-optimal}.
Using the transducer approach, it achieves the optimal additive query complexity $\order{\norm{\mcl{L}}_\mr{BE}t+\frac{\log (1/\varepsilon)}{\log(e+(\norm{\mcl{L}}_\mr{BE}t)^{-1}\log(1/\varepsilon))}}$ in the block-encodings of the Hamiltonian and Lindblad operators, where the norm $\norm{\mcl{L}}_\mr{BE}$ is as large as $\norm{\mcl{L}}_\mr{Pauli} \in \order{N}$ for generic lattice Lindbladians.
However, it does not mean the optimality in gate counts for lattice Lindbladians.
With a standard explicit implementation of the lattice block encodings,
using $\order{N}$ elementary gates per oracle call and normalization $\order{N}$, the resulting gate-count bound is $\Otilde{N^2t}$.
It is as large as the LCU-based approach \cite{Li-Wang-2022-open} for lattice Lindbladians.
Algorithm \ref{Algorithm_Sparse}, achieving the near-optimal gate count $\Otilde{Nt}$, outperforms it for sparsely dissipative systems.
For simulating time-evolved observables under generic one-dimensional lattice Lindbladians, Algorithm \ref{Algorithm_Generic} requires $\Otilde{(Nt)^{4/3}}$, achieving the better dependence on the size $N$.
This holds also for high-dimensional systems (See Appendix \ref{SecA:High_dim}).

\section*{Statement of AI use}
The research ideas, original proofs, and initial manuscript were developed entirely by the author. AI tools were used to assist with proofreading, checking references to prior work, and identifying potential issues in the mathematical definitions and proofs. All revisions were reviewed and finalized by the author, who takes full responsibility for the content of the manuscript.

\textbf{}

\section*{Acknowledgment}

K. M. thanks Kazuki Sakamoto and Yuki Ito for fruitful discussions.
K. M. is supported by JST PRESTO Grant No. JPMJPR235A and JSPS KAKENHI Grant No. JP24K16974.
This work was supported by JST [Moonshot R\&D] [Grant Number JPMJMS256J].
\bibliography{bibliography.bib}

\clearpage
\onecolumngrid
\appendix

\section{Some basic facts used for the proof}\label{SecA:Basic}

\renewcommand{\thetheorem}{\thesection\arabic{theorem}}
\setcounter{theorem}{0}

\subsection{Quasi-probabilistic sampling}\label{SubsecA:Quasi_probabilistic}

Here, we briefly review the exact protocol of the quasi-probabilistic sampling used in our algorithms.
A generic linear super-operator $\mcl{A}$ on $N$-qubit states can be expanded by 
\begin{equation}
    \mcl{A}\rho = \sum_{\mu,\nu} \gamma_{\mu\nu} P_\mu \rho P_\nu, \quad \gamma_{\mu\nu} \in \bbC
\end{equation}
with $N$-qubit Pauli matrices $\{P_\mu\}$.
Here, we suppose that $\mcl{A}$ is HP, i.e., $\gamma_{\nu\mu}=\gamma_{\mu\nu}^\ast$ for every $\mu,\nu$.
We define the Pauli norm by
\begin{equation}
    \norm{\mcl{A}}_\mr{Pauli} = \sum_{\mu,\nu} |\gamma_{\mu\nu}|.
\end{equation}
Then, we consider the problem of calculating the expectation value $\mr{Tr}[O(1+\mcl{A})\rho]$ within an additive error $\varepsilon$, while the map $1+\mcl{A}$ is not necessarily CP.
This can be executed by quasi-probabilistic sampling of quantum circuits with some sampling overhead as follows:

\begin{lemma}\label{LemmaA:quasiprobabilistic}
\textbf{(Quasi-probabilistic sampling)}

Let $\mcl{A}$ be an HP map and $O$ be an observable such that $\norm{O} \leq 1$.
We prepare one ancilla qubit labeled by $a$, and consider a block-encoding $Q_{\mu\nu}$, defined by
\begin{equation}\label{EqA_Basic:sample_block_encoding}
    Q_{\mu\nu} = (\mr{Had}_a \otimes I) \left( \ket{0}\bra{0}_a \otimes e^{i \arg (\gamma_{\mu\nu})} P_\mu + \ket{1}\bra{1}_a\otimes P_\nu\right) (\mr{Had}_a \otimes I)
\end{equation}
In each experiment, we prepare the state $\ket{0}\bra{0}_a \otimes \rho$, apply a unitary gate randomly sampled from $\{ I_a \otimes I, Q_{\mu\nu}, I_a \otimes P_\mu\}$.
We make projective measurement respectively on the ancilla system in the basis $\{ \ket{z}\bra{z} \}_{z=0,1}$ and on the target system in $\{\ket{o_x}\bra{o_x} \}_{x}$, where the latter is given by the spectral decomposition $O=\sum_x o_x \ket{o_x}\bra{o_x}$.
Then, we can obtain an estimate $\hat{o} \in \bbR$ such that $|\hat{o}-\mr{Tr}[O(1+\mcl{A})\rho]| \leq \varepsilon$  with constant probability by $\order{\frac{ (1+3\norm{\mcl{A}}_\mr{Pauli})^2}{\varepsilon^2}}$-times experiments.
\end{lemma}

\textbf{Remark.---} Using Hoeffding's inequality, the success probability can be larger than $1-\delta$ for arbitrary $\delta \in (0,1)$ with the sampling complexity $\order{\frac{ (1+3\norm{\mcl{A}}_\mr{Pauli})^2}{\varepsilon^2} \log (1/\delta)}$.

\textbf{Proof.---}
With the definition of $\mcl{A}$, we calculate $(1+\mcl{A})\rho$ as follows,
\begin{eqnarray}
    (1+\mcl{A})\rho &=& \rho + \frac12 \sum_{\mu,\nu} (\gamma_{\mu\nu} P_\mu \rho P_\nu + \gamma_{\mu\nu}^\ast P_\nu \rho P_\mu )\nonumber \\
    &=& \rho + \frac12 \sum_{\mu,\nu}|\gamma_{\mu\nu}|\left[ \left( e^{i \arg (\gamma_{\mu\nu})} P_\mu + P_\nu \right) \rho \left( e^{-i \arg (\gamma_{\mu\nu})} P_\mu + P_\nu \right) - P_\mu \rho P_\mu - P_\nu \rho P_\nu  \right] \nonumber \\
    &=& \rho + \sum_{\mu,\nu} 2 |\gamma_{\mu,\nu}| \braket{0|Q_{\mu\nu}|0}_a \rho \braket{0|Q_{\mu\nu}|0}_a^\dagger - \sum_\mu \left( \sum_\nu |\gamma_{\mu\nu}|\right) P_\mu \rho P_\mu.
\end{eqnarray}
We use the relation $\gamma_{\nu\mu}=\gamma_{\mu\nu}^\ast$, which comes from the HP property of $\mcl{A}$, in the first line.
We execute $M$ experiments.
In the $m$th experiment ($m =1,2,\cdots,M$), we prepare the state $\ket{0}\bra{0}_a \otimes \rho$, and randomly apply the unitary gates in the following way;
\begin{equation}\label{EqA_Basic:sampled_circuit}
    \begin{cases}
        I_a \otimes I & \left( \text{with probability $\frac{1}{1+3 \norm{A}_\mr{Pauli}}$} \right) \\
         Q_{\mu\nu} & \left( \text{with probability $\frac{2|\gamma_{\mu\nu}|}{1+3 \norm{A}_\mr{Pauli}}$} \right) \\
          I_a \otimes P_\mu & \left( \text{with probability $\sum_\nu \frac{|\gamma_{\mu\nu}|}{1+3 \norm{A}_\mr{Pauli}}$} \right)
    \end{cases}.
\end{equation}
After that, we make projective measurement respectively on the ancilla and target systems in the basis $\{\ket{z}\bra{z}_a \otimes \ket{o_x}\bra{o_x} \}$.
When we obtain the measurement outcomes $z_m =0,1$ and $(o_x)_m$, we return
\begin{equation}\label{EqA_Basic:estimator}
    \hat{o}_m = \begin{cases}
        (1+3\norm{\mcl{A}}_\mr{Pauli})(o_x)_m & (\text{if $I_a \otimes I$ is sampled}) \\
        (1+3\norm{\mcl{A}}_\mr{Pauli}) (1-z_m) (o_x)_m & (\text{if $Q_{\mu\nu}$ is sampled}) \\
        - (1+3\norm{\mcl{A}}_\mr{Pauli}) (o_x)_m & (\text{if $I_a \otimes P_\mu$ is sampled})
    \end{cases}.
\end{equation}

We give an estimate of $\mr{Tr}[O(1+\mcl{A})\rho]$ by $\hat{o}=\frac1M \sum_{m=1}^M \hat{o}_m$.
Equation (\ref{EqA_Basic:estimator}) immediately implies that its expectation value $\bbE[\hat{o}]$ is equal to $\mr{Tr}[O(1+\mcl{A})\rho]$.
In addition, its variance is bounded by
\begin{equation}
    \bbE [\hat{o}^2] - \bbE [\hat{o}]^2 = \frac1M \max_m(\bbE [\hat{o}_m^2] - \bbE [\hat{o}_m]^2)
    \leq \frac{(1+3\norm{\mcl{A}}_\mr{Pauli})^2}{M}.
\end{equation}
Using the Chebyshev inequality, we conclude that $M \in \order{\frac{ (1+3\norm{\mcl{A}}_\mr{Pauli})^2}{\varepsilon^2}}$ experiments are sufficient for determining $\mr{Tr}[O(1+\mcl{A})\rho]$ within an additive error $\varepsilon$. $\quad \square$

\textit{Gate count of the sampled quantum circuit.---}
Each sampled circuit has nontrivial part in either $Q_{\mu\nu}$ or $P_\mu$ according to Eq. (\ref{EqA_Basic:sampled_circuit}).
The latter is trivially implemented by $|\supp(P_\mu)|$ single-qubit gates.
The former one $Q_{\mu\nu}$ is the block-encoding given by Eq. (\ref{EqA_Basic:sample_block_encoding}).
It can be implemented by $2$ Hadamard gates and $(|\supp(P_\mu)|+|\supp(P_\nu)|)$ two-qubit gates.
Since both $|\supp(P_\mu)|$ and $|\supp(P_\nu)|$ are smaller than the locality $k(\mcl{A})$ defined by Eq. (\ref{Eq_Setup:k_def}), sampled quantum circuits have at most $\order{k(\mcl{A})}$ local gates, as discussed in Section \ref{Subsec:Setup}.

\textit{Sampling overhead when reproducing multiple non-CP maps.---}
In the algorithms, we calculate the observables after applying multiple non-CP maps to the state.
Consider the case where we wish to estimate $\mr{Tr}[O(1+\mcl{A}')(1+\mcl{A})\rho]$ for some HP maps $\mcl{A}$, $\mcl{A}'$ within an error $\varepsilon$.
When we sample quantum gates respectively from Pauli basis appearing in $\mcl{A}$, $\mcl{A}'$ and do classical postprocessing, the variance of the estimator is as large as $\order{(1+\norm{\mcl{A}}_\text{Pauli})^2(1+\norm{\mcl{A}'}_\text{Pauli})^2}$.
The sampling complexity amounts to $\order{(1+\norm{\mcl{A}}_\text{Pauli})^2(1+\norm{\mcl{A}'}_\text{Pauli})^2 \varepsilon^{-2}}$.
The sampling overhead when considering multiple non-CP maps is given by the product of the overheads for the non-CP maps as well.

\textit{Sampling overhead for reproducing backward time evolution.---}
Let us consider reproducing observables under the non-CP map $e^{-\mcl{L}\tau}$ for some time $\tau>0$ and Lindbladian $\mcl{L}$.
It is expanded by $e^{-\mcl{L}\tau}= 1 + \sum_{q=1}^\infty (-\mcl{L}\tau)^q/q!$.
The sampling overhead for reproducing $\mr{Tr}[O e^{-\mcl{L}_{A_\alpha}\tau} \rho]$ amounts to
\begin{eqnarray}
    \left( 1 + 3 \norm{\sum_{q=1}^\infty \frac{
    (-\mcl{L}_{A_\alpha}\tau)^q}{q!}}_\text{Pauli} \right)^2 &\leq& \left( 1+ 3 \sum_{q=1}^\infty \frac1{q!}(\norm{\mcl{L}_{A_\alpha}}_\text{Pauli} \tau)^q \right)^2 \leq e^{6 \norm{\mcl{L}_{A_\alpha}}_\text{Pauli} \tau} \in e^{\order{\norm{\mcl{L}_{A_\alpha}}_\text{Pauli} \tau}}.
\end{eqnarray}
When we simply apply the HHKL algorithm to Lindbladian simulation and independently execute quasi-probabilistic sampling for the backward evolutions, the sampling overhead becomes exponentially large in spacetime.
This is why the HHKL algorithm, which is near-optimal for Hamiltonian simulation, is invalid for Lindbladian simulation, as discussed in Section \ref{Sec:Patching}.

\subsection{Corollaries of the patching lemma}\label{SubsecA:Patching_corollary}

In this section, we prove Corollary \ref{Cor_Patch:HHKL_Lindbladian} and Corollary \ref{Cor_Sparse:Patching_strategy} in the main text.
Corollary \ref{Cor_Patch:HHKL_Lindbladian} gives the dissipative counterpart of the HHKL algorithms \cite{Haah2021-hhkl}, though it fails to be efficient in contrast to Hamiltonian simulation (See Section \ref{Sec:Patching}).
Corollary \ref{Cor_Sparse:Patching_strategy} is of particular importance, which is directly used for constructing the algorithms.
It allows us to decompose the time evolution into those of flexible-size blocks.
As a result, we can delete or suppress the increase of the sampling overhead by the non-CP maps respectively for sparsely-dissipative or generic lattice Lindbladians.
Corollary \ref{Cor_Sparse:Patching_strategy} is a generalized version of Corollary \ref{Cor_Patch:HHKL_Lindbladian}.
Namely, when we set $A_3 = A_5 = \cdots = \emptyset$ and set $|A_\alpha|=R$ for the other domains $\{A_\alpha\}$, it reduces to the latter.
Thus, it is sufficient to prove Corollary \ref{Cor_Sparse:Patching_strategy}.

Before giving the proof, we note that the repetition of the patching lemma is rather nontrivial for Lindbladian dynamics in contrast to Hamiltonian dynamics.
Let us consider the case where we further split the subsystem $B \cup C$ into $A'$, $B'$, and $C'$ after applying the patching lemma with $\Lambda = A \cup B \cup C$.
The sizes of the blocks $B$ and $B'$ are at least $R$.
For Hamiltonian dynamics, the repetition of the patching lemma, Eq. (\ref{Eq_Patch:Patching_Lemma_Hamiltonian}), immediately allows the decomposition of $e^{-iH\tau}$ into the time evolutions under $H_{AB}, H_B, H_{A'B'}, H_{B'}, H_{B'C'}$ as follows.
\begin{eqnarray}
    && \norm{e^{-iH\tau}-e^{-iH_{AB}\tau} e^{iH_B\tau} e^{-iH_{A'B'}\tau}e^{iH_{B'}\tau}e^{-iH_{B'C'}\tau}} \nonumber \\
    && \quad \leq \norm{e^{-iH\tau}-e^{-iH_{AB}\tau} e^{iH_B\tau} e^{-iH_{BC}\tau}} + \norm{e^{-iH_{AB}\tau} e^{iH_B\tau}} \times \norm{e^{-iH_{BC}\tau}-e^{-iH_{A'B'}\tau}e^{iH_{B'}\tau}e^{-iH_{B'C'}\tau}} \nonumber \\
    && \quad \leq 2 c e^{-\frac{R}\xi}.
\end{eqnarray}
We use the unitarity $\norm{e^{-iH_{AB}\tau} e^{iH_B\tau}} = 1 $ in the last line.
In contrast, similar calculation does not apply to Lindbladian dynamics since $\norm{e^{\mcl{L}_{AB}\tau}e^{-\mcl{L}_B\tau}}_\Diamond$ is not necessarily bounded by $1$.
Thus, instead of repeating Eq. (\ref{Eq_Patch:error}), we directly prove Corollary \ref{Cor_Sparse:Patching_strategy} like Theorem \ref{Thm_Patch:Patching_lemma} as follows.

\begin{corollary*}
\textbf{(Restatement of Corollary \ref{Cor_Sparse:Patching_strategy})}

We split the lattice $\Lambda$ into the blocks $\{ A_\alpha \}$ by
\begin{equation}
    A_\alpha = \Set{\sum_{\alpha' < \alpha} R_{\alpha'} + 1, \sum_{\alpha' < \alpha} R_{\alpha'} + 2, \cdots, \sum_{\alpha' \leq \alpha} R_{\alpha'} }.
\end{equation}
We assume that the size of the even-indexed blocks $R_{2\alpha}$ is at least $R$, satisfying
\begin{equation}\label{EqA_Basic:Range_Block_assumption}
    R \geq \xi \max \{ 1, \log N \}.
\end{equation}
We define the patch operator by
\begin{equation}
    \mcl{U}_\mr{Patch}(\tau) = \prod_{\alpha} e^{\mcl{L}_{A_{4\alpha-2}A_{4\alpha-1}A_{4\alpha}} \tau } \prod_{\alpha} e^{-\mcl{L}_{A_{2\alpha}}\tau} \prod_{\alpha} e^{\mcl{L}_{A_{4\alpha}A_{4\alpha+1}A_{4\alpha+2}}\tau}. \label{EqA_Basic:Patch_op}
\end{equation}
When the time $\tau$ is small enough to satisfy
\begin{equation}\label{EqA_Basic:tau_assumption_strategy}
    0 \leq \tau \leq \frac{1}{6e\xi g} \in \order{1},
\end{equation}
the exact time evolution $e^{\mcl{L}\tau}$ is approximated by $\mcl{U}_\mr{Patch}(\tau)$ with an error bounded by
\begin{equation}\label{Eq_Sparse:error_strategy}
    \norm{e^{\mcl{L}\tau} - \mcl{U}_\mr{Patch}(\tau)}_\Diamond \leq N e^{-\frac{R}\xi}.
\end{equation}

\end{corollary*}

\textbf{Proof.---}
The proof is essentially the same as that for the patching lemma, Theorem \ref{Thm_Patch:Patching_lemma}.
We define a map $\mcl{N}_\mr{Patch}(\tau)=e^{-\mcl{L}\tau} \mcl{U}_\mr{Patch}(\tau)$, and it is computed as follows,
\begin{eqnarray}
    \mcl{N}_\mr{Patch}(\tau) &=& \mcl{N}_\mr{Patch}(0) + \int_0^\tau \dd \tau' \dv{\tau'} \left( e^{-\mcl{L}\tau'}\prod_{\alpha} e^{\mcl{L}_{A_{4\alpha-2}A_{4\alpha-1}A_{4\alpha}} \tau'} \prod_{\alpha} e^{-\mcl{L}_{A_{2\alpha}}\tau'} \prod_{\alpha} e^{\mcl{L}_{A_{4\alpha}A_{4\alpha+1}A_{4\alpha+2}}\tau'}\right) \nonumber \\
    &=& 1 - \int_0^\tau \dd \tau' e^{-\mcl{L}\tau'} \left( \mcl{L}-\sum_\alpha \mcl{L}_{A_{4\alpha-2}A_{4\alpha-1}A_{4\alpha} } \right) \prod_{\alpha} e^{\mcl{L}_{A_{4\alpha-2}A_{4\alpha-1}A_{4\alpha}} \tau'} \prod_{\alpha} e^{-\mcl{L}_{A_{2\alpha}}\tau'} \prod_{\alpha} e^{\mcl{L}_{A_{4\alpha}A_{4\alpha+1}A_{4\alpha+2}}\tau'} \nonumber \\
    && \quad + \int_0^\tau \dd \tau' e^{-\mcl{L}\tau'} \prod_{\alpha} e^{\mcl{L}_{A_{4\alpha-2}A_{4\alpha-1}A_{4\alpha}} \tau'} \nonumber \\
    && \qquad \qquad \qquad \times \prod_{\alpha} e^{-\mcl{L}_{A_{2\alpha}}\tau'} \left( \sum_\alpha \mcl{L}_{A_{4\alpha}A_{4\alpha+1}A_{4\alpha+2}} - \sum_\alpha \mcl{L}_{A_{2\alpha}}\right) \prod_{\alpha} e^{\mcl{L}_{A_{4\alpha}A_{4\alpha+1}A_{4\alpha+2}}\tau'} . 
\end{eqnarray}
In the second equality, we use the fact that each pair of $\{\mcl{L}_{A_{4\alpha-2}A_{4\alpha-1}A_{4\alpha}}\}$, $\{\mcl{L}_{A_{2\alpha}}\}$, or $\{ \mcl{L}_{A_{4\alpha}A_{4\alpha+1}A_{4\alpha+2}}\}$ commutes with one another.
Since the size of the even-indexed blocks $A_\alpha$ is at least $\xi$ by Eq. (\ref{EqA_Basic:Range_Block_assumption}), we have
\begin{eqnarray}
    \mcl{L}-\sum_\alpha \mcl{L}_{A_{4\alpha-2}A_{4\alpha-1}A_{4\alpha} } &=& \sum_\alpha \mcl{L}_{A_{4\alpha}A_{4\alpha+1}A_{4\alpha+2}} - \sum_\alpha \mcl{L}_{A_{2\alpha}} \nonumber \\
    &=& \sum_\alpha \mcl{L}_{A_{4\alpha+1}} + \sum_\alpha \left( \mcl{L}_{A_{4\alpha}A_{4\alpha+1}A_{4\alpha+2}} -  \mcl{L}_{A_{4\alpha}} -  \mcl{L}_{A_{4\alpha+1}} -  \mcl{L}_{A_{4\alpha+2}} \right), 
\end{eqnarray}
according to Fig. \ref{Fig_patch_sparse} (a).
Let us define a map $\mcl{J}_\alpha$ by
\begin{equation}
    \mcl{J}_\alpha = \mcl{L}_{A_{4\alpha}A_{4\alpha+1}A_{4\alpha+2}} -  \mcl{L}_{A_{4\alpha}} -  \mcl{L}_{A_{4\alpha+1}} -  \mcl{L}_{A_{4\alpha+2}}.
\end{equation}
The support of $\mcl{J}_\alpha$ lies in the boundary of $A_{4\alpha}$ and $A_{4\alpha+1}$ or the one of $A_{4\alpha+1}$ and $A_{4\alpha+2}$.
Considering the range $\xi$, we have the relations,
\begin{equation}\label{EqA_Basic:support_boundaries}
    \supp (\mcl{J}_\alpha) \subset \bigcup_{\alpha'=0,1} \Set{ j \in \Lambda |   \min_{j' \in A_{4\alpha+\alpha'+1}} \left( j' \right) - \xi \leq j \leq \max_{j' \in A_{4\alpha+\alpha'}} \left( j' \right) + \xi},
\end{equation}
and $|\supp(\mcl{J}_\alpha)| \leq 4\xi$.
Since $\mcl{L}_{A_{4\alpha+1}}$ commutes with any of $\{\mcl{L}_{A_{4\alpha-2}A_{4\alpha-1}A_{4\alpha}}\}$ and $\{\mcl{L}_{A_{2\alpha}}\}$, the map $\mcl{N}_\mr{Patch}(\tau)$ can be expressed as 
\begin{eqnarray}
    \mcl{N}_\mr{Patch}(\tau) &=& 1 - \int_0^\tau \dd \tau' e^{-\mcl{L}\tau'} \left[ \sum_\alpha \mcl{J}_\alpha ,  \prod_{\alpha} e^{\mcl{L}_{A_{4\alpha-2}A_{4\alpha-1}A_{4\alpha}} \tau'} \prod_{\alpha} e^{-\mcl{L}_{A_{2\alpha}}\tau'} \right] \prod_{\alpha} e^{\mcl{L}_{A_{4\alpha}A_{4\alpha+1}A_{4\alpha+2}}\tau'} \nonumber \\
    &=& 1 + \int_0^\tau \dd \tau' \mcl{K}_\mr{Patch}(\tau') \mcl{N}_\mr{Patch}(\tau') \nonumber \\
    &=& \sum_{n=0}^\infty \int_0^\tau \dd \tau_n \cdots \int_0^{\tau_2} \dd \tau_1 \mcl{K}_\mr{Patch}(\tau_n) \cdots \mcl{K}_\mr{Patch}(\tau_1).
\end{eqnarray}
The generator $\mcl{K}_\mr{Patch}(\tau')$ is defined by
\begin{eqnarray}
    \mcl{K}_\mr{Patch}(\tau') &=& e^{-\tau'\ad_{\mcl{L}}} \left( \prod_{\alpha} e^{\tau' \ad_{\mcl{L}_{A_{4\alpha-2}A_{4\alpha-1}A_{4\alpha}}}} \prod_{\alpha} e^{-\tau'\ad_{\mcl{L}_{A_{2\alpha}}}} -1 \right) \sum_\alpha \mcl{J}_\alpha \nonumber \\
    &=& \sum_\alpha e^{-\tau'\ad_{\mcl{L}}} \left( e^{\tau' \ad_{\mcl{L}_{A_{4\alpha-2}A_{4\alpha-1}A_{4\alpha}} + \mcl{L}_{A_{4\alpha+2}A_{4\alpha+3}A_{4\alpha+4}}}}  e^{-\tau'\ad_{\mcl{L}_{A_{4\alpha}} + \mcl{L}_{A_{4\alpha+2}}}} -1 \right) \mcl{J}_\alpha. \label{EqA_Basic:K_s_Patch}
\end{eqnarray}
Using this map, the error can be bounded by
\begin{eqnarray}
    \norm{e^{\mcl{L}\tau}-\mcl{U}_\mr{Patch}(\tau)}_\Diamond \leq \norm{\mcl{N}_\mr{Patch}(\tau)-1}_\Diamond \leq \sum_{n=1}^\infty \frac{1}{n!} \left[ \sup_{\tau' \in [0,\tau]} \left( \norm{\mcl{K}_\mr{Patch}(\tau')}_\Diamond \right) \tau \right]^n, 
\end{eqnarray}
and hence it is sufficient to evaluate an upper bound on $\norm{\mcl{K}_\mr{Patch}(\tau')}_\Diamond$.

We evaluate Eq. (\ref{EqA_Basic:K_s_Patch}) in a similar manner to Eq. (\ref{Eq_Patch:K_s}).
The support of $\mcl{J}_\alpha$ lies around the boundary of $A_{4\alpha}$ and $A_{4\alpha+1}$ or the one of $A_{4\alpha+1}$ and $A_{4\alpha+2}$ within the distance $\xi$ as discussed in Eq. (\ref{EqA_Basic:support_boundaries}).
Under the assumption $R \geq \xi$, it is not connected to the domains $A_{4\alpha-2}$, $A_{4\alpha-1}$, $A_{4\alpha+3}$, or $A_{4\alpha+4}$.
Commutators with local terms are irrelevant until these domains can be connected to $A_{4\alpha}$ or $A_{4\alpha+2}$, which results in
\begin{equation}
    \left( \ad_{\mcl{L}_{A_{4\alpha-2}A_{4\alpha-1}A_{4\alpha}} + \mcl{L}_{A_{4\alpha+2}A_{4\alpha+3}A_{4\alpha+4}}} \right)^m  \left(\ad_{\mcl{L}_{A_{4\alpha}}+\mcl{L}_{A_{4\alpha+2}}} \right)^n \mcl{J}_\alpha = \left(\ad_{\mcl{L}_{A_{4\alpha}} + \mcl{L}_{A_{4\alpha+2}}} \right)^{m+n} \mcl{J}_\alpha,
\end{equation}
under $m+n < R/\xi -1 $ like Eq. (\ref{Eq_Patch:commutator_boundary}).
The corresponding terms appearing in the expansion of Eq. (\ref{EqA_Basic:K_s_Patch}) vanish as
\begin{equation}
    \sum_{\substack{m,n \geq 0: \\ 1 \leq m+n < R/\xi-1}} \frac{1}{m!n!} \left( \ad_{\mcl{L}_{A_{4\alpha-2}A_{4\alpha-1}A_{4\alpha}} + \mcl{L}_{A_{4\alpha+2}A_{4\alpha+3}A_{4\alpha+4}}}\right)^m \left( -\ad_{\mcl{L}_{A_{4\alpha}} + \mcl{L}_{A_{4\alpha+2}}}\right)^n \mcl{J}_\alpha = 0.
\end{equation}
For any $\tau' \in [0,\tau]$, we arrive at the relation,
\begin{eqnarray}
    \norm{\mcl{K}_\mr{Patch}(\tau')}_\Diamond \tau &\leq& \sum_\alpha \sum_{l=0}^\infty \sum_{\substack{m,n \geq 0: \\ \lceil \frac{R}\xi-1 \rceil \leq m+n}} \frac{(\tau')^{l+m+n}}{l!m!n!} \nonumber \\
    && \qquad \times \norm{(\ad_{\mcl{L}})^l\left(\ad_{\mcl{L}_{A_{4\alpha-2}A_{4\alpha-1}A_{4\alpha}} + \mcl{L}_{A_{4\alpha+2}A_{4\alpha+3}A_{4\alpha+4}}} \right)^m \left( \ad_{\mcl{L}_{A_{4\alpha}} + \mcl{L}_{A_{4\alpha+2}}} \right)^n \mcl{J}_\alpha}_\Diamond \tau \nonumber \\
    &\leq& \sum_\alpha  \sum_{l=0}^\infty \sum_{\substack{m,n \geq 0: \\ \lceil \frac{R}\xi-1 \rceil \leq m+n}} \frac{(\tau')^{l+m+n}}{l!m!n!}(l+m+n)! (2kg)^{l+m+n}\norm{\mcl{J}_\alpha}_\mr{Pauli} \tau \nonumber \\
    &\leq& \frac{N}4 \sum_{q=\lceil \frac{R}\xi -1 \rceil}^\infty \sum_{\substack{l,m,n \geq 0: \\ l+m+n = q}} \frac{q!}{l!m!n!} (2kg\tau')^q \times (4 \xi g \tau) \leq \frac{e}{6(e-1)} N e^{-\frac{R}\xi}.
\end{eqnarray}
The second inequality comes from the fact that $\mcl{L}_{A_{4\alpha-2}A_{4\alpha-1}A_{4\alpha}} + \mcl{L}_{A_{4\alpha+2}A_{4\alpha+3}A_{4\alpha+4}}$, $\mcl{L}_{A_{4\alpha}} + \mcl{L}_{A_{4\alpha+2}}$, and $\mcl{J}_\alpha$ are all at-most $k$-local and $g$-extensive, which enables us to use Lemma \ref{Lem_Patch:Commutator_bound}.
The third inequality comes from the fact that the number of $\{\mcl{J}_\alpha\}$ is at most $N/4$.
When we set the block size $R \geq \xi \log N$ as shown in Eq. (\ref{EqA_Basic:Range_Block_assumption}), this upper bound is smaller than $1$.
Finally, we obtain the error bound,
\begin{equation}
    \norm{e^{\mcl{L}\tau}-\mcl{U}_\mr{Patch}(\tau)}_\Diamond \leq \sum_{n=1}^\infty \frac1{n!} \left( \frac{e}{6(e-1)} N e^{-R/\xi}\right)^n \leq N e^{-R/\xi},
\end{equation}
which completes the proof. $\quad \square$
\section{Implementation of merged block}\label{SecA:Merged_block}

In this section, we provide the proofs of Theorem \ref{Thm:Merge_expansion} and Theorem \ref{Thm:Merge_truncation} on the properties of the merged block, and discuss how it can be reproduced by the quasi-probabilistic sampling.
Theorem \ref{Thm:Merge_expansion} states that the merged operator $\mcl{U}_\mr{Merge}^\alpha(\tau)$ defined by Eq. (\ref{Eq_Gen:Merge_op}) has a sufficiently small norm to avoid large overhead by quasi-probabilistic sampling.
This comes from the fact that $\mcl{U}_\mr{Merge}^\alpha(\tau)$ has the form of a second-order PF.
Theorem \ref{Thm:Merge_truncation} further ensures that it can be approximated by local operators, making the cost of each sampled circuit small.

We prove the theorems as follows.
First, we discuss PF errors for generic local and extensive linear maps in Appendix \ref{SubsecA:PF_error}.
We obtain the accurate description of generic PFs, which improves the result of Childs et al. (2021) \cite{childs2021-trotter} exponentially in spacetime.
We then apply it to the second-order PF for lattice Lindbladians and prove Theorem \ref{Thm:Merge_expansion}.
In Appendix \ref{SubsecA:Merge_truncation}, we consider the truncation of the merged operator for the sampled quantum circuit to be efficient.
We prove Theorem \ref{Thm:Merge_truncation} stating that $\order{[\log(1/\epsilon)]^2}$ local gates are sufficient to reproduce it within a sufficiently small error $\epsilon$.
We finally discuss the explicit implementation of the merged operator by quasiprobabilistic sampling in Appendix \ref{SubsecA:Merge_implement}.
For brevity, we omit the superscript $\alpha$ from $\mcl{U}_\mr{Merge}^\alpha(\tau)$ and other operators, and simply denote $B_{3\alpha-1}$, $B_{3\alpha}$ respectively as $B$, $B'$.
Namely, the merged operator defined by Eq. (\ref{Eq_Gen:Merge_op}) is written by
\begin{equation}\label{EqA_Merge:Merge_op}
    \mcl{U}_\mr{Merge}(\tau) = e^{-(\mcl{L}_B+\mcl{L}_{B'})\tau/2} e^{\mcl{L}_{BB'}\tau} e^{-(\mcl{L}_B+\mcl{L}_{B'})\tau/2}
\end{equation}
throughout this appendix.

\renewcommand{\thetheorem}{\thesection\arabic{theorem}}
\setcounter{theorem}{0}

\subsection{The product formula error for non-unitary dynamics and Theorem \ref{Thm:Merge_expansion}}\label{SubsecA:PF_error}

Here, we aim at proving Theorem \ref{Thm:Merge_expansion}. 
Since it forms the second-order PF, we begin with discussing the errors of generic PFs for local and extensive linear maps.

Suppose that a bounded linear operator $\mcl{A}$ is decomposed into 
\begin{equation}
    \mcl{A} = \sum_{\gamma=1}^\Gamma \mcl{A}_\gamma,
\end{equation}
where each $\mcl{A}_\gamma$ is also a bounded operator.
We consider the problem of approximating the non-unitary time evolution $e^{\mcl{A}\tau}$ with those under $\{\mcl{A}_\gamma\}$.
The $p$th-order product formula $\mcl{T}_p(\tau)$ is defined by a formula such that
\begin{equation}\label{EqA_Merge:generic_PF}
    \mcl{T}_p(\tau) = \prod_{v=1}^{V_p} \prod_{\gamma=1}^\Gamma e^{a_{v\gamma} \mcl{A}_{\pi_v(\gamma)}\tau}, \quad a_{v\gamma} \in \bbR,
\end{equation}
satisfying the order condition $e^{\mcl{A}\tau} = \mcl{T}_p(\tau) + \order{\tau^{p+1}}$ under $\tau \to 0$.
The number $V_p$ denotes a constant determined by $p$, and the symbol $\pi_v$ : $\{1,\cdots,\Gamma\} \to \{1,\cdots,\Gamma\}$ is the re-ordering of the index $\gamma$.
We set $|a_{v\gamma}| \leq 1$ without loss of generality.
For instance, the first- and second-order formulas are given by
\begin{equation}
    \mcl{T}_1(\tau) = \prod_{\gamma=1}^\Gamma e^{\mcl{A}_\gamma \tau}, \quad \mcl{T}_2(\tau) = \left(\prod_{\gamma=\Gamma}^{1} e^{\mcl{A}_\gamma \tau/2} \right) \left(\prod_{\gamma=1}^{\Gamma} e^{\mcl{A}_\gamma \tau/2} \right).
\end{equation}
To prove Theorem \ref{Thm:Merge_expansion}, we have to evaluate the difference $e^{\mcl{L}_{B:B'}\tau} - \mcl{U}_\mr{Merge} (\tau) \in \order{\tau^3}$, where the map $\mcl{U}_\mr{Merge} (\tau)$ gives the second-order PF.
This corresponds to the PF error $e^{\mcl{A}\tau}-\mcl{T}_p(\tau)$, and hence we discuss its expression.

The existing PF error for generic non-unitary dynamics is known to be bounded by
\begin{equation}\label{EqA_Merge:PF_error_Childs}
    \norm{e^{\mcl{A}\tau}-\mcl{T}_p(\tau)} \leq \mr{Const.} \times \alpha_{\mr{com},p} \tau^{p+1} e^{4 V_p \tau \sum_\gamma \norm{\mcl{A}_{\gamma}}}, \quad \alpha_{\mr{com},p} = \sum_{\gamma_0,\gamma_1,\cdots,\gamma_p=1}^\Gamma \norm{[\mcl{A}_{\gamma_p}, \cdots, [\mcl{A}_{\gamma_1}, \mcl{A}_{\gamma_0}]]}, 
\end{equation}
where the constant depends solely on $p$ and $\Gamma$ (See Theorem 6 in Ref. \cite{childs2021-trotter}).
The norm $\norm{\cdot}$ can be arbitrary.
However, this upper bound contains an exponentially-large factor in the $1$-norm $\sum_\gamma \norm{\mcl{A}_{\gamma}}$, and becomes exponentially large in the system size $N$.
We derive the alternative expression of the PF error for generic non-unitary dynamics and its upper bound, which exponentially improves Eq. (\ref{EqA_Merge:PF_error_Childs}) for dissipative dynamics as follows.

\begin{theorem}
\textbf{}

The $p$th-order PF error $e^{\mcl{A}\tau}-\mcl{T}_p(\tau) \in \order{\tau^{p+1}}$ is expanded by
\begin{equation}\label{EqA_Merge:PF_error_Dyson}
    e^{\mcl{A}\tau} - \mcl{T}_p(\tau) = - e^{\mcl{A}\tau} \sum_{m=1}^\infty \int_0^\tau \dd \tau_m \int_0^{\tau_m} \dd \tau_{m-1} \cdots \int_0^{\tau_2} \dd \tau_1 \Delta_p (\tau_1) \Delta_p (\tau_2) \cdots \Delta_p(\tau_m),
\end{equation}
where the operator $\Delta_p(\tau)$ is given by
\begin{eqnarray}
    \Delta_p(\tau) &=& \sum_{v=1}^{V_p} \sum_{\gamma=1}^\Gamma a_{v\gamma} \left( \prod_{v'=v-1}^1 \prod_{\gamma'=\Gamma}^1 e^{-a_{v'\gamma'}\tau \ad_{\mcl{A}_{\pi_{v'}(\gamma')}}} \right) \left(\prod_{\gamma'=\gamma-1}^1 e^{-a_{v\gamma'}\tau \ad_{\mcl{A}_{\pi_{v}(\gamma')}}} \right) \mcl{A}_{\pi_v(\gamma)} \nonumber \\
    && \qquad  \qquad \qquad\qquad \qquad\qquad \qquad\qquad \qquad \qquad \qquad- \sum_{\gamma=1}^\Gamma \left( \prod_{v'=V_p}^1 \prod_{\gamma'=\Gamma}^1 e^{-a_{v'\gamma'}\tau \ad_{\mcl{A}_{\pi_{v'}(\gamma')}}}  \right) \mcl{A}_\gamma.
    \label{EqA_Merge:generic_Delta_p}
\end{eqnarray}

\end{theorem}

\textbf{Proof.---}
We have
\begin{eqnarray}
    \dv{\tau} \left( e^{-\mcl{A}\tau} \mcl{T}_p(\tau) \right) &=& - e^{-\mcl{A}\tau} \mcl{A} \mcl{T}_p(\tau) + e^{-\mcl{A}\tau} \dv{\tau} \mcl{T}_p(\tau) \nonumber \\
    &=& e^{-\mcl{A}\tau} \mcl{T}_p(\tau) \left( \mcl{T}_p(\tau)^{-1}\dv{\tau} \mcl{T}_p(\tau) - \mcl{T}_p(\tau)^{-1} \mcl{A} \mcl{T}_p(\tau) \right) \nonumber \\
    &=& e^{-\mcl{A}\tau} \mcl{T}_p(\tau) \Delta_p(\tau),
\end{eqnarray}
where we define the operator $\Delta_p(\tau)$ by
\begin{equation}
    \Delta_p(\tau) = \mcl{T}_p(\tau)^{-1}\dv{\tau} \mcl{T}_p(\tau) - \mcl{T}_p(\tau)^{-1} \mcl{A} \mcl{T}_p(\tau).
\end{equation}
Substituting Eq. (\ref{EqA_Merge:generic_PF}) as $\mcl{T}_p(\tau)$, we obtain the expression, Eq. (\ref{EqA_Merge:generic_Delta_p}).
Since the operator $\Delta_p(\tau)$ is bounded, $e^{-\mcl{A}\tau} \mcl{T}_p(\tau)$ can be given by a Dyson series expansion,
\begin{equation}
    e^{-\mcl{A}\tau} \mcl{T}_p(\tau) = \sum_{m=0}^\infty \int_0^\tau \dd \tau_m \int_0^{\tau_m} \dd \tau_{m-1} \cdots  \int_0^{\tau_2} \dd \tau_1 \Delta_p (\tau_1) \Delta_p (\tau_2) \cdots \Delta_p(\tau_m).
\end{equation}
This immediately implies Eq. (\ref{EqA_Merge:PF_error_Dyson}). $\quad \square$

The form of the operator $\Delta_p(\tau)$ by Eq. (\ref{EqA_Merge:generic_Delta_p}) is essentially the same as an operator appearing in the PF error analyzed by Childs et al. (2021) \cite{childs2021-trotter}.
The order condition $e^{\mcl{A}\tau}-\mcl{T}_p(\tau) \in \order{\tau^{p+1}}$ indicates $\Delta_p(\tau) \in \order{\tau^p}$, i.e., the cancellation of the low-order terms $\Theta(1),\Theta(\tau), \cdots, \Theta (\tau^{p-1})$.
Namely, the operator $\Delta_p(\tau)$ can be expanded by a set of $(q+1)$-fold nested commutators with $q \geq p$ as follows,
\begin{eqnarray}
    \Delta_p(\tau) &=& \sum_{v=1}^{V_p} \sum_{\gamma=1}^\Gamma a_{v\gamma} \sum_{\substack{\{l_{v'\gamma'},l_{\gamma'}\}: \\ \sum l_{v'\gamma'}+\sum l_{\gamma'} \geq p}}\left( \prod_{v'=v-1}^1 \prod_{\gamma'=\Gamma}^1 \frac{\left(-a_{v'\gamma'}\tau \ad_{\mcl{A}_{\pi_{v'}(\gamma')}}\right)^{l_{v'\gamma'}}}{l_{v'\gamma'}!} \right) \left(\prod_{\gamma'=\gamma-1}^1 \frac{\left(-a_{v\gamma'}\tau \ad_{\mcl{A}_{\pi_{v}(\gamma')}}\right)^{l_{\gamma'}} }{l_{\gamma'}!}\right) \mcl{A}_{\pi_v(\gamma)} \nonumber \\
    && \qquad \qquad \qquad \qquad \qquad\qquad \qquad \qquad \qquad- \sum_{\gamma=1}^\Gamma \sum_{\substack{\{l_{v'\gamma'}\}: \\ \sum l_{v'\gamma'} \geq p}} \left( \prod_{v'=V_p}^1 \prod_{\gamma'=\Gamma}^1 \frac{\left(-a_{v'\gamma'}\tau \ad_{\mcl{A}_{\pi_{v'}(\gamma')}}\right)^{l_{v'\gamma'}}}{l_{v'\gamma'}!}  \right) \mcl{A}_\gamma.
\end{eqnarray}
We can use the commutator bound by the locality, Lemma \ref{Lem_Patch:Commutator_bound}, and obtain an exponentially better error bound compared to Eq. (\ref{EqA_Merge:PF_error_Childs}) as follows. 

\begin{lemma}\label{LemmaA:nonuniaty_PF_error}
\textbf{}

We adopt the norm satisfying submultiplicativity (i.e., $\norm{\mcl{A} \mcl{A}'} \leq \norm{\mcl{A}} \, \norm{\mcl{A}'}$ is satisfied for any $\mcl{A},\mcl{A}'$) and suppose that the exact time evolution $e^{\mcl{A}\tau}$ is dissipative in the sense $\norm{e^{\mcl{A}\tau}} \leq 1$.
We also assume that $\mcl{A}=\sum_{\gamma=1}^\Gamma \mcl{A}_\gamma$ is a $k$-local and $g$-extensive map such that
\begin{equation}\label{EqA_Merge:Commutator_bound}
    \norm{[\mcl{A}_{\gamma_q}, \cdots, [\mcl{A}_{\gamma_1}, \mcl{A}_{\gamma_0}]]} \leq q! (2kg)^q g \, |\supp(\mcl{A}_{\gamma_0})|
\end{equation}
is satisfied for every $\gamma_0,\cdots,\gamma_q \in \{1,\cdots,\Gamma\}$.
When the time $\tau$ is small enough to satisfy
\begin{equation}\label{EqA_Merge:tau_assumption_PF}
    0 \leq \tau \leq \frac{1}{4 V_p \Gamma kg}, \qquad 2 (2V_p\Gamma kg \tau)^{p+1} \, \max_\gamma \left( |\supp(\mcl{A}_\gamma)| \right) \leq 1,
\end{equation}
the PF error can be bounded by
\begin{equation}\label{EqA_Merge:generic_PF_error_bound}
    \norm{e^{\mcl{A}\tau}-\mcl{T}_p(\tau)} \leq 2(e-1) (2V_p\Gamma kg \tau)^{p+1} \, \max_\gamma \left( |\supp(\mcl{A}_\gamma)| \right) .
\end{equation}
\end{lemma}

\textbf{Remark.---}
The relation Eq. (\ref{EqA_Merge:Commutator_bound}) holds for generic local and extensive operators \cite{Kuwahara2016-yn}, while the definitions of the locality $k$ and the extensiveness $g$ are slightly different from those for lattice Lindbladians, Eqs. (\ref{Eq_Setup:k_def}) and (\ref{Eq_Setup:g_def}).
We also note that the assumption on the time $\tau$ by Eq. (\ref{EqA_Merge:tau_assumption_PF}) is reasonable.
When employing the PF, we set the time $\tau$ such that the PF error becomes smaller than the allowable error $\varepsilon \in (0,1)$.
We require the right hand side of Eq. (\ref{EqA_Merge:generic_PF_error_bound}) to be $\order{\varepsilon}$, and hence the assumption Eq. (\ref{EqA_Merge:tau_assumption_PF}) is satisfied.

\textbf{Proof.---}
Using the expression of the error by Eq. (\ref{EqA_Merge:PF_error_Dyson}), the error can be bounded by
\begin{eqnarray}
    \norm{e^{\mcl{A}\tau}-\mcl{T}_p(\tau)} &\leq& \norm{e^{\mcl{A}\tau}} \sum_{m=1}^\infty \frac1{m!} \left[ \sup_{\tau' \in [0,\tau]} \left( \norm{\Delta_p(\tau')}\right) \tau\right]^m.
\end{eqnarray}
The second term in Eq. (\ref{EqA_Merge:generic_Delta_p}) is bounded by
\begin{eqnarray}
    \norm{\sum_{\gamma=1}^\Gamma \sum_{\substack{\{l_{v'\gamma'}\}: \\ \sum l_{v'\gamma'} \geq p}} \left( \prod_{v'=V_p}^1 \prod_{\gamma'=\Gamma}^1 \frac{\left(-a_{v'\gamma'}\tau \ad_{\mcl{A}_{\pi_{v'}(\gamma')}}\right)^{l_{v'\gamma'}}}{l_{v'\gamma'}!}  \right) \mcl{A}_\gamma} &\leq& \sum_{q=p}^\infty \tau^q \sum_{\gamma=1}^\Gamma \sum_{\substack{\{l_{v'\gamma'}\}: \\ \sum l_{v'\gamma'} = q}} \norm{\left( \prod_{v'=V_p}^1 \prod_{\gamma'=\Gamma}^1 \frac{\left(\ad_{\mcl{A}_{\pi_{v'}(\gamma')}}\right)^{l_{v'\gamma'}}}{l_{v'\gamma'}!}  \right) \mcl{A}_\gamma} \nonumber \\
     &\leq& \sum_{q=p}^\infty \tau^q \sum_{\gamma=1}^\Gamma  \sum_{\substack{\{l_{v'\gamma'}\}: \\ \sum l_{v'\gamma'} = q}} \frac{q!(2kg)^q}{\prod_{v',\gamma'} l_{v'\gamma'} !} g \, |\supp(\mcl{A}_\gamma)| \nonumber \\
     &\leq& \sum_{q=p}^\infty (2V_p\Gamma kg \tau)^q g \, \sum_{\gamma=1}^\Gamma |\supp(\mcl{A}_\gamma)| \nonumber \\
     &\leq& 2 (2V_p\Gamma kg \tau)^p g \, \sum_{\gamma=1}^\Gamma |\supp(\mcl{A}_\gamma)|. \label{EqA_Merge:Delta_p_2nd_term}
\end{eqnarray}
We use the relation Eq. (\ref{EqA_Merge:Commutator_bound}) in the second line and use $2V_p\Gamma kg \tau \leq 1/2$ from the assumption Eq. (\ref{EqA_Merge:tau_assumption_PF}) in the last line.
The same calculation goes also for the first term in Eq. (\ref{EqA_Merge:generic_Delta_p}), giving $2 V_p(2V_p\Gamma kg \tau)^p g \, \sum_{\gamma=1}^\Gamma |\supp(\mcl{A}_\gamma)|$ as its upper bound.
Thus, we obtain an upper bound,
\begin{equation}
    \sup_{\tau' \in [0,\tau]} \left( \norm{\Delta_p(\tau')} \right) \, \tau \leq 2 (2V_p\Gamma kg \tau)^{p+1} \, \max_\gamma \left( |\supp(\mcl{A}_\gamma)| \right) \leq 1.
\end{equation}
As a result, we arrive at the error bound,
\begin{eqnarray}
    \norm{e^{\mcl{A}\tau}-\mcl{T}_p(\tau)} &\leq& \sup_{\tau' \in [0,\tau]} \left( \norm{\Delta_p(\tau')}\right) \tau \sum_{m=1}^\infty \frac1{m!} \nonumber \\
    &\leq& 2(e-1) (2V_p\Gamma kg \tau)^{p+1} \, \max_\gamma \left( |\supp(\mcl{A}_\gamma)| \right),
\end{eqnarray}
which completes the proof. $\quad \square$

We note that Lemma \ref{LemmaA:nonuniaty_PF_error} provides an exponentially better error bound for dissipative dynamics than Eq. (\ref{EqA_Merge:PF_error_Childs}), proven by Childs et al. (2021) \cite{childs2021-trotter}.
The error bound in Eq. (\ref{EqA_Merge:PF_error_Childs}) suffers from the exponentially large factor $e^{\Theta (V_p \tau \sum_\gamma \norm{\mcl{A}_{\gamma}})}$ for generic non-unitary dynamics, which comes from the norm of the exact time evolution $e^{\mcl{A}\tau}$ and the one of the approximate time evolution $\mcl{T}_p(\tau)$.
Although the former one can be deleted for dissipative dynamics satisfying $\norm{e^{\mcl{A}\tau}} \leq 1$, the contribution from the latter is still exponentially large.
By contrast, the error bound by Eq. (\ref{EqA_Merge:generic_PF_error_bound}) is simply proportional to the commutator bound Eq. (\ref{EqA_Merge:Commutator_bound}) without any exponentially large factor, under the reasonable short time assumption, Eq. (\ref{EqA_Merge:tau_assumption_PF}).

We go back to the discussion on the merged operator $\mcl{U}_\mr{Merge}(\tau)$ for Lindbladian dynamics.
We apply the above error analysis to express the difference $e^{\mcl{L}_{B:B'}\tau} - \mcl{U}_\mr{Merge} (\tau) \in \order{\tau^3}$, and obtain the following expression.

\begin{corollary*}
\textbf{(Restatement of Theorem \ref{Thm:Merge_expansion})}

Suppose that the time $\tau$ satisfies
\begin{equation}\label{EqA_Merge:time_assumption}
    0 \leq \tau \leq \frac1{6e \xi g} \in \order{1}.
\end{equation}
The merged operator $\mcl{U}_\mr{Merge}(\tau)$, defined by Eq. (\ref{EqA_Merge:Merge_op}), is expressed by
\begin{equation}\label{Eq_Gen:Merge_form}
    \mcl{U}_\mr{Merge}(\tau) =  e^{\mcl{L}_{B:B'}\tau} \left[ 1 + \mcl{A}_\mr{Merge} (\tau)\right],
\end{equation}
where an HP map $\mcl{A}_\mr{Merge} (\tau)$ is bounded by
\begin{equation}
    \norm{\mcl{A}_\mr{Merge} (\tau)}_\text{Pauli} \leq 43 (\xi g\tau)^3 \in \order{\tau^3}.
\end{equation}
\end{corollary*}

\textbf{Proof.---}
This immediately follows from the proof of Lemma \ref{LemmaA:nonuniaty_PF_error}, though we slightly modify it for our setup.
The merged operator $\mcl{U}_\mr{Merge}(\tau)$ is the second-order PF, where we have
\begin{equation}
    \mcl{A} = \mcl{L}_{B:B'}, \quad \mcl{A}_1 = - \mcl{L}_{B} - \mcl{L}_{B'}, \quad \mcl{A}_2 = \mcl{L}_{BB'}, \quad p = 2, \quad \Gamma =2, \quad V_p =2, \quad a_{v\gamma} = \frac12. 
\end{equation}
The operator $\mcl{A}_\mr{Merge}(\tau)$ is immediately given by
\begin{equation}\label{EqA_Merge:A_Merge_Dyson}
    \mcl{A}_\mr{Merge}(\tau) = \sum_{m=1}^\infty \int_0^\tau \dd \tau_m \int_0^{\tau_m} \dd \tau_{m-1} \cdots \int_0^{\tau_2} \dd \tau_1 \Delta_\mr{Merge} (\tau_1) \Delta_\mr{Merge} (\tau_2) \cdots \Delta_\mr{Merge}(\tau_m),
\end{equation}
where $\Delta_\mr{Merge}(\tau)$ is defined by
\begin{equation}
    \Delta_\mr{Merge}(\tau) = \mcl{U}_\mr{Merge}(\tau)^{-1} \dv{\tau} \mcl{U}_\mr{Merge}(\tau) - \mcl{U}_\mr{Merge}(\tau)^{-1} \mcl{L}_{B:B'} \mcl{U}_\mr{Merge}(\tau).
\end{equation}
Expanding this operator in $\tau$, we obtain the expression,
\begin{eqnarray}
    \Delta_\mr{Merge}(\tau) &=& \sum_{q=2}^\infty \frac{(-\tau)^q}{2}
\sum_{\substack{l_1,l_2,l_3 \geq 0: \\ l_1+l_2+l_3=q-1}}
\frac{(\ad_{\mcl{L}_{B}+\mcl{L}_{B'}})^{l_1} (\ad_{\mcl{L}_{BB'}})^{l_2} (\ad_{\mcl{L}_{B}+\mcl{L}_{B'}})^{l_3}}
{(-2)^{l_1+l_3}l_1!\,l_2!\,(l_3+1)!} \ad_{\mcl{L}_{B}+\mcl{L}_{B'}} \mcl{L}_{B:B'} \nonumber \\
&& \qquad \qquad\qquad\qquad - \sum_{q=2}^\infty \frac{(-\tau)^q}{2}\sum_{\substack{l_1,l_2 \geq 0: \\ l_1+l_2=q-1}}
\frac{(\ad_{\mcl{L}_{B}+\mcl{L}_{B'}})^{l_1} (\ad_{\mcl{L}_{BB'}})^{l_2}}
{(-2)^{l_1} l_1!\,(l_2+1)!} \ad_{\mcl{L}_{B}+\mcl{L}_{B'}} \mcl{L}_{B:B'}, \label{EqA_Merge:Delta_Merge_expansion}
\end{eqnarray}
which corresponds to Eq. (\ref{EqA_Merge:generic_Delta_p}).
We can use the commutator bound by Lemma \ref{Lem_Patch:Commutator_bound} as a counterpart of Eq. (\ref{EqA_Merge:Commutator_bound}).
We can slightly improve the calculation compared to Eq. (\ref{EqA_Merge:Delta_p_2nd_term}) for the merged operator.
Since the nested commutators always appear with $\mcl{L}_{B:B'}$, the quantity $|\supp (\mcl{A}_\gamma)|$ can be replaced by $|\supp(\mcl{L}_{B:B'})| \leq 2\xi$.
We also reflect the value $|a_{v\gamma}|=1/2$.
This results in the upper bound,
\begin{eqnarray}
    \sup_{\tau' \in [0,\tau]} \left( \norm{\Delta_\mr{Merge}(\tau')}_\text{Pauli} \right) \, \tau &\leq& \sum_{q=2}^\infty \frac{\tau^{q+1}}2 \left[ \sum_{\substack{l_1,l_2,l_3 \geq 0: \\ l_1+l_2+l_3=q-1}} \frac{q! 2^{-l_1-l_3}}{l_1!l_2!(l_3+1)!} + \sum_{\substack{l_1,l_2 \geq 0: \\ l_1+l_2=q-1}} \frac{q!2^{-l_1}}{l_1!(l_2+1)!}\right] (2kg)^qg \, |\supp(\mcl{L}_{B:B'})| \nonumber \\
    &\leq& \sum_{q=2}^\infty \frac{\tau^{q+1}}2 \left[ 2 \sum_{\substack{l_1,l_2,l_3 \geq 0: \\ l_1+l_2+l_3=q}} \frac{q! 2^{-l_1-l_3}}{l_1! l_2!l_3!} - \sum_{\substack{l_1,l_2 \geq 0: \\ l_1+l_2=q}} \frac{q!2^{-l_1}}{l_1!l_2!}\right] (2\xi g)^{q+1} \nonumber\\
    &\leq& \frac12 \sum_{q=2}^\infty (4 \xi g\tau)^{q+1} \leq \frac12 \, \frac{1}{1-\frac2{3e}} (4\xi g \tau)^3, \label{EqA_Merge:Delta_Merge_bound}
\end{eqnarray}
where we employ $6\xi g\tau \leq e^{-1}$ from the assumption.
We note that the above quantity is bounded from above by $1$ under the same assumption.
Using the relation, Eq. (\ref{EqA_Merge:A_Merge_Dyson}), we arrive at the upper bound,
\begin{eqnarray}
    \norm{\mcl{A}_\mr{Merge}(\tau)}_\text{Pauli} &\leq& \sum_{m=1}^\infty \frac1{m!} \left[ \sup_{\tau' \in [0,\tau]} \left( \norm{\Delta_\mr{Merge}(\tau')}_\text{Pauli}\right) \tau\right]^m\nonumber \\
    &\leq& \sup_{\tau' \in [0,\tau]} \left( \norm{\Delta_\mr{Merge}(\tau')}_\text{Pauli}\right) \tau \sum_{m=1}^{\infty} \frac1{m!}\left[\sup_{\tau' \in [0,\tau]} \left( \norm{\Delta_\mr{Merge}(\tau')}_\text{Pauli}\right) \tau  \right]^{m-1} \nonumber \\
    &\leq& \frac12 \, \frac{1}{1-\frac2{3e}} (4\xi g \tau)^3  \exp \left( \frac12 \, \frac{1}{1-\frac2{3e}} \left( \frac{2}{3e} \right)^3   \right) \nonumber \\
    &\leq& 43 (\xi g\tau)^3 \in \order{\tau^3}, \label{EqA_Merge:A_bound_proof}
\end{eqnarray}
which completes the proof. $\quad \square$

Theorem \ref{Thm:Merge_expansion} indicates that the non-CP part of the merged operator, expressed by $1+\mcl{A}_\mr{Merge}(\tau)$, can be reproduced by quasi-probabilistic sampling with the sampling overhead $1+\order{\tau^3}=e^{\order{\tau^3}}$.
This is smaller than the overhead for reproducing backward evolutions without merging, $e^{\order{R\tau}}$.

\subsection{Truncation of the merged block and Theorem \ref{Thm:Merge_truncation}}\label{SubsecA:Merge_truncation}

Next, we discuss the sampled quantum circuits required to reproduce the merged operator.
For this purpose, we have to truncate the infinite series in $\mcl{A}_\mr{Merge}(\tau)$ so that the sampled circuits can be identified efficiently by classical computation and implemented by local quantum gates.
In this section, we provide the proof of Theorem \ref{Thm:Merge_truncation}, which ensures the approximation of the merged operator with the truncated one composed of local operators.

The non-CP part of the merged operator, $\mcl{A}_\mr{Merge}(\tau)$, contains infinite series in Eqs. (\ref{EqA_Merge:A_Merge_Dyson}) and (\ref{EqA_Merge:Delta_Merge_expansion}).
We define its truncated version by introducing the truncation orders $M_\mr{d} \geq 1$ and $Q_\mr{d} \geq 2$ as follows,
\begin{eqnarray}
    \mcl{A}_\mr{Merge}^{M_\mr{d},Q_\mr{d}}(\tau) &=& \sum_{m=1}^{M_\mr{d}} \int_0^\tau \dd \tau_m \int_0^{\tau_m} \dd \tau_{m-1} \cdots \int_0^{\tau_2} \dd \tau_1 \Delta_\mr{Merge}^{Q_\mr{d}} (\tau_1) \Delta_\mr{Merge}^{Q_\mr{d}} (\tau_2) \cdots \Delta_\mr{Merge}^{Q_\mr{d}}(\tau_m), \label{EqA_Merge:A_Md_Qd}\\
    \Delta_\mr{Merge}^{Q_\mr{d}}(\tau) &=& \sum_{q=2}^{Q_\mr{d}} \frac{(-\tau)^q}{2}
\sum_{\substack{l_1,l_2,l_3 \geq 0: \\ l_1+l_2+l_3=q-1}}
\frac{(\ad_{\mcl{L}_{B}+\mcl{L}_{B'}})^{l_1} (\ad_{\mcl{L}_{BB'}})^{l_2} (\ad_{\mcl{L}_{B}+\mcl{L}_{B'}})^{l_3}}
{(-2)^{l_1+l_3}l_1!\,l_2!\,(l_3+1)!} \ad_{\mcl{L}_{B}+\mcl{L}_{B'}} \mcl{L}_{B:B'} \nonumber \\
&& \qquad \qquad\qquad\qquad - \sum_{q=2}^{Q_\mr{d}} \frac{(-\tau)^q}{2}\sum_{\substack{l_1,l_2 \geq 0: \\ l_1+l_2=q-1}}
\frac{(\ad_{\mcl{L}_{B}+\mcl{L}_{B'}})^{l_1} (\ad_{\mcl{L}_{BB'}})^{l_2}}
{(-2)^{l_1} l_1!\,(l_2+1)!} \ad_{\mcl{L}_{B}+\mcl{L}_{B'}} \mcl{L}_{B:B'}. \label{EqA_Merge:Delta_Qd}
\end{eqnarray}
The following theorem ensures the approximation by the truncated version and its locality.

\begin{theorem*}
\textbf{(Restatement of Theorem \ref{Thm:Merge_truncation})}

Suppose that the time $\tau \in \order{1}$ satisfies Eq. (\ref{EqA_Merge:time_assumption}).
We define the truncated merged operator by
\begin{equation}
    \mcl{U}_\mr{Merge}^{M_\mr{d},Q_\mr{d}}(\tau) = e^{\mcl{L}_{B:B'}\tau} \left( 1+\mcl{A}_\mr{Merge}^{M_\mr{d},Q_\mr{d}}(\tau)\right).
\end{equation}
For any fixed $\epsilon \in (0,1)$, it satisfies the following conditions when we properly choose the truncation orders $M_\mr{d}, Q_\mr{d} \in \Theta (\log (1/\epsilon))$:
\begin{enumerate}
    \item The merged operator is approximated by $\mcl{U}_\mr{Merge}^{M_\mr{d},Q_\mr{d}}(\tau)$ with an error bounded by
    \begin{equation}\label{EqA_Merge:Merge_allowable_error}
    \norm{\mcl{U}_\mr{Merge}(\tau)-\mcl{U}_\mr{Merge}^{M_\mr{d},Q_\mr{d}}(\tau)}_\Diamond \leq \epsilon.
\end{equation}
    \item The map $\mcl{A}_\mr{Merge}^{M_\mr{d},Q_\mr{d}}(\tau)$ is an HP map, whose Pauli norm is bounded by
     \begin{equation}\label{EqA_Merge:A_Md_Qd_bound}
        \norm{\mcl{A}_\mr{Merge}^{M_\mr{d},Q_\mr{d}}(\tau)}_\text{Pauli} \leq 43 (\xi g\tau)^3 \in \order{\tau^3}.
    \end{equation}
    In addition, the locality of $\mcl{A}_\mr{Merge}^{M_\mr{d},Q_\mr{d}}(\tau)$
    is at most $\order{[\log (1/\epsilon)]^2}$.
\end{enumerate}
\end{theorem*}

\textbf{Proof.---}
We evaluate the truncation errors.
We obtain the upper bound on each order-$q$ term of $\Delta_\mr{Merge}(\tau)$ in Eq. (\ref{EqA_Merge:Delta_Merge_bound}).
This calculation goes also for the truncated one $\Delta_\mr{Merge}^{Q_\mr{d}}(\tau)$ and the difference $\Delta_\mr{Merge}(\tau)-\Delta_\mr{Merge}^{Q_\mr{d}}(\tau)$, which respectively contain the contributions from $q \in [2,Q_\mr{d}]$ and those from $q \in [Q_\mr{d}+1,\infty)$.
In a similar manner to Eq. (\ref{EqA_Merge:Delta_Merge_bound}), we obtain
\begin{equation}\label{EqA_Merge:Delta_Qd_bound}
    \sup_{\tau'\in [0,\tau]} \left( \norm{\Delta_\mr{Merge}^{Q_\mr{d}}(\tau')}_\text{Pauli} \right) \, \tau \leq \frac12 \sum_{q=2}^{Q_\mr{d}} (4 \xi g\tau)^{q+1}  \leq \frac12 \, \frac{1}{1-\frac2{3e}} (4\xi g \tau)^3 < e^{-1},
\end{equation}
and
\begin{equation}
    \sup_{\tau'\in [0,\tau]} \left( \norm{\Delta_\mr{Merge}(\tau')- \Delta_\mr{Merge}^{Q_\mr{d}}(\tau')}_\text{Pauli} \right) \, \tau \leq \frac12 \sum_{q=Q_\mr{d}+1}^\infty (4 \xi g\tau)^{q+1} \leq \frac12 \sum_{q=Q_\mr{d}+1}^\infty e^{-q-1} < \frac{e^{-Q_\mr{d}-1}}{e-1}.
\end{equation}
We denote $\sup_{\tau'\in [0,\tau]} ( \|\Delta_\mr{Merge}(\tau')\|_\text{Pauli} ) \, \tau$ and $\sup_{\tau'\in [0,\tau]} ( \|\Delta_\mr{Merge}^{Q_\mr{d}}(\tau')\|_\text{Pauli} ) \, \tau$ respectively as $D(\tau)$ and $D^{Q_\mr{d}}(\tau)$ below for brevity.
Then, the truncation error by introducing $M_\mr{d}$ and $Q_\mr{d}$ is bounded by
\begin{eqnarray}
    && \norm{\mcl{U}_\mr{Merge}(\tau)-\mcl{U}_\mr{Merge}^{M_\mr{d},Q_\mr{d}}(\tau)}_\Diamond \nonumber \\
    && \qquad \leq \sum_{m=M_\mr{d}+1}^\infty \int_0^\tau \dd \tau_m \cdots \int_0^{\tau_2} \dd \tau_1 \norm{\Delta_\mr{Merge} (\tau_1) \cdots \Delta_\mr{Merge}(\tau_m)}_\Diamond \nonumber \\
    && \qquad \qquad\qquad \qquad\qquad  + \sum_{m=1}^{M_\mr{d}} \int_0^\tau \dd \tau_m \cdots \int_0^{\tau_2} \dd \tau_1 \norm{\Delta_\mr{Merge} (\tau_1) \cdots \Delta_\mr{Merge} (\tau_m) -  \Delta_\mr{Merge}^{Q_\mr{d}} (\tau_1) \cdots \Delta_\mr{Merge}^{Q_\mr{d}} (\tau_m)}_\Diamond \nonumber \\
    && \qquad \leq  \sum_{m=M_\mr{d}+1}^\infty \frac{\left[ D(\tau) \right]^m}{m!} + \sum_{m=1}^{M_\mr{d}} \frac{\left[\max \left( D(\tau), D^{Q_\mr{d}}(\tau)\right)\right]^{m-1}}{(m-1)!} \sup_{\tau' \in [0,\tau]} \left(\norm{\Delta_\mr{Merge}(\tau')-\Delta_\mr{Merge}^{Q_\mr{d}}(\tau')}_\Diamond \right) \, \tau. \nonumber \\
    && \qquad \leq  \sum_{m=M_\mr{d}+1}^\infty e^{-m} + \sum_{m=1}^\infty \frac{1}{(m-1)!} \, \frac{e^{-Q_\mr{d}-1}}{e-1} = \frac{e^{-M_\mr{d}}+e^{-Q_\mr{d}}}{e-1}. \label{EqA_Merge:Merge_truncation_error}
\end{eqnarray}
We use the relation between the diamond norm and the Pauli norm, Eq. (\ref{Eq_Setup:Norm_relation}), in the second inequality.
We use $D(\tau) \leq e^{-1}$ for the first term in the last line and use $D(\tau), D^{Q_\mr{d}}(\tau) \leq 1$ for the second term, which are confirmed by Eqs. (\ref{EqA_Merge:Delta_Merge_bound}) and (\ref{EqA_Merge:Delta_Qd_bound}) in combination with the assumption Eq. (\ref{EqA_Merge:time_assumption}).
Let us set the truncation orders $M_\mr{d}$ and $Q_\mr{d}$ by
\begin{equation}
    M_\mr{d} = \max \left( \left\lceil \log \frac{2}{(e-1)\epsilon} \right\rceil, 1 \right) \in \Theta (\log (1/\epsilon)) , \quad Q_\mr{d} = \max \left( \left\lceil \log \frac{2}{(e-1)\epsilon} \right\rceil ,2 \right) \in \Theta (\log (1/\epsilon)).
\end{equation}
The error given by Eq. (\ref{EqA_Merge:Merge_truncation_error}) is smaller than $\epsilon$, indicating the satisfaction of Eq. (\ref{EqA_Merge:Merge_allowable_error}).

We next examine the properties of $\mcl{A}_\mr{Merge}^{M_\mr{d},Q_\mr{d}}(\tau)$.
The map $\mcl{A}_\mr{Merge}^{M_\mr{d},Q_\mr{d}}(\tau)$ is clearly HP, since it is composed of nested commutators among HP maps as Eqs. (\ref{EqA_Merge:A_Md_Qd}) and (\ref{EqA_Merge:Delta_Qd}).
The upper bound on its Pauli norm is evaluated in the same way as Eq. (\ref{EqA_Merge:A_bound_proof}), which results in
\begin{equation}
    \norm{\mcl{A}_\mr{Merge}^{M_\mr{d},Q_\mr{d}}(\tau)}_\mr{Pauli} \leq \sum_{m=1}^{M_\mr{d}} \frac{[D^{Q_\mr{d}}(\tau)]^m}{m!} \leq\frac12 \, \frac{1}{1-\frac2{3e}} (4\xi g \tau)^3  \exp \left( \frac12 \, \frac{1}{1-\frac2{3e}} \left( \frac{2}{3e} \right)^3   \right) \leq 43 (\xi g\tau)^3.
\end{equation}
This confirms the satisfaction of Eq. (\ref{EqA_Merge:A_Md_Qd_bound}).
Finally, the locality of $\mcl{A}_\mr{Merge}^{M_\mr{d},Q_\mr{d}}(\tau)$ is immediately obtained as follows.
Each truncated operator $\Delta_\mr{Merge}^{Q_\mr{d}}(\tau)$ is at most $(Q_\mr{d}+1)k$-local.
As a result, the locality of $\mcl{A}_\mr{Merge}^{M_\mr{d},Q_\mr{d}}(\tau)$ is bounded by $M_\mr{d}(Q_\mr{d}+1)k \in \order{[\log (1/\epsilon)]^2}$.
This completes the proof. $\quad \square$

\subsection{Implementation of the quasi-probabilistic sampling}\label{SubsecA:Merge_implement}

We discuss how to sample the quantum circuits for reproducing the operator $\mcl{U}_\mr{Merge}(\tau)$.
We set $\epsilon = \varepsilon/\poly{N,t}$ following the algorithm.
The operator $\Delta_\mr{Merge}^{Q_\mr{d}}(\tau)$ defined by Eq. (\ref{EqA_Merge:Delta_Qd}) can be expanded by
\begin{equation}\label{EqA_Merge:Delta_Qd_Pauli}
    \Delta_\mr{Merge}^{Q_\mr{d}}(\tau) \rho = \sum_{\mu,\nu} \gamma_{\mu\nu}^{Q_\mr{d}}(\tau) P_\mu \rho P_\nu.
\end{equation}
We set $\gamma_{\mu\nu}^{Q_\mr{d}}(\tau) \geq 0$ without loss of generality: Although it can be a complex number, we can reproduce such a map by quasi-probabilistic sampling, in which we sample a quantum circuit for $\exp (i \arg [\gamma_{\mu\nu}^{Q_\mr{d}}(\tau)])  P_\mu$ and $P_\nu$ with the probability dependent on $| \gamma_{\mu\nu}^{Q_\mr{d}}(\tau)| \geq 0$.

Since it is composed of $\order{Q_\mr{d}}$-fold nested commutators among finite-ranged local terms as Eq. (\ref{EqA_Merge:Delta_Qd}), we can calculate all nonzero $\{\gamma_{\mu\nu}^{Q_\mr{d}}(\tau)\}$ with $\poly{N,t,1/\varepsilon}$-time classical computation as follows.
First, let us focus on a $q$-fold nested commutator, $(\ad_{\mcl{L}_{B}+\mcl{L}_{B'}})^{l_1} (\ad_{\mcl{L}_{BB'}})^{l_2} (\ad_{\mcl{L}_{B}+\mcl{L}_{B'}})^{l_3} \ad_{\mcl{L}_{B}+\mcl{L}_{B'}} \mcl{L}_{B:B'} $, with $l_1+l_2+l_3=q-1$ in Eq. (\ref{EqA_Merge:Delta_Qd}).
When we independently count all the nested commutators in the Pauli basis, the time to classically determine its coefficients is as large as the number of the numerous branches, $\order{q!}$.
However, each $q$-fold nested commutator can have at most $\poly{q} e^{\order{q}} \, |\supp (\mcl{L}_{B:B'})|$ terms in the Pauli basis, reflecting the number of connected domains giving the supports of nested commutators.
When we recursively compute $q'$-fold nested commutators from $q'=1$ to $q'=q$ and collect identical Pauli terms at each $q'$, we can efficiently determine all the coefficients of the $q$-fold nested commutator in the Pauli basis with $\poly{q} e^{\order{q}}$ time \cite{Haah-2022-learning,Bakshi-2024-learning}.
Since we have at most $3^q$ combinations for the choice of $l_1,l_2,l_3$ (and also $2^q$ combinations for the second term) in Eq. (\ref{EqA_Merge:Delta_Qd}), the time to classically calculate all the coefficients $\{\gamma_{\mu\nu}^{Q_\mr{d}}(\tau) \}$ for $\Delta_\mr{Merge}^{Q_\mr{d}}(\tau)$ in Eq. (\ref{EqA_Merge:Delta_Qd_Pauli}) can be
\begin{equation}
    \sum_{q=2}^{Q_\mr{d}} (3^q+2^q) \times \poly{q} \, e^{\order{q}} \subset e^{\order{Q_\mr{d}}} \subset \order{\poly{N,t,1/\varepsilon}},
\end{equation}
for each $\tau$.

Substituting the expansion by Eq. (\ref{EqA_Merge:Delta_Qd_Pauli}), we obtain
\begin{equation}
    \left[ 1+ \mcl{A}_\mr{Merge}^{M_\mr{d},Q_\mr{d}}(\tau)\right]  \rho = \sum_{m=0}^{M_\mr{d}} \sum_{\mu_1,\nu_1,\cdots,\mu_m,\nu_m} \int_0^{\tau} \dd \tau_m \cdots \int_0^{\tau_2} \dd \tau_1 \left( \prod_{m'=m}^1 \gamma_{\mu_{m'}\nu_{m'}}^{Q_\mr{d}}(\tau_{m'}) \right) \left( \prod_{m'=m}^1 P_{\mu_{m'}}\right) \rho \left( \prod_{m'=1}^m P_{\nu_{m'}}\right).
\end{equation}
The quasi-probabilistic sampling based on Lemma \ref{LemmaA:quasiprobabilistic} requires the sampling of $m,\mu_1,\nu_1,\cdots,\mu_m,\nu_m$ with the joint probability distribution,
\begin{eqnarray}
    p(m,\mu_1,\nu_1,\cdots,\mu_m,\nu_m) &\propto& \int_0^{\tau} \dd \tau_m \cdots \int_0^{\tau_2} \dd \tau_1 \gamma_{\mu_1\nu_1}^{Q_\mr{d}}(\tau_1) \cdots \gamma_{\mu_m\nu_m}^{Q_\mr{d}}(\tau_m).
\end{eqnarray}
It can be efficiently sampled in a similar manner to the sampling-based simulation of time-dependent Hamiltonians \cite{Zhang-2022-time_dep}, as the operator $\mcl{A}_\mr{Merge}^{M_\mr{d},Q_\mr{d}}(\tau)$ is expressed by the Dyson series expansion, Eq. (\ref{EqA_Merge:A_Md_Qd}).
To be concrete, we first randomly choose the order $m \in \{0,1,2,\cdots,M_\mr{d} \}$ with the probability,
\begin{equation}
    p(m) \propto \frac1{m!} \left( \sum_{\mu,\nu} \int_0^\tau \dd \tau' \gamma_{\mu\nu}^{Q_\mr{d}}(\tau') \right)^m,
\end{equation}
which is the truncated version of the Poisson distribution.
Then, we pick up $m$ real numbers from $[0,\tau]$, where each value $\tau'$ is chosen with the probability density,
\begin{equation}
    p(\tau') \propto \sum_{\mu,\nu} \gamma_{\mu\nu}^{Q_\mr{d}}(\tau').
\end{equation}
We denote them as $\tau_1,\tau_2,\cdots,\tau_m$ after re-ordering them in the ascending order.
For each $m'=1,2,\cdots,m$, we sample the pair of indices $(\mu_{m'},\nu_{m'})$ with the probability distribution,
\begin{equation}
    p(\mu_{m'},\nu_{m'}) = \frac{\gamma_{\mu_{m'}\nu_{m'}}^{Q_\mr{d}}(\tau_{m'})}{\sum_{\mu,\nu} \gamma_{\mu\nu}^{Q_\mr{d}}(\tau_{m'})}.
\end{equation}
The sampled Pauli products $\prod_{m'=m}^1 P_{\mu_{m'}}$ and $\prod_{m'=1}^m P_{\nu_{m'}}$ are both at most $\order{[\log (Nt/\varepsilon)]^2}$-local, and hence the number of $\order{1}$-qubit gates required for implementing the quantum circuits in Eq. (\ref{EqA_Basic:sampled_circuit}) is at most $\order{[\log (Nt/\varepsilon)]^2}$.
It is small compared to those for the other parts in the algorithm.
We note that the overhead in this sampling is given by
\begin{equation}
    \order{\left[ \sum_{m=0}^{M_\mr{d}} \sum_{\mu_1,\nu_1,\cdots,\mu_m,\nu_m} \int_0^{\tau} \dd \tau_m \cdots \int_0^{\tau_2} \dd \tau_1  \prod_{m'=m}^1 \gamma_{\mu_{m'}\nu_{m'}}^{Q_\mr{d}}(\tau_{m'}) \right]^2} \subset e^{\order{\int_0^\tau \dd \tau'\norm{\Delta_\mr{Merge}(\tau')}_\mr{Pauli}}} \subset e^{\order{(\xi g \tau)^3}},
\end{equation}
which retains the discussion based on Lemma \ref{LemmaA:quasiprobabilistic}.
The quasi-probabilistic sampling affects the computational cost by the sampling overhead as we discuss in Section \ref{Subsec:Outline_algo}.

\textbf{Remark.---}
The above construction of the sampled quantum circuits based on Theorem \ref{Thm:Merge_truncation} is redundant for one-dimensional systems.
To be precise, the support of the merged operator $\mcl{U}_\mr{Merge}(\tau)$ has the size as large as $R \in \Theta (\log (Nt/\varepsilon))$.
The merged operator can be expanded by $\Theta (\log (Nt/\varepsilon))$-local Pauli operators, and sampling $\order{[\log (Nt/\varepsilon)]^2}$ local gates is redundant.
In practice, the matrix dimension of the merged operator is at most $2^{\order{R}} \subset \order{\poly{N,t,1/\varepsilon}}$, and hence it is easy to calculate all of its matrix elements by classical computation based on Eq. (\ref{EqA_Merge:Merge_op}).
Expanding it in the Pauli basis, the sampling can be much simpler for one-dimensional systems.

In contrast, the above construction becomes significant for high-dimensional systems with the dimension $d \geq 2$.
As discussed in Appendix \ref{SecA:High_dim}, the support size of the merged operator is as large as $\order{(R_\mr{p})^{d-1}R}$ with some optimized block size $R_\mr{p} \in \order{\poly{Nt}}$, corresponding to the boundary size in $d$ dimension.
The merged operator has the exponentially large matrix dimension in $N$ and $t$, and cannot be efficiently calculated by classical computers.
However, the counterparts of Theorem \ref{Thm:Merge_expansion} and Theorem \ref{Thm:Merge_truncation} are still valid.
They support that it is sufficient to sample $\order{[\log(Nt/\varepsilon)]^2}$ local gates with sampling probability that can be efficiently determined by classical computation also for high-dimensional systems.
See Appendix \ref{SecA:High_dim} for its detail.

\section{Extension to high-dimensional systems}\label{SecA:High_dim}

\renewcommand{\thetheorem}{\thesection\arabic{theorem}}
\setcounter{theorem}{0}

In this appendix, we show the extension of the algorithms to Lindbladian simulation for high-dimensional systems.

\subsection{Setup and the patching lemma for high-dimensional systems}

We first clarify the setup.
We consider a $d$-dimensional lattice, and suppose that the dimension $d$ is constant, i.e., independent of $N$, $t$, and $\varepsilon$.
The lattice $\Lambda$ is assumed to be a hypercubic lattice given by $\Lambda = \{1,2,\cdots,L\}^d$ with the number of sites $N = L^d$ for simplicity, but we note that its geometry is not essential for the scaling of the computational cost.
We also define the locality and the range for high-dimensional systems.
Let us expand a generic lattice Lindbladian $\mcl{L}$ in the Pauli basis as Eq. (\ref{Eq_Setup:Lindbladian_expansion}).
The locality $k$ is defined in the same way as Eq. (\ref{Eq_Setup:k_def}), which means that each term of interactions or dissipation involves at most $k$ sites.
We assume that the interactions and dissipation are finite-ranged in a sense that there exists a quantity $\xi > 0$ satisfying Eq. (\ref{Eq_Basic:range_xi}).
We note that the distance measure $\mr{dist}(i,j)$ in the range $r(X)$ [See Eq. (\ref{Eq_Setup:domain_size_def})] is replaced by
\begin{equation}\label{EqA_High:distance}
    \mr{dist}(i,j) = \sqrt{\sum_{d'=1}^d (i_{d'}-j_{d'})^2}
\end{equation}
for lattice sites $i =(i_1,\cdots,i_d) \in \Lambda$ and $j=(j_1,\cdots,j_d) \in \Lambda$.
We suppose that the locality $k$ and the range $\xi$ are $\order{1}$ constants.
It is easy to see that we have $k \leq \mr{Const.} \times \xi^d$, while we do not use it explicitly here.
We define the extensiveness $g$ in the same way as Eq. (\ref{Eq_Setup:g_def}).
It means the maximum energy scale per site, and we have $g \in \order{1}$ as well as one-dimensional systems.

As the first step for the extension, we discuss the patching lemma for high-dimensional lattice Lindbladians, corresponding to Theorem \ref{Thm_Patch:Patching_lemma}.
We split the lattice $\Lambda$ into the subsystems $A$, $B$, and $C$, and suppose that $A$ and $C$ are separated by the intermediate region $B$.
We assume $\mr{dist}(A,C) \geq R$.
The error of the patching lemma, i.e., the right hand side of Eq. (\ref{Eq_Patch:error}), comes from the boundaries of the subsystems.
Thus, the difference from the one-dimensional case arises due to the boundary sizes among the subsystems.
We define the boundary domain between $A$ and $B$ with the size $\xi$ by
\begin{equation*}
    \partial_{AB}(\xi) = \{ i \in A \, | \, \text{$^\exists j \in B$ s.t. $\mr{dist}(i,j) \leq \xi$}\} \, \cup \, \{ j \in B \, | \, \text{$^\exists i \in A$ s.t. $\mr{dist}(i,j) \leq \xi$}\},
\end{equation*}
and also define $\partial_{BC}(\xi)$ in a similar manner.
The error bound on the patching lemma depends on the boundary sizes as follows.

\begin{theorem}
\textbf{(Patching lemma for $d$-dimensional systems)}

Suppose that the distance between the subsystems $A$ and $C$, denoted by $R=\mr{dist}(A,C)$, satisfies
\begin{equation}\label{EqA_High:R_condition}
    R > \xi \, \max \{ 1, \log (|\partial_{BC}(\xi)|) \}.
\end{equation}
When the time $\tau$ is small enough to satisfy
\begin{equation}\label{EqA_High:time_assumption}
    0 \leq \tau \leq \frac{1}{6ekg} \in \order{1},
\end{equation}
the time evolution $e^{\mcl{L}\tau}$ is approximated by $e^{\mcl{L}_{AB}\tau} e^{-\mcl{L}_{B}\tau} e^{\mcl{L}_{BC}\tau}$ with an error bound
\begin{equation}\label{EqA_High:error_Patching}
    \norm{e^{\mcl{L}\tau} - e^{\mcl{L}_{AB}\tau} e^{-\mcl{L}_{B}\tau} e^{\mcl{L}_{BC}\tau}}_\Diamond \leq \frac{|\partial_{BC}(\xi)|}{2} \, e^{-\frac{R}{\xi}},
\end{equation}
where the subsystem Lindbladians $\mcl{L}_{AB}$, $\mcl{L}_B$, and $\mcl{L}_{BC}$ are defined based on Eq. (\ref{Eq_Patch:Subsys_Lindbladian}).

\end{theorem}

\textbf{Proof.---}
The proof is essentially the same as the one for Theorem \ref{Thm_Patch:Patching_lemma}.
We follow the calculation from Eq. (\ref{Eq_Patch:N_tau}) to Eq. (\ref{Eq_Patch:N_tau_1_bound_result}).
We note that $R > \xi$ in Eq. (\ref{EqA_High:R_condition}) is required for the relations Eqs. (\ref{Eq_Patch:AB_boundary_1}) and (\ref{Eq_Patch:AB_boundary_2}) to be valid.
In the third inequality in Eq. (\ref{Eq_Patch:K_s_bound}), we replace the relation $|\supp (\mcl{L}_{B:C})| \leq 2 \xi$, which is valid for one-dimensional systems, by $|\supp (\mcl{L}_{B:C})| \leq |\partial_{BC}(\xi)|$.
The map $\mcl{K}(\tau')$ defined by Eq. (\ref{Eq_Patch:K_s}) is bounded by
\begin{eqnarray}
    \norm{\mcl{K}(\tau')}_\Diamond &\leq& |\partial_{BC}(\xi)| \sum_{q=\lceil \frac{R}\xi -1 \rceil}^\infty \sum_{\substack{l,m,n \geq 0: \\ l+m+n = q}} \frac{q!}{l!m!n!} (2kg\tau')^q g \nonumber \\
    &\leq& |\partial_{BC}(\xi)| g \sum_{q=\lceil \frac{R}\xi -1 \rceil}^\infty (6kg\tau')^q \nonumber \\
    &\leq& \frac{e^2}{e-1} e^{-\frac{R}\xi} |\partial_{BC}(\xi)| g.
\end{eqnarray}
Since we have $\tau \times \sup_{\tau' \in [0,\tau]} (\norm{\mcl{K}(\tau')}_\Diamond) \leq \frac{e}{6(e-1)} e^{-\frac{R}{\xi}} |\partial_{BC}(\xi)| \leq 1$ under the assumptions Eq. (\ref{EqA_High:R_condition}), the error is bounded by
\begin{equation}
    \norm{e^{\mcl{L}\tau} - e^{\mcl{L}_{AB}\tau} e^{-\mcl{L}_{B}\tau} e^{\mcl{L}_{BC}\tau}}_\Diamond \leq \sum_{n=1}^\infty \frac{1}{n!} \left( \tau \, \sup_{\tau' \in [0,\tau]} (\norm{\mcl{K}(\tau')}_\Diamond)\right)^n \leq  \frac{e}6 |\partial_{BC}(\xi)| e^{-\frac{R}\xi}.
\end{equation}
This completes the proof of Eq. (\ref{EqA_High:error_Patching}). $\quad \square$

The error bound of the patching lemma grows linearly in the boundary size $|\partial_{BC}(\xi)|$, but this does not matter in the algorithms.
For an arbitrarily small value $\epsilon \in (0,1)$, we can suppress the error bound Eq. (\ref{EqA_High:error_Patching}) up to $\epsilon$ by setting $R = \xi \log (|\partial_{BC}(\xi)|/\epsilon)$.
Since the boundary size $|\partial_{BC}(\xi)|$ is smaller than the system size $N$ and the quantity $\epsilon$ will be set to $\order{\varepsilon/\poly{N,t}}$, the block size $R \in \order{\xi \log (Nt/\varepsilon)}$ is sufficient.
The choice of the intermediate block size $R$ is the same as the one-dimensional case.
As a result, the algorithms based on the patching lemma work also for high-dimensional systems even with the additional factor by the boundary size in the patching lemma.

In Algorithm \ref{Algorithm_Sparse} and Algorithm \ref{Algorithm_Generic}, we do not directly use the patching lemma in the form of Theorem \ref{Thm_Patch:Patching_lemma}, but instead use Corollary \ref{Cor_Sparse:Patching_strategy}.
Its counterpart for high-dimensional systems immediately follows from the same discussion as Theorem \ref{Thm_Patch:Patching_lemma}.
We split the lattice $\Lambda$ into the subsystems, $A_1,A_2, \cdots, A_{N_\mr{p}}$, and define their boundaries with the width $R$ by
\begin{equation}\label{EqA_High:internal}
    \partial_{A_\alpha} (R) = \{ j \in A_\alpha \, | \, \text{$^\exists j' \in \Lambda \backslash A_\alpha$ s.t. $\mr{dist}(j,j') < R$} \}, \quad \alpha =1,2,\cdots, N_\mr{p}.
\end{equation}
We also define the outside boundaries by
\begin{equation}\label{EqA_High:external}
    \overline{\partial}_{A_\alpha} (R) = \{ j \in (\Lambda \backslash A_\alpha) \, | \, \text{$^\exists j' \in A_\alpha$ s.t. $\mr{dist}(j,j') < R$} \}, \quad \alpha =1,2,\cdots, N_\mr{p}.
\end{equation}
We denote the union of these boundaries by $\partial(R) = \partial_{A_1} (R) \cup \cdots \cup \partial_{A_{N_\mr{p}}} (R)$.
Figure \ref{FigA_high_dim_sparse} (a) briefly shows the splitting and the boundaries.
We give the counterpart of Corollary \ref{Cor_Sparse:Patching_strategy} as follows.

\begin{corollary}\label{CorA_High:Patching_strategy}
\textbf{}

Suppose that the distance $R$ satisfies $R \geq \xi \log N$.
When the time $\tau$ is small enough to satisfy $0 \leq \tau \leq (6ek g)^{-1} \in \order{1}$, the following inequality is satisfied,
\begin{equation}\label{EqA_High:Paching_strategy_error}
    \norm{e^{\mcl{L}\tau}- e^{\mcl{L}_{\partial (R)}\tau} \left( \prod_{\alpha =1}^{N_\mr{p}} e^{-\mcl{L}_{\partial_{A_\alpha} (R)} \tau} \right) \left( \prod_{\alpha =1}^{N_\mr{p}} e^{\mcl{L}_{A_\alpha} \tau} \right)}_\Diamond \leq N e^{-\frac{R}{\xi}}.
\end{equation}
\end{corollary}

\textbf{Proof.---} This immediately follows from the same calculation for the proof of Corollary \ref{Cor_Sparse:Patching_strategy} [See Appendix \ref{SubsecA:Patching_corollary}]. $\quad \square$

\begin{figure*}
    \centering
    \includegraphics[width=0.95\linewidth]{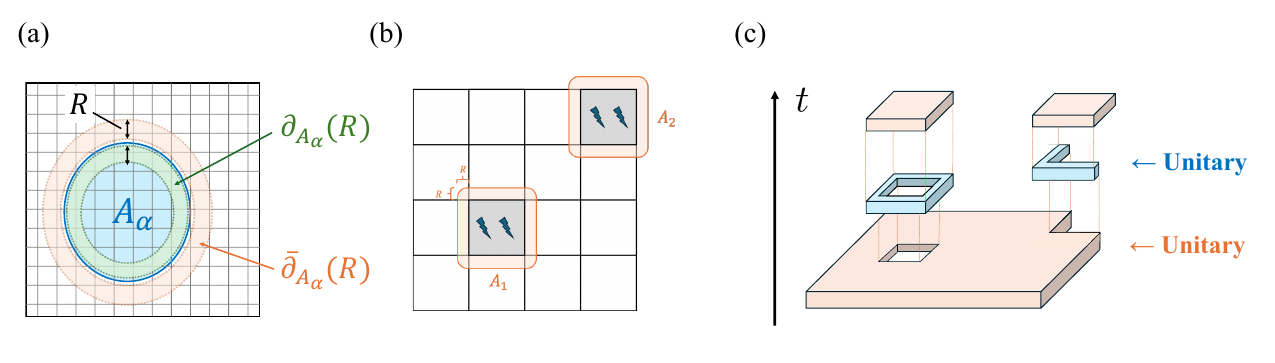}
    \caption{(a) The definition of the internal and external boundaries of each block $A_\alpha$, defined by Eqs. (\ref{EqA_High:internal}) and (\ref{EqA_High:external}). (b) The way of splitting the lattice into blocks for high-dimensional sparsely dissipative systems. (c) The approximate time-evolution operator $\mcl{U}_\mr{Patch}(\tau)$ for high-dimensional sparsely dissipative systems, where we set the blocks $\{A_\alpha\}$ by Eq. (\ref{EqA_High:block_sparse}). }
    \label{FigA_high_dim_sparse}
\end{figure*}

\subsection{Near-optimal algorithm for sparsely dissipative systems}

We extend Algorithm \ref{Algorithm_Sparse} for sparsely dissipative systems.
The definition of sparsely dissipative systems for high-dimensional cases is given in the same way as Definition \ref{Def_Setup:sparse_dissipation}, in which the distance measure $\mr{dist}(i,j)$ for the domain size and the domain distance is replaced by the $d$-dimensional one, Eq. (\ref{EqA_High:distance}).
We split the lattice by
\begin{equation}\label{EqA_High:block_sparse}
        A_\alpha = \Lambda_\alpha^\mr{diss} \cup \overline{\partial}_{ \Lambda_\alpha^\mr{diss}} (R) \quad  (\alpha =1,2,\cdots,N_\mr{d}), \qquad A_{N_\mr{d}+1} = \Lambda \backslash \left( \bigcup_{\alpha =1}^{N_{\mr{d}}} A_\alpha \right),
\end{equation}
where the number of the patches $N_\mr{p}$ is equal to $N_\mr{d} + 1$.
We show the schematic picture of this partition in Fig. \ref{FigA_high_dim_sparse} (b).
Owing to the sparsity of the dissipation, Eq. (\ref{Eq_Setup:dissipative_domain_dist}), each boundary domain $\partial_{A_\alpha}(R)$ or $\partial (R)$ has no intersection with the domains under dissipation, $\{ \Lambda_\alpha^\mr{diss} \}$ when we set the block size $R \in \order{\log N}$.
The time evolutions $e^{-\mcl{L}_{\partial_{A_\alpha}(R)}\tau}$ and $e^{\mcl{L}_{\partial (R)}\tau}$ become unitary, and hence we can construct an algorithm analogous to Algorithm \ref{Algorithm_Sparse}.
We obtain the following computational cost as a result.

\begin{theorem}
\textbf{}

Suppose that $t,1/\varepsilon \in \order{\poly{N}}$ is satisfied.
There exists a quantum algorithm that simulates the time-evolved state $e^{\mcl{L}t} \rho$ of sparsely dissipative Lindbladians in $d$ dimensions within an error $\varepsilon$ with the following cost:
\begin{itemize}
    \item The number of $\order{1}$-qubit gates: $\order{Nt \, \polylog{Nt/\varepsilon}}$ (near-optimal).
    \item The number of ancilla qubits and the circuit depth:
    The algorithm runs with $\order{\polylog {Nt/\varepsilon}})$ ancilla qubits, and then it yields the circuit depth $\order{Nt \, \polylog{Nt/\varepsilon}}$.
    When $\tilde{\Theta}(N)$ ancilla qubits are available, the circuit depth can be $\order{t \, \polylog{Nt/\varepsilon}}$.
\end{itemize}

\end{theorem}

\textbf{Proof.---}
The algorithm is essentially the same as Algorithm \ref{Algorithm_Sparse}.
We split the time $t$ into $r_t$ parts, setting $\tau = t/r_t$.
We set $r_t \in \Theta (kgt)$ so that Eq. (\ref{EqA_High:time_assumption}) is satisfied and set $R \in \Theta ( \xi \log (Nr_t/\varepsilon)) = \Theta ( \xi \log (Nt/\varepsilon))$ so that the error in Eq. (\ref{EqA_High:Paching_strategy_error}) can be bounded by $\order{\varepsilon/r_t}$.
It is sufficient to implement each component of the decomposed time evolution within an error $\order{\varepsilon/(Nr_t)}=\order{\varepsilon/(Nt)}$ as follows.

\begin{itemize}
    \item Implementation of $e^{\mcl{L}_{A_\alpha}\tau}$ ($\alpha =1,\cdots,N_\mr{d}$):
    We use the LCU-based approach for Lindbladian simulation \cite{Li-Wang-2022-open}.
    Since the domain size $A_\alpha$ is at most $\order{R^d} \subset \order{[\log (Nt/\varepsilon)]^d}$, the gate count for this part is at most
    \begin{equation}
        \Otilde{R^{2d} \tau} \subset \order{\polylog {Nt/\varepsilon}}.
    \end{equation}
    
    \item Implementation of $e^{\mcl{L}_{A_{N_\mr{d}+1}}\tau}$: 
    The domain $A_{N_\mr{d}+1}$ does not contain the dissipative terms, and we run the HHKL algorithm for Hamiltonian simulation \cite{Haah2021-hhkl}.
    The gate count is at most $\Otilde{|A_{N_\mr{d}+1}|\tau} \subset \Otilde{N}$.

    \item Implementation of $e^{-\mcl{L}_{\partial_{A_\alpha}(R)} \tau}$ ($\alpha =1,\cdots,N_\mr{d}+1$) and $e^{\mcl{L}_{\partial (R)} \tau}$:
    Every domain $\partial_{A_\alpha}(R)$ does not contain the dissipative terms due to the sparsity of dissipation.
    We can apply the HHKL algorithm for Hamiltonian simulation, whose gate count results in $\Otilde{|\partial_{A_\alpha}(R)|\tau} \subset \Otilde{|\partial_{A_\alpha}(R)|}$.
    The same goes also for $e^{\mcl{L}_{\partial (R)} \tau}$, and the gate count for it amounts to $\Otilde{|\partial (R)|\tau} \subset \Otilde{N}$.
\end{itemize}
Summing the above gate counts in $r_t$ steps, we obtain the total gate count $\order{Nt \, \polylog{Nt/\varepsilon}}$.
The relation between the number of ancilla qubits and the circuit depth is obtained by the parallel discussion in Theorem \ref{Thm_Setup:sparse}. $\quad \square$

In conclusion, we can achieve the near-optimal gate count $\order{Nt \, \polylog{Nt/\varepsilon}}$ also for sparsely dissipative systems in high dimension.

\subsection{Efficient algorithms by patching and merging for generic dissipative systems}

\begin{figure*}
    \centering
    \includegraphics[width=0.95\linewidth]{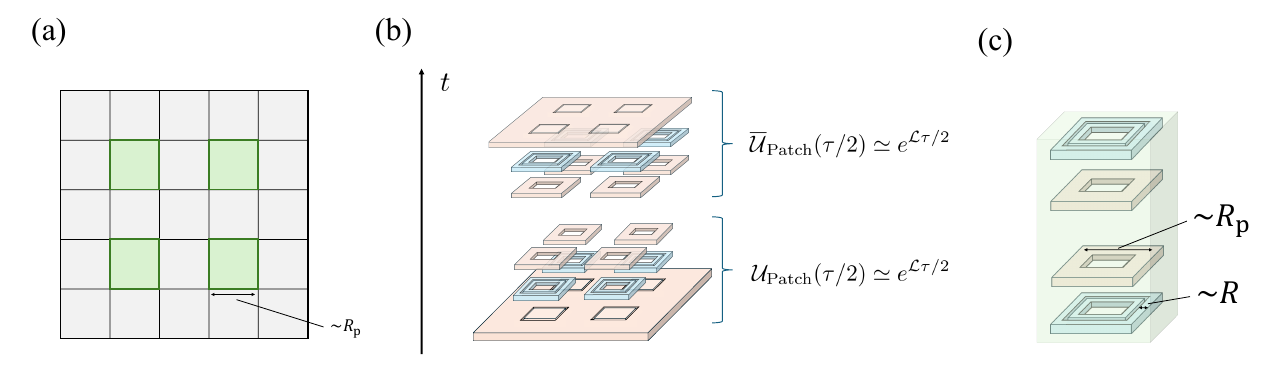}
    \caption{(a) The way of splitting the lattice based on Eq. (\ref{EqA_High:domains_generic}). Green domains, whose length scale is $R_\mr{p}$, are chosen as the set of $\{ B_\alpha \}_{\alpha \in (2\bbN)^d}$. (b) The approximation of the time-evolution operator $e^{\mcl{L}\tau}$ by the patching lemma. (c) The merged block. It is composed of the internal layers of $\overline{\mcl{U}}_\mr{Patch}(\tau/2) \, \mcl{U}_\mr{Patch}(\tau/2)$.}
    \label{Fig_2d}
\end{figure*}

We consider $d$-dimensional lattice systems, in which the Lindbladian has generic finite-ranged interactions and dissipation, and extend Algorithm \ref{Algorithm_Generic} to them.
The strategy is essentially the same as Section \ref{Subsec:Outline_algo}.
As the patching strategy, we introduce the flexible block size $R_\mr{p}$ in addition to the boundary one $R \in \Theta (\log (Nt/\varepsilon))$, and optimize $R_\mr{p}$ later.
As the merging strategy, we decompose the half time-evolution $e^{\mcl{L}\tau/2}$ by the patching lemma in two ways, and merge some of the block time-evolution operators in the boundaries so that the sampling complexity can be suppressed.

We split the lattice $\Lambda$ into the subsystems $\{B_\alpha\}$ with the length $R_\mr{p}$, each of which is defined by
\begin{equation}\label{EqA_High:domains_generic}
    B_\alpha = \{ (j_1,\cdots, j_d) \in \Lambda \, | \, (\alpha_{d'}-1) R_\mr{p} + 1 \leq j_{d'} \leq \alpha_{d'} R_\mr{p}, \, d' = 1,2,\cdots, d\},
\end{equation}
for $\alpha = (\alpha_1,\cdots,\alpha_d) \in \{1,2,\cdots, L/R_\mr{p} \}^d$.
We will set the block size at the boundaries $R$ so that $R \in \Theta (\log (Nt/\varepsilon))$ and $R \in o(R_\mr{p})$ can be satisfied.
We repeat the patching and merging steps like Section \ref{Subsec:Outline_algo} as follows.
See also Figures \ref{Fig_2d} and \ref{Fig_2d_repetition} for the schematic picture of the protocol in the two-dimensional case.

\begin{enumerate}
    \item (Patching) We cut out the patches for the even indices, i.e., $\{ B_\alpha \}$ for $\alpha \in (2\bbN)^d$.
    We denote the other domain in the lattice by 
    \begin{equation}
        \overline{B}_1 = \Lambda \backslash \left(\bigcup_{\alpha \in (2\bbN)^d} B_\alpha\right).
    \end{equation}
    Applying the patching lemma by Corollary \ref{CorA_High:Patching_strategy}, we obtain two approximations,
    \begin{eqnarray}
        \mcl{U}_\mr{Patch}(\tau/2) &=& \left( \prod_{\alpha \in (2\bbN)^d} e^{\mcl{L}_{\partial_{B_\alpha}(R) \overline{\partial}_{B_\alpha}(R)} \frac{\tau}{2}} \right) \left( \prod_{\alpha \in (2\bbN)^d} e^{-\mcl{L}_{\partial_{B_\alpha}(R)} \frac{\tau}{2}} e^{ -\mcl{L}_{\overline{\partial}_{B_\alpha}(R)}  \frac{\tau}{2}} \right) \left( \prod_{\alpha \in (2\bbN)^d} e^{\mcl{L}_{B_\alpha} \frac{\tau}{2}}\right) e^{\mcl{L}_{\overline{B}_1}\frac{\tau}{2}}, \nonumber \\
        && \label{EqA_High:U_Patch/2}\\
        \overline{\mcl{U}} _\mr{Patch}(\tau/2) &=& \left( \prod_{\alpha \in (2\bbN)^d} e^{\mcl{L}_{B_\alpha} \frac{\tau}{2}}\right) e^{\mcl{L}_{\overline{B}_1}\frac{\tau}{2}} \left( \prod_{\alpha \in (2\bbN)^d} e^{-\mcl{L}_{\partial_{B_\alpha}(R)} \frac{\tau}{2}} e^{ -\mcl{L}_{\overline{\partial}_{B_\alpha}(R)}  \frac{\tau}{2}} \right) \left( \prod_{\alpha \in (2\bbN)^d} e^{\mcl{L}_{\partial_{B_\alpha}(R) \overline{\partial}_{B_\alpha}(R)} \frac{\tau}{2}} \right) , \nonumber \\
        && \label{EqA_High:U_bar_Patch/2}
    \end{eqnarray}
    whose errors are smaller than $\order{N e^{-R/\xi}}$.
    We note that the domains $\partial_{B_\alpha}(R)$ and $\overline{\partial}_{B_\alpha}(R)$ are respectively the set of the boundary sites inside or outside $B_\alpha$ within the distance $R$, defined by Eqs. (\ref{EqA_High:internal}) and (\ref{EqA_High:external}).
    See Figure \ref{FigA_high_dim_sparse} (a).

    \item (Merging)
    We merge some of the non-CP terms in $\overline{\mcl{U}}_\mr{Patch}(\tau/2) \mcl{U}_\mr{Patch}(\tau/2) = e^{\mcl{L}\tau}+\order{Ne^{-R/\xi}}$.
    We define the following merged operator,
    \begin{equation}\label{EqA_High:Merge_high_Step1}
        \mcl{U}_\mr{Merge}^\alpha(\tau) = e^{-\mcl{L}_{\partial_{B_\alpha}(R)} \tau/2} e^{ -\mcl{L}_{\overline{\partial}_{B_\alpha}(R)}  \tau/2} e^{\mcl{L}_{\partial_{B_\alpha}(R) \overline{\partial}_{B_\alpha}(R)} \tau} e^{-\mcl{L}_{\partial_{B_\alpha}(R)} \tau/2}  e^{ -\mcl{L}_{\overline{\partial}_{B_\alpha}(R)}  \tau/2},
    \end{equation}
    which nontrivially acts on the boundary of $B_\alpha$, i.e., $\partial_{B_\alpha}(R) \cup \overline{\partial}_{B_\alpha}(R)$.
    Figure \ref{Fig_2d} (c) shows its schematic picture.
    It gives a second-order PF for the boundary Lindbladian,
    \begin{equation}
        \mcl{L}_{\partial_{B_\alpha}(R) :\overline{\partial}_{B_\alpha}(R)} = \mcl{L}_{\partial_{B_\alpha}(R) \overline{\partial}_{B_\alpha}(R)}-\mcl{L}_{\partial_{B_\alpha}(R)}-\mcl{L}_{\overline{\partial}_{B_\alpha}(R)},
    \end{equation}
    which is composed of inter-block terms between $\partial_{B_\alpha}(R)$ and $\overline{\partial}_{B_\alpha}(R)$ like Eq. (\ref{Eq_Gen:Merge_op_PF}).
    Its support size is $\order{\xi (R_\mr{p})^{d-1}}$.
    The merged operator can be expressed as
    \begin{equation}\label{EqA_High:Merge_op_high_dim_expansion}
        \mcl{U}_\mr{Merge}^\alpha(\tau) = e^{ \mcl{L}_{\partial_{B_\alpha}(R) :\overline{\partial}_{B_\alpha}(R)} \tau } \left[ 1 + \mcl{A}_\mr{Merge}^\alpha(\tau) \right], \quad \norm{\mcl{A}_\mr{Merge}^\alpha(\tau)}_\text{Pauli} \in \order{(g\tau)^3 (R_\mr{p})^{d-1}}
    \end{equation}
    in a similar manner to Theorem \ref{Thm:Merge_expansion}, as we will confirm as Corollary \ref{CorA_High:Merge}.
    
    \item (Patching)
    We further decompose the time evolution $e^{\mcl{L}_{\overline{B}_1}\tau/2}$ in Eqs. (\ref{EqA_High:U_Patch/2}) and (\ref{EqA_High:U_bar_Patch/2}) by the patching lemma.
    This corresponds to the decomposition of the top and bottom layers in the left panel of Fig. \ref{Fig_2d_repetition} into those for smaller blocks in the central panel.
    Let $\overline{B}_2$ be the subsystem defined by
    \begin{equation}
        \overline{B}_2 = \overline{B}_1 \backslash \left(\bigcup_{\substack{\alpha = (\alpha_1,\cdots,\alpha_d): \\ \alpha_1 \in 2 \bbN -1, \, \alpha_2,\cdots,\alpha_d \in 2\bbN}} B_\alpha\right).
    \end{equation}
    We also define the boundary sites when regarding $\overline{B}_1$ as a whole system by
    \begin{eqnarray}
        \partial_{B_\alpha} (R; \overline{B}_1) &=& \{ j \in B_\alpha \, | \, \text{$^\exists j' \in \overline{B}_1 \backslash B_\alpha$ s.t. $\mr{dist}(j,j') < R$} \}, \\
        \overline{\partial}_{B_\alpha} (R; \overline{B}_1) &=& \{ j \in (\overline{B}_1 \backslash B_\alpha) \, | \, \text{$^\exists j' \in B_\alpha$ s.t. $\mr{dist}(j,j') < R$} \}.
    \end{eqnarray}
    Based on Corollary \ref{CorA_High:Patching_strategy}, the time evolution $e^{\mcl{L}_{\overline{B}_1}\tau/2}$ can be approximated by $\overline{\mcl{U}} _\mr{Patch}^{\overline{B}_1}(\tau/4) \, \mcl{U}_\mr{Patch}^{\overline{B}_1}(\tau/4)$ within an error $\order{Ne^{-R/\xi}}$, where the two operators are respectively given by
    \begin{eqnarray}
        \mcl{U}_\mr{Patch}^{\overline{B}_1}(\tau/4) &=& \left( \prod_{\substack{\alpha=(\alpha_1,\cdots,\alpha_d): \\ \alpha_1 \in 2\bbN-1, \\ \alpha_2,\cdots,\alpha_d \in 2\bbN}} e^{\mcl{L}_{\partial_{B_\alpha}(R; \overline{B}_1) \overline{\partial}_{B_\alpha}(R; \overline{B}_1)} \tau/4} \right)  \nonumber \\
        && \qquad \times \left( \prod_{\substack{\alpha=(\alpha_1,\cdots,\alpha_d): \\ \alpha_1 \in 2\bbN-1, \\ \alpha_2,\cdots,\alpha_d \in 2\bbN}} e^{-\mcl{L}_{\partial_{B_\alpha}(R; \overline{B}_1)} \tau/4} e^{ -\mcl{L}_{\overline{\partial}_{B_\alpha}(R; \overline{B}_1)}  \tau/4} \right) \left( \prod_{\substack{\alpha=(\alpha_1,\cdots,\alpha_d): \\ \alpha_1 \in 2\bbN-1, \\ \alpha_2,\cdots,\alpha_d \in 2\bbN}} e^{\mcl{L}_{B_\alpha} \tau/4}\right) e^{\mcl{L}_{\overline{B}_2}\tau/4}, \\
        \overline{\mcl{U}}_\mr{Patch}^{\overline{B}_1}(\tau/4) &=& \left( \prod_{\substack{\alpha=(\alpha_1,\cdots,\alpha_d): \\ \alpha_1 \in 2\bbN-1, \\ \alpha_2,\cdots,\alpha_d \in 2\bbN}} e^{\mcl{L}_{B_\alpha} \tau/4}\right) e^{\mcl{L}_{\overline{B}_2}\tau/4}  \left( \prod_{\substack{\alpha=(\alpha_1,\cdots,\alpha_d): \\ \alpha_1 \in 2\bbN-1, \\ \alpha_2,\cdots,\alpha_d \in 2\bbN}} e^{-\mcl{L}_{\partial_{B_\alpha}(R; \overline{B}_1)} \tau/4} e^{ -\mcl{L}_{\overline{\partial}_{B_\alpha}(R; \overline{B}_1)}  \tau/4} \right)\nonumber \\
        && \quad \qquad \qquad \qquad \qquad \qquad\qquad \qquad \qquad \qquad \times  \left( \prod_{\substack{\alpha=(\alpha_1,\cdots,\alpha_d): \\ \alpha_1 \in 2\bbN-1, \\ \alpha_2,\cdots,\alpha_d \in 2\bbN}} e^{\mcl{L}_{\partial_{B_\alpha}(R; \overline{B}_1) \overline{\partial}_{B_\alpha}(R; \overline{B}_1)} \tau/4} \right).
    \end{eqnarray}
    
    \item (Merging)
    As shown in the central panel of Fig. \ref{Fig_2d_repetition}, we define the merged operator 
    \begin{eqnarray}
        \mcl{U}_\mr{Merge}^{\alpha, \overline{B}_1}(\tau/2) &=& e^{-\mcl{L}_{\partial_{B_\alpha}(R; \overline{B}_1)} \tau/4} e^{ -\mcl{L}_{\overline{\partial}_{B_\alpha}(R;\overline{B}_1)}  \tau/4} e^{\mcl{L}_{\partial_{B_\alpha}(R; \overline{B}_1) \overline{\partial}_{B_\alpha}(R; \overline{B}_1)} \tau/2} e^{-\mcl{L}_{\partial_{B_\alpha}(R; \overline{B}_1)} \tau/4} e^{ -\mcl{L}_{\overline{\partial}_{B_\alpha}(R;\overline{B}_1)}  \tau/4} \nonumber \\
        &=& e^{\mcl{L}_{\partial_{B_\alpha}(R;\overline{B}_1):\overline{\partial}_{B_\alpha}(R; \overline{B}_1)}\frac\tau{2} } \left[ 1 + \order{(g\tau)^3(R_\mr{p})^{d-1}} \right] \label{EqA_High:Merge_op_2nd}
    \end{eqnarray}
    for each index $\alpha =(\alpha_1,\cdots,\alpha_d)$ such that $\alpha_1 \in 2\bbN-1$ and $\alpha_2,\cdots,\alpha_d \in 2\bbN$.
    The norm of the non-CP part is evaluated in the same way as Eq. (\ref{EqA_High:Merge_op_high_dim_expansion}) [See Corollary \ref{CorA_High:Merge}].

    \item (Repeat patching and merging)
    We execute the above protocols for the time evolution $e^{\mcl{L}_{\overline{B}_2}\tau/4}$.
    We set the domain
    \begin{equation}
        \overline{B}_3 = \overline{B}_2 \backslash \left(\bigcup_{\substack{\alpha = (\alpha_1,\cdots,\alpha_d): \\ \alpha_2 \in 2 \bbN -1, \, \alpha_1,\alpha_3,\cdots,\alpha_d \in 2\bbN}} B_\alpha\right),
    \end{equation}
    and further decompose the time evolution like the right panel of Fig. \ref{Fig_2d_repetition}.
    We repeat the above decomposition until all the blocks have at most $\order{(R_\mr{p})^d}$ sites.
    In the two-dimensional case, the right panel of Fig. \ref{Fig_2d_repetition} shows the situation after finishing the protocol.
    In generic $d$-dimensional systems, we classify the blocks $\{ B_\alpha \}$ based on the parities of $\alpha =(\alpha_1,\alpha_2,\cdots,\alpha_d)$ and cut out each group of $\{B_\alpha\}$ from $\Lambda$ at each step like Steps 1-4.
    As a result, the repetition number of the patching and merging protocol is equal to $2^d-1$, which is a constant independent of $N$, $t$, or $\varepsilon$.
    
\end{enumerate}

After the above steps, the time evolution operator for the whole system $e^{\mcl{L}\tau}$ is approximated by those for the blocks whose sizes are at most $\order{(R_\mr{p})^d}$.
There are two types of forward time evolution operators.
One is for the size-$\order{(R_\mr{p})^d}$ blocks $\{ B_\alpha \}$.
The other type is for the size-$\order{(R_\mr{p})^{d-1}}$ blocks at the boundaries of $\{B_\alpha\}$, like $e^{ \mcl{L}_{\partial_{B_\alpha}(R) :\overline{\partial}_{B_\alpha}(R)} \tau }$ in Eq. (\ref{EqA_High:Merge_op_high_dim_expansion}).
All the backward time evolutions are absorbed in the merged blocks.
The algorithm is constructed in a similar manner to Algorithm \ref{Algorithm_Generic}.
We employ the LCU-based approach for Lindbladian dynamics \cite{Li-Wang-2022-open} to implement the time evolution operators for the block Lindbladians $\{ \mcl{L}_{B_\alpha} \}$.
We use the same approach for the Lindbladians composed of the boundary terms like $\mcl{L}_{\partial_{B_\alpha}(R) :\overline{\partial}_{B_\alpha}(R)}$ [See Eq. (\ref{EqA_High:Merge_op_high_dim_expansion})] in the merged operators.
We execute quasi-probabilistic sampling for reproducing the non-CP parts in the merged operators.
The cost of the algorithm is given by the following theorem.

\begin{figure*}
    \centering
    \includegraphics[width=0.95\linewidth]{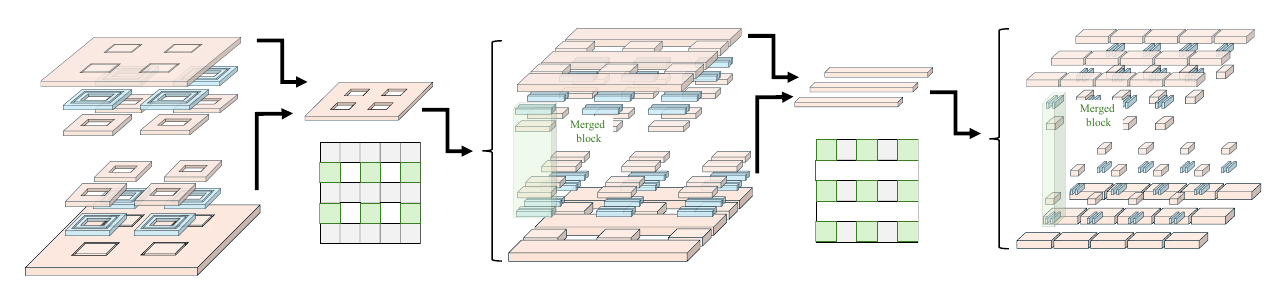}
    \caption{Repetition of patching and merging for high-dimensional systems. The left panel shows the decomposition by Steps 1 and 2. The central panel shows the one by Steps 3 and 4. We repeat the decomposition until all the blocks become as large as or smaller than $\order{(R_\mr{p})^d}$ like the right panel.}
    \label{Fig_2d_repetition}
\end{figure*}

\begin{theorem}\label{ThmA_High:generic}
\textbf{}

Let $\mcl{L}$ be a lattice Lindbladian composed of finite-ranged interactions and dissipation in $d \in \order{1}$ dimension with the system size $N$.
There exists a quantum algorithm that outputs the time-evolved observable $\mr{Tr}[O e^{\mcl{L}t} \rho]$ within an error $\varepsilon$, running with the following cost:
\begin{itemize}
    \item The number of $\order{1}$-qubit gates per sample:
    \begin{equation}
        \order{(Nt)^{\frac32- \frac1{4d+2}} \, \polylog{Nt/\varepsilon}}.
    \end{equation}
    
    \item Sampling complexity: $\Theta (\varepsilon^{-2})$.

    \item The number of ancilla qubits and the circuit depth:
    It requires $\Theta (\polylog{Nt/\varepsilon})$ ancilla qubits, and then the circuit depth amounts to $ \order{(Nt)^{3/2- 1/(4d+2)} \, \polylog{Nt/\varepsilon}}$.
    When we can use $\tilde{\Theta} \left(N^{(d+1)/(2d+1)}\right)$ ancilla qubits, the circuit depth can be as large as
    \begin{equation}
       \order{t (Nt)^{\frac{2d}{2d+1}} \, \polylog{Nt/\varepsilon}}.
    \end{equation}
\end{itemize}

\end{theorem}

Before proving the above theorem, we consider the counterparts of Theorems \ref{Thm:Merge_expansion} and \ref{Thm:Merge_truncation} in high-dimensional systems, which ensure the efficient implementation of quasi-probabilistic sampling.

\begin{corollary}\label{CorA_High:Merge}
\textbf{}

Suppose that the time $\tau$ is small enough to satisfy Eq. (\ref{EqA_High:time_assumption}) and 
\begin{equation}\label{EqA_High:time_assumption_corollary}
    \mr{Const.} \times (g\tau)^3 (R_\mr{p})^{d-1} \leq 1,
\end{equation}
where the constant is independent of $N$, $t$, or $\varepsilon$.
Then, the merged operator defined by Eq. (\ref{EqA_High:Merge_high_Step1}) is written in the form of Eq. (\ref{EqA_High:Merge_op_high_dim_expansion}) with an HP map $\mcl{A}_\mr{Merge}^\alpha(\tau)$ bounded by $\norm{\mcl{A}_\mr{Merge}^\alpha(\tau)}_\text{Pauli} \in \order{(g\tau)^3 (R_\mr{p})^{d-1}}$.
In addition, for $\epsilon \in (0,1)$, there exists an $\order{[\log(1/\epsilon)]^2}$-local HP map $\tilde{\mcl{A}}_\mr{Merge}^\alpha(\tau)$ such that the relations,
\begin{equation}
    \norm{\mcl{U}_\mr{Merge}^\alpha(\tau) - e^{ \mcl{L}_{\partial_{B_\alpha}(R) :\overline{\partial}_{B_\alpha}(R)} \tau } \left( 1 + \tilde{\mcl{A}}_\mr{Merge}^\alpha(\tau) \right)}_\Diamond \leq \epsilon, \quad \norm{\tilde{\mcl{A}}_\mr{Merge}^\alpha(\tau)}_\text{Pauli} \in \order{(g\tau)^3 (R_\mr{p})^{d-1}},
\end{equation}
are satisfied.

\end{corollary}

\textbf{Proof.---}
The proof is essentially the same as those for Theorems \ref{Thm:Merge_expansion} and \ref{Thm:Merge_truncation} in Appendix \ref{SecA:Merged_block}.
To be precise, when we prove the upper bound on the map $\mcl{A}_\mr{Merge}^\alpha(\tau)$ defined by Eq. (\ref{EqA_High:Merge_op_high_dim_expansion}), we replace $|\supp (\mcl{L}_{B:B'})|$ in Eq. (\ref{EqA_Merge:Delta_Merge_bound}) by $|\supp (\mcl{L}_{\partial_{B_\alpha}(R) :\overline{\partial}_{B_\alpha}(R)})| \in \order{\xi (R_\mr{p})^{d-1}}$.
The upper bound on the quantity corresponding to $\sup_{\tau'\in [0,\tau]} (\norm{\Delta_\mr{Merge}(\tau')}_\mr{Pauli}) \tau$ becomes $\order{(g\tau)^3 (R_\mr{p})^{d-1}}$ instead of Eq. (\ref{EqA_Merge:Delta_Merge_bound}).
When this quantity is smaller than $1$ by the assumption Eq. (\ref{EqA_High:time_assumption_corollary}), we can obtain the upper bound, $\norm{\mcl{A}_\mr{Merge}^\alpha(\tau)}_\text{Pauli} \in \order{(g\tau)^3 (R_\mr{p})^{d-1}}$, which completes the proof of the counterpart of Theorem \ref{Thm:Merge_expansion}.
The latter part corresponding to Theorem \ref{Thm:Merge_truncation} is obtained in a similar manner. $\quad \square$

The above corollary indicates that the merged block can be reproduced by quasi-probabilistic sampling with the sampling overhead as large as
\begin{equation}\label{EqA_High:Overhead_each_block}
    \left( 1+ \order{(g\tau)^3 (R_\mr{p})^{d-1}} \right)^2 \subset e^{\order{(g\tau)^3(R_\mr{p})^{d-1}}}.
\end{equation}
In addition, it is expanded by $\order{[\log (1/\epsilon)]^2}$-local Pauli operators, and each sampled circuit has a gate count of $\order{[\log (1/\epsilon)]^2}$.
Identifying the probability distribution of the sampling and the sampled quantum circuits can be executed in the same way as the one-dimensional case.
Namely, we can efficiently calculate them by classical computation with truncating the Dyson series, as shown in Appendix \ref{SubsecA:Merge_implement}.
We note that the same statement clearly applies to all the merged operators in the above steps like Eq. (\ref{EqA_High:Merge_op_2nd}). 
Finally, using this fact, we prove Theorem \ref{ThmA_High:generic} as follows.

\textbf{Proof of Theorem \ref{ThmA_High:generic}.---}
The proof is essentially the same as the one for Theorem \ref{Thm_Sparse:generic}.
We split the time $t$ into $r_t$ parts with $\tau=t/r_t$.
We set the block size $R \in \Theta (\log (Nt/\varepsilon))$ so that every approximation error by the patching lemma (i.e., Corollary \ref{CorA_High:Patching_strategy}) can be bounded by $\order{\varepsilon/(Nr_t)}$, which is small enough to achieve the error $\order{\varepsilon}$ in total.
First, we consider the sampling overhead for implementing the quasi-probabilistic sampling of the non-CP parts involved in the merged operators like Eqs. (\ref{EqA_High:Merge_op_high_dim_expansion}) and (\ref{EqA_High:Merge_op_2nd}).
Since there are at most $\order{N/(R_\mr{p})^d}$ copies of the merged operator at each time step, the total sampling overhead is as large as
\begin{eqnarray}
    && \left( e^{\order{(g\tau)^3(R_\mr{p})^{d-1}}}\right)^{Nr_t /(R_\mr{p})^d} = \exp \left[ \order{\frac{Nt^3}{(r_t)^2 R_p}}\right],
\end{eqnarray}
where we use the overhead for each merged block by Eq. (\ref{EqA_High:Overhead_each_block}).
It is sufficient to choose the number $r_t$ by
\begin{equation}\label{EqA_High:r_t_choice}
    r_t = \left\lceil t\sqrt{\frac{Nt}{R_\mr{p}}}\, \right\rceil
\end{equation}
to suppress the sampling overhead up to $\order{1}$.
This leads to the sampling complexity $\Theta (\varepsilon^{-2})$ for estimating the observable within an error $\varepsilon$ with $\Theta(1)$ probability.

We next evaluate the gate count per sample.
In the decomposed lattice, there are $\Theta (N/(R_\mr{p})^d)$ copies of $d$-dimensional blocks with the system size $\order{(R_\mr{p})^d}$ and $(d-1)$-dimensional boundaries with the system size $\order{(R_\mr{p})^{d-1}}$.
We run the LCU-based approach for implementing the forward time evolutions of these blocks and the boundaries.
The time evolution of the size-$\order{ (R_\mr{p})^d}$ can be implemented by $\Otilde{\max \left[(R_\mr{p})^d,(R_\mr{p})^{2d} \tau\right]}$ gates, which comes from Eq. (\ref{Eq_Setup:Gate_LCU}) with considering the minimal gate count for the short time $\tau$ like Eq. (\ref{Eq_Gen:dominant_blocks}).
The time evolution of the size-$\order{(R_\mr{p})^{d-1}}$ boundary block trivially has cheaper cost.
We also run the quasi-probabilistic sampling for reproducing the merged blocks.
For each merged block, we use $\order{[\log (Nt/\varepsilon)]^2}$ quantum gates as we discuss in Corollary \ref{CorA_High:Merge}.
This cost is cheaper than that of the LCU-based approach as well as the one-dimensional case in the main text.
Since the system contains at most $\order{N/(R_\mr{p})^d}$ blocks and the above implementation is repeated $r_t$ times, the gate count in total amounts to
\begin{equation}\label{EqA_High:Gate_count_derivation}
    \Otilde{r_t \times \frac{N}{(R_\mr{p})^d} \times \max \left[(R_\mr{p})^d,(R_\mr{p})^{2d} \tau\right]} \subset \Otilde{Nt \left[ \sqrt{\frac{Nt}{R_\mr{p}}} + (R_\mr{p})^d \right]},
\end{equation}
like Eq. (\ref{Eq_Gen:Gate_count_derivation}).
We set the block size $R_\mr{p}$ by
\begin{equation}\label{Eq_High:block_choice}
    R_\mr{p} \in \Theta \left( (Nt)^{\frac{1}{2d+1}}\right),
\end{equation}
which minimizes the scaling, Eq. (\ref{EqA_High:Gate_count_derivation}).

Under the choice of $r_t$ and $R_\mr{p}$ respectively by Eqs. (\ref{EqA_High:r_t_choice}) and (\ref{Eq_High:block_choice}), the time $\tau = t/r_t$ is as large as
\begin{equation}
    \tau \in \Otilde{\sqrt{\frac{R_\mr{p}}{Nt}}} \subset \Otilde{(Nt)^{-\frac{d}{2d+1}}} \subset \order{1}.
\end{equation}
In addition, the left-hand side of Eq. (\ref{EqA_High:time_assumption_corollary}) scales as $(g\tau)^3 (R_\mr{p})^{d-1} \in \order{(Nt)^{-1}}$.
As a result, the assumptions on the time $\tau$, i.e., Eqs. (\ref{EqA_High:time_assumption}) and (\ref{EqA_High:time_assumption_corollary}), can be satisfied by properly choosing the constant in Eq. (\ref{Eq_High:block_choice}).
This ensures the validity of the above analysis based on Corollary \ref{CorA_High:Patching_strategy} and Corollary \ref{CorA_High:Merge}.
Finally, substituting the expression of $R_\mr{p}$ by Eq. (\ref{Eq_High:block_choice}) into Eq. (\ref{EqA_High:Gate_count_derivation}), we obtain a gate count of $\Otilde{(Nt)^{3/2-1/(4d+2)}}$.
The relation between the number of ancilla qubits and the circuit depth is obtained in a similar manner to Theorem \ref{Thm_Sparse:generic}, depending on parallelization.
This completes the proof of Theorem \ref{ThmA_High:generic}. $\quad \square$

Extrapolation of the second-order PF \cite{Wang-2026-open} achieves a gate count of $\Otilde{(Nt)^{3/2}}$ for simulating time-evolved observables, which has the best dependence on the system size $N$ among the previous algorithms (See Table \ref{Table:Gate_counts}).
The gate count $\Otilde{(Nt)^{3/2-1/(4d+2)}}$ is better than its cost in any dimension $d \in \order{1}$, while keeping the poly-logarithmic dependency in $1/\varepsilon$.

\end{document}